\documentclass[12pt,oneside,letterpaper]{article}
\usepackage{amsmath}
\usepackage{amsfonts}
\usepackage{amssymb}
\usepackage{amsthm}
\usepackage{graphicx}
\usepackage{multirow}
\usepackage{color}
\usepackage{subfig}
\usepackage{subcaption}
\usepackage{rotfloat}
\usepackage[T1]{fontenc}
\usepackage{amssymb,amsmath}
\usepackage[left=1.25in, bottom=1.25in, right=1.25in, top=1.25in]{geometry}
\usepackage{times}
\usepackage{natbib}
\usepackage{cases}
\usepackage{inputenc}
\bibpunct{(}{)}{;}{a}{,}{,}
\usepackage{float}
\usepackage{pdfsync}
\usepackage{setspace}
\usepackage{enumerate}
\usepackage{mathtools}
\usepackage{tikz}
\usetikzlibrary{arrows.meta,positioning,calc}
\usepackage{pgfplots}
\usepgfplotslibrary{groupplots}
\pgfplotsset{compat=1.18}
\numberwithin{equation}{section}
\newtheorem{lemma}{Lemma}[section]
\newtheorem{proposition}{Proposition}[section]
\newtheorem{theorem}{Theorem}[section]
\newtheorem{assumption}{Assumption}
\newtheorem{corollary}{Corollary}[section]
\theoremstyle{definition}\newtheorem{example}{Example}
\theoremstyle{definition}\newtheorem{remark}{Remark}[section]
\DeclareMathOperator*{\argmin}{arg\,min}

\DeclareMathOperator*{\argmax}{arg\,max}

\title{Ordinal Distributional Change and Conservative Transition Benchmarks: Measurement, Identification, and Inference \footnote{I thank participants of the 2026 Bristol Econometrics Study Group meeting, Russell Davidson, Gaurav Datt, and Otavio Bartalotti for helpful discussions and comments. This research was supported by Monash eResearch capabilities, including M3.}}
\author{Rami V. Tabri\footnote{Address: Department of Econometrics and Business Statistics, Monash University, Clayton, Victoria, Australia, 3800. Email: rami.tabri@monash.edu.}}
\begin{document}
\maketitle

\begingroup
\setstretch{1.0}

\begin{abstract}
Repeated cross-sections reveal changes in ordinal distributions but not the
transitions producing them. I axiomatically characterize a threshold-weighted
probability metric for ordinal change from threshold-crossing principles. For
any threshold-additive ordinal geometry, the discrepancy coincides with the
Wasserstein--1 distance induced by that ground metric and measures minimum
displacement; its optimizing plans define conservative transition benchmarks.
With missing outcomes, I derive sharp identified sets for the discrepancy and
endpoint-conditioned benchmark plans. I develop finite-sample-valid projection
inference using randomized Monte Carlo calibration and global search with an
almost-sure convergence guarantee. Applied to Arab Barometer data, the framework
documents a robust shift toward broader and more regular remittance receipt in
Lebanon. The discrepancy interval remains well separated from zero after
allowing for item nonresponse and sampling uncertainty, while benchmark bounds
provide strong numerical evidence that least-displacement restructuring
excludes movement toward less frequent receipt and requires reassignment from
nonreceipt to recurrent receipt.

\end{abstract}

\vspace{-0.4em}


\smallskip


\endgroup

\newpage

\section{Introduction}

When panel data are unavailable, repeated cross-sections are a primary source
of information on distributional change. They show how population shares shift
across time or groups but not the individual transitions behind those changes.
This limitation is especially important for ordinal outcomes, where one may
want to know not only whether distributions differ, but also the extent and
nature of restructuring across ordered categories. Examples include
health status, educational attainment, well-being, institutional trust,
and remittance receipt frequency.

The lack of longitudinal data is a persistent empirical challenge. As
\cite{Deaton1985} notes, repeated cross-sections are often available even when
following the same units over time is impossible. Pseudo-panel and
synthetic-panel methods infer economic dynamics from repeated cross-sections
(e.g.,~\citealp{DangEtAl2014,DangLanjouw2023,HeraultJenkins2019}), but require
assumptions that link observations across periods or groups. Methods based on
category shares, cumulative distributions, or stochastic dominance instead
compare ordinal marginals directly (e.g.,~\citealp{Jenkins,Tabri2021}).
These methods can reveal whether and in what direction distributions differ,
but they do not measure the restructuring that the marginal change itself
requires.

This paper develops an intermediate form of information: instead of
reconstructing unobserved individual transitions, I ask what
\emph{minimum ordinal restructuring is required by observed marginal changes.}
This first calls for a measure of ordinal change that avoids imposing cardinal
distances, and for a characterization of transition structures that reconcile
the marginals with the least displacement measured by that discrepancy.
The remainder of the framework makes these population objects empirically
operational. When outcomes are missing, the latent marginals, discrepancy, and
benchmark plans are only partially identified; I derive their sharp identified
sets. Sampling uncertainty is then incorporated by finite-sample
projection-based inference, with derivative-free search and persistent global
exploration used to compute the resulting nonsmooth projections.

I begin with the measurement problem. The ordered categories define a natural
threshold-crossing geometry: moving between two categories requires crossing
the ordinal thresholds that separate them. I show that norm inducibility, threshold separability, and consistency with this category-level geometry characterize a threshold-weighted cumulative-distribution metric for ordinal distributions. When there is no substantive basis for weighting adjacent thresholds differently, threshold neutrality and unit normalization yield the canonical rank-based discrepancy used throughout the empirical analysis.

The measurement contribution is related to an axiomatic literature on ordinal
distributions in which cumulative probabilities provide natural coordinates for
ordered categorical data. \citet{BlairLacy1996,BlairLacy2000}, for example, develop
measures of ordinal variation using the cumulative distribution, while
\citet{Keska2017} studies transfers between neighboring categories and exploits
their effects on individual cumulative thresholds. These approaches concern
dispersion or inequality of a single distribution, typically relative to a
distinguished reference. The present problem instead concerns the magnitude of
change between two arbitrary ordinal marginals. The recurrence of
cumulative-$L_1$ forms in ordinal measurement is therefore not accidental:
adjacent-category movements isolate individual threshold coordinates, while
separability across distinct cuts makes their absolute contributions additive.
Here these principles determine how population imbalances across thresholds
aggregate into a probability metric for distributional change. Among existing
two-distribution comparisons, \citet{Cuhadaroglu2023} axiomatically characterizes the
Net Difference Index, which measures relative rank advantage across groups,
whereas the discrepancy studied here is symmetric and measures the magnitude of
change between the marginals.

There is also a close connection with the economic literature on mobility
measurement. \cite{Cowell1985} develops an axiomatic measure of distributional
change by taking the joint distribution of individuals' initial and final
positions as primitive and aggregating individual-level movements.
\cite{DAgostinoDardanoni2009} similarly study rank mobility from joint rank
information, characterizing mobility orderings and an index that, for global
rank transitions, is equivalent to Spearman's $\rho$. More recently, \cite{Cowell_Flachaire_2018}
develop an axiomatic and statistical framework for
mobility measurement. In these approaches the joint distribution, transition
structure, or matched observations provide the information from which mobility
is evaluated.

Recent work also emphasizes the substantive role played by that joint
information. \cite{McGee2026} studies when commonly used measures of
intergenerational mobility---including elasticities, rank-based measures,
transition probabilities, and axiomatic indices---give consistent mobility
comparisons. A central role is played by the concordance ordering of the
parent--child joint distribution, and his decomposition distinguishes changes
in dependence from changes in marginal inequality. The information problem
studied here is complementary: the dependence or transition structure itself
is unavailable. Only the two marginal distributions are observed, and the
question is what transition-relevant information those marginals alone
necessarily imply.

A notable exception is \cite{RayGenicot2023}, who axiomatically derive a
panel-independent measure of upward mobility. Their object is directional and
cardinal: the axioms reward income growth among individuals who are initially
worse off and ultimately yield a measure that can be implemented from the
initial and terminal income distributions without individual matching. The
present paper addresses a different measurement problem. It studies symmetric
change in ordinal distributions, using only category order and marginal
probabilities, and asks how much ordinal restructuring must have occurred in
the sense that no joint distribution with the observed marginals can involve
less displacement. Thus both approaches obtain transition-relevant measures
from marginal information, but the economic objects and the information
extracted from the marginals are different.

For ordinal outcomes, this leaves an intermediate empirical problem for the
practitioner. Only the ordinal marginals are observed, but the question is
transition-relevant: how much restructuring is necessarily implied by their
change, and how can that minimum restructuring be organized across category
pairs?

The paper answers these questions in two steps. The first is a measurement
problem: the ordered structure determines the relevant notion of displacement
without imposing cardinal distances. Once that geometry is fixed, the absence
of individual matching turns the second step into an optimization problem. The
observed marginals determine a Fr\'echet class of admissible joint
distributions, and classical marginal-problem methods obtain sharp bounds on
functionals of the joint distribution by optimizing over that class
\citep{Ruschendorf1991}. Here the relevant functional is aggregate ordinal
displacement. Minimizing it over the Fr\'echet class yields the least
restructuring compatible with the marginals. This is precisely an
optimal-transport problem \citep{villani2009optimal}: its minimum gives the scalar
discrepancy, while its optimizing plans describe the alternative ways in
which that minimum restructuring can be organized across category pairs. I
refer to these plans as \emph{conservative transition benchmarks}.

The characterization also connects the economic measurement problem to the
broader literature on probability metrics, which studies the structure and
problem-dependent choice of distances between probability distributions; see, for example, \citet{Zolotarev1983} and
\citet{RachevKlebanovStoyanovFabozzi2013}. Of particular
relevance is the Kantorovich--Rubinstein metric-extension problem. In the
general formulation, a ground metric is taken as given and the
Kantorovich--Rubinstein norm is characterized as the maximal norm compatible
with that metric; see~\citet{MellerayPetrovVershik2008}. The ordered setting considered
here supplies additional structure. Adjacent thresholds provide canonical
elementary directions of distributional change, and local consistency fixes
their costs. The triangle inequality then bounds every locally consistent
norm above by the weighted cumulative-distribution norm, while threshold
separability makes that bound bind. Thus the resulting probability metric has
the Kantorovich--Rubinstein maximality property, but maximality is not imposed
as a primitive requirement: it follows from the ordinal aggregation
principles. Compatibility with the ground metric alone is not sufficient for
this conclusion, as Section~\ref{subsec:ordinal-geometry} shows by example.

Once characterized, the discrepancy admits a classical optimal-transport
representation for every threshold-additive ordinal geometry. In particular,
the weighted cumulative-distribution metric coincides with the
Wasserstein--1 distance induced by the corresponding ground metric on the
ordered category space. Under threshold neutrality and unit normalization,
this general representation specializes to the Wasserstein--1 distance
between the canonical rank distributions. The role of optimal transport is
therefore interpretive rather than foundational: the transport criterion
emerges from the preceding measurement principles and gives the characterized
discrepancy its minimum-restructuring interpretation. Section~\ref{subsec:ordinal-geometry} gives the
formal representation and its historical antecedents.

The classical identity delivers the minimum value, but the present analysis
also retains the full set of optimizing plans. This set-valued perspective is
important because the marginals generally do not determine a unique
least-displacement reconstruction. Rather than selecting an arbitrary
optimizer, I preserve all conservative transition benchmarks compatible with
the available information. With missing outcomes, this nonuniqueness is
compounded because the observable data need not determine the marginal
distributions themselves.

Missing outcomes are pervasive in sample surveys, so making the framework
empirically operational requires allowing the ordinal outcome to be only
partially observed. The same minimum-restructuring objects then become set
identified. Under unrestricted missingness (e.g.,~\citealp{manski2005}), the
observable data determine sets of admissible latent marginals. I derive the
sharp identified interval for the discrepancy and characterize the
conservative benchmark plans associated with its endpoints. Because both the
endpoint-attaining marginals and their optimal transport plans may be
nonunique, I also derive sharp cellwise bounds over the endpoint-conditioned
benchmark plans. These bounds distinguish reallocations that are unavoidable
in every least-displacement reconstruction from those that can be avoided by
some such reconstruction.

For inference, I exploit the fact that these identified objects are determined
by a finite-dimensional observable-data distribution. I first construct an
empirical-likelihood confidence region for the observable probabilities and
then project it through the optimization problems defining the discrepancy and
benchmark objects. In the finite categorical model, the empirical likelihood
ratio coincides with the multinomial likelihood ratio, so randomized Monte
Carlo calibration \citep{DUFOUR2006} gives finite-sample-valid confidence
regions, including at the boundary of the probability simplex. Projection then
transfers this validity to the partially identified optimal values and
optimizer sets without requiring differentiability of the relevant mappings
or uniqueness of their optimizers.

The projection construction is related to the broader literature on inference
in partially identified models; see, among others,
\citet{Chernozhukov-Hong-Tamer}, \citet{RomanoShaikh2010}, and
\citet{KaidoMolinariStoye2019}. The present finite categorical setting permits
a different route: a finite-sample-valid confidence region for the observable
distribution is projected directly through the optimization problems defining
the identified objects. This avoids deriving an asymptotic law for the
generally nonsmooth value and solution mappings. Related asymptotic approaches
to nonsmooth functionals exploit directional differentiability
\citep{Fang-Santos-2018}; recent work applies this approach directly to values of
linear programs \citep{GoffMbakop2025}.

The remaining challenge is computational. Conditional on the Monte Carlo
draws, the inverted confidence region can be nonsmooth and disconnected.
I use derivative-free search with persistent global exploration of the
observable-parameter space. Under a Near-optimal accessibility condition, the
best feasible values generated by the search converge almost surely to the
exact projection extrema. This separates the statistical validity of the exact
projected confidence sets from the numerical problem of approximating their
extrema.

The empirical application uses repeated cross-sections from the
\emph{Arab Barometer} to study remittance receipt in Lebanon
before and during Lebanon's economic and financial crisis beginning 
in 2019. Although aggregate nominal remittance inflows remained
comparatively stable, remittances became increasingly important for household
support \citep{UNDP2023Remittances}. The data reveal a robust shift toward broader and more regular remittance
receipt. The projected discrepancy interval remains well separated from zero
after accounting for item nonresponse and sampling uncertainty. The projected
benchmark bounds further provide strong numerical evidence that minimum
restructuring can be organized through several movements toward more frequent
receipt, while excluding movement in the opposite direction and requiring some
reassignment from nonreceipt to recurrent receipt.

The paper therefore moves from observed marginal change to progressively richer
statements about what that change requires. The axiomatic analysis determines
how ordinal distributional change should be measured and, on a threshold-additive
geometry, yields the Kantorovich--Rubinstein norm, with maximality emerging from
the ordinal aggregation principles rather than being imposed as a primitive.
Optimal transport gives this metric its minimum-restructuring interpretation
and associated conservative transition benchmarks, while partial identification,
finite-sample inference, and computation make these objects operational when
outcomes are missing and the marginals must be learned from samples. Throughout,
the objective is not to reconstruct an unobserved panel, but to characterize
what repeated cross-sectional marginal changes themselves imply. More broadly,
the paper contributes to the growing use of optimal transport in econometric
methodology; see \citet{GalichonHenry2026} for a recent survey.

The remainder of the paper is organized as follows.
Section~\ref{Section Measuring Dist Change} develops the ordinal metric and
conservative transition benchmarks. Section~\ref{Section PI} derives sharp
identified sets under partial observation. Section~\ref{Section Inference}
develops finite-sample projection-based inference.
Section~\ref{Section Implementation} presents the computational implementation
and global-search guarantee. Section~\ref{Section Discussion} discusses
interpretation and limitations. Section~\ref{Section Empirical Section}
presents the empirical application, and Section~\ref{Section Conclusion}
concludes.Longer proofs of the main results are collected in Appendix~A; additional
regularity results, computational details, and extensions are provided in
Appendices~\ref{appendix:projection-measurability}-\ref{appendix:further-discussion}.

\section{Conservative Transition Benchmarks}\label{Section Measuring Dist Change}
\subsection{Ordinal Geometry and the Discrepancy}\label{subsec:ordinal-geometry}

I consider repeated cross-sections with outcomes taking values in a finite
ordered set
\(
\Xi=\{\xi_1\prec\xi_2\prec\cdots\prec\xi_K\},
\)
where only the marginal distributions are observed.  The indices
$1,\ldots,K$ record the ordering of the substantive categories; at this
stage they carry no cardinal interpretation and impose no metric spacing.

To measure displacement on the ordered category space, additional geometric
structure is required.  Let $d$ be a metric on $\Xi$.  I say that $d$ is
\emph{threshold additive} if, for every $i<j<\ell$,
\(
d(\xi_i,\xi_\ell)
=
d(\xi_i,\xi_j)+d(\xi_j,\xi_\ell).
\)
Threshold additivity formalizes the idea that moving between two ordered
categories consists of crossing the successive ordinal thresholds separating
them.

\begin{proposition}[Threshold-additive ordinal metrics]
\label{prop:threshold_metric}
Let $d$ be a metric on the ordered support
$\Xi=\{\xi_1\prec\cdots\prec\xi_K\}$, and set
\(
w_k:=d(\xi_k,\xi_{k+1})>0,
\) k=1,\ldots,K-1. Then $d$ is threshold additive if and only if
\begin{align}\label{eq - Threshhold add metric}
d(\xi_i,\xi_j)
=
\sum_{k=\min\{i,j\}}^{\max\{i,j\}-1}w_k,
\quad i\neq j.
\end{align}
\end{proposition}

\begin{proof}
See Appendix~\ref{Proof - prop -- threshold metric}.
\end{proof}

Proposition~\ref{prop:threshold_metric} shows that a threshold-additive
ordinal geometry is determined entirely by the costs assigned to crossing
the $K-1$ adjacent thresholds.

Every such geometry also admits a one-dimensional isometric representation.
Set
\(
a_1=0,
\;
a_i=\sum_{\ell=1}^{i-1} w_\ell,
\; i=2,\ldots,K,
\)
and define the function
\(
r_w:\Xi\to\mathbb{R},
\;
r_w(\xi_i)=a_i.
\)
Then, writing $d_w$ for the metric in~\eqref{eq - Threshhold add metric},
for every $i,j$,
\[
d_w(\xi_i,\xi_j)
=
|a_i-a_j|
=
|r_w(\xi_i)-r_w(\xi_j)|,
\]
so $r_w$ is an isometric embedding of $(\Xi,d_w)$ into
$(\mathbb{R},|\cdot|)$. Thus the locations $a_i$ are metric coordinates
representing the threshold-additive ground geometry.

This separates the measurement problem into two steps. The ordinal structure
first restricts the admissible ground geometry to a weighted path metric; I next
ask how that category-level geometry should extend to probability distributions
on the ordered support.

For notational convenience, identify a probability distribution $\mu$ on
$\Xi$ with its coordinate vector
\(
\bigl(\mu(\{\xi_1\}),\ldots,\mu(\{\xi_K\})\bigr)\in\Delta^K,
\)
where
\[
\Delta^K
\coloneqq
\left\{
\gamma\in\mathbb R^K:
\gamma_k\geq 0\ \text{for all }k,\;
\sum_{k=1}^K\gamma_k=1
\right\}
\]
denotes the probability simplex indexed by the ordered categories in $\Xi$.
To describe distributional changes, define
\(
\mathcal M_0^K
\coloneqq
\left\{
\sigma\in\mathbb R^K:
\sum_{i=1}^K\sigma_i=0
\right\},
\)
the vector space of signed distributional changes on $\Xi$, expressed in
these coordinates. In particular, if $\mu,\nu\in\Delta^K$, then
$\mu-\nu\in\mathcal M_0^K$.\footnote{If
$\sigma\in\mathcal M_0^K$ is nonzero, let
$m=\sum_i\sigma_i^+=\sum_i\sigma_i^-$. Then
$\sigma=m(\mu-\nu)$, where $\mu=\sigma^+/m$ and
$\nu=\sigma^-/m$ are probability distributions. Thus
$\mathcal M_0^K$ is the natural linear space generated by feasible
differences of probability distributions.}

The ordered structure of the support provides a natural basis for this
vector space. For $k=1,\ldots,K-1$, let
\(
b_k:=\delta_k-\delta_{k+1},
\)
where $\delta_k$ denotes the degenerate probability distribution concentrated
on category $\xi_k$. The vector $b_k$ represents an elementary signed
redistribution across the adjacent threshold separating $\xi_k$ and
$\xi_{k+1}$. Since
$\mathcal M_0^K$ has dimension $K-1$, and
$b_1,\ldots,b_{K-1}$ are linearly independent, they form a basis of
$\mathcal M_0^K$. Hence every $\sigma\in\mathcal M_0^K$ admits a unique
representation
\(
\sigma
=
\sum_{k=1}^{K-1} \beta_k(\sigma)b_k.
\)

The coordinates in this adjacent-threshold basis have a particularly simple
form. Comparing components in the preceding representation gives
\(
\beta_k(\sigma)
=
\sum_{i=1}^k\sigma_i,
\;
k=1,\ldots,K-1.
\)
I therefore write
\(
C_k(\sigma)
:=
\sum_{i=1}^k\sigma_i
\)
and refer to $C_k(\sigma)$ as the \emph{threshold coordinate} of $\sigma$
at threshold $k$. Thus
\(
\sigma
=
\sum_{k=1}^{K-1}C_k(\sigma)b_k.
\)
For a probability distribution $\mu\in\Delta^K$, define
\(
F_\mu(k)
:=
\mu(\{\xi_1,\ldots,\xi_k\})
=
\sum_{i=1}^k\mu(\{\xi_i\}),
\; k=1,\ldots,K.
\)
Then, for two probability distributions $\mu,\nu\in\Delta^K$,
\(
C_k(\mu-\nu)=F_\mu(k)-F_\nu(k),
\; k=1,\ldots,K-1.
\)
Hence the $k$th coordinate is the signed difference between the population
shares lying weakly below the $k$th ordinal threshold. The basis
decomposition therefore represents every distributional change uniquely by
its net elementary components across the successive thresholds of the
ordered outcome space. This decomposition motivates the aggregation
principle used below.

I restrict attention to probability metrics induced by norms on
$\mathcal M_0^K$. Given a norm $N$, define
\(
\Delta_N(\mu,\nu):=N(\mu-\nu).
\)
Norm inducibility has a useful population interpretation. For any
$\mu,\nu,\lambda\in\Delta^K$ and $t\in[0,1]$,
\(
\Delta_N\big((1-t)\lambda+t\mu,\,
             (1-t)\lambda+t\nu\big)
=
t\Delta_N(\mu,\nu).
\)
Thus a common population component contributes no discrepancy, while the
discrepancy generated by the remaining component scales proportionally with
its population share. In this sense, norm inducibility incorporates
cancellation of common population mass and linear replication of a given
distributional change.

The ordered structure supplies a second aggregation principle. For
$\sigma\in\mathcal M_0^K$, define its \emph{threshold support} by
\(
\mathcal{T}(\sigma)
:=
\left\{
k\in\{1,\ldots,K-1\}:C_k(\sigma)\neq0
\right\}.
\)
I say that $N$ satisfies \emph{threshold separability} if
\(
\mathcal{T}(\sigma)\cap \mathcal{T}(\tau)=\varnothing
\quad\Longrightarrow\quad
N(\sigma+\tau)=N(\sigma)+N(\tau).
\)
The restriction concerns distributional changes whose net effects operate
through disjoint ordinal cuts. Since such components involve distinct
elementary threshold directions in the decomposition above, threshold
separability requires their contributions to overall distributional change
to aggregate without interaction. It does not impose additivity when two
changes operate through the same threshold.

The final requirement connects the distribution-level discrepancy to the
category-level geometry. I impose \emph{local ground-metric consistency}:
\(
N(\delta_k-\delta_{k+1})=w_k,
\; k=1,\ldots,K-1.
\)
Thus the discrepancy associated with one unit of elementary movement across
threshold $k$ agrees with the cost $w_k$ assigned to crossing that threshold
in the underlying ordinal geometry. 

The second step is closely related to the metric-extension problem studied in
the broader theory of probability metrics; see, for example,
\citet{RachevKlebanovStoyanovFabozzi2013}. In the general
Kantorovich--Rubinstein framework, a ground metric is taken as given and one
asks how it extends to signed measures with zero total mass.
\citet{MellerayPetrovVershik2008} show that the Kantorovich--Rubinstein norm is
the maximal norm compatible with that given ground metric.

The characterization below takes a different, measurement-based route. On an
ordered support, the adjacent thresholds provide canonical elementary
directions of distributional change. Local ground-metric consistency fixes the
cost of each such direction, while threshold separability determines how
changes operating through distinct ordinal thresholds aggregate. Conditional
on the threshold weights, these requirements determine the norm completely,
without imposing maximality or beginning from a transport problem.

\begin{theorem}[Threshold-separable ordinal probability metrics]
\label{thm:ordinal_metric}
Let $d_w$ be a threshold-additive metric on $\Xi$ with threshold weights
$w_1,\ldots,w_{K-1}$ as in Proposition~\ref{prop:threshold_metric}. Let $N$
be a norm on $\mathcal M_0^K$ satisfying threshold separability and the local
ground-metric consistency conditions. Then, for every $\sigma\in\mathcal M_0^K$,
\(
N(\sigma)
=
\sum_{k=1}^{K-1}w_k|C_k(\sigma)|.
\)
Consequently, the induced probability metric is
\(
\Delta_N(\mu,\nu)
=
\sum_{k=1}^{K-1}
w_k|F_\mu(k)-F_\nu(k)|.
\)
Moreover, $\Delta_N$ extends the ground metric $d_w$ in the sense that
\(
\Delta_N(\delta_i,\delta_j)=d_w(\xi_i,\xi_j)
\)
for every $i,j$.

Conversely, the weighted cumulative-distribution metric above is induced by
a norm on $\mathcal M_0^K$, satisfies threshold separability, and is locally
consistent with $d_w$.
\end{theorem}

\begin{proof}
See Appendix~\ref{Proof - Thm - ordinal-metric charac}.
\end{proof}

Theorem~\ref{thm:ordinal_metric} separates the choice of coordinates from the
rule used to aggregate them. The threshold coordinates are dictated by the
ordered category space: they record the net population imbalance across each
ordinal cut and give the unique decomposition of a signed distributional
change into elementary adjacent-threshold movements. Norm inducibility
determines how discrepancy behaves under common population components and
population scaling, while threshold separability determines how changes
operating through distinct ordinal cuts aggregate. Local ground-metric
consistency then fixes the cost of each elementary threshold direction.
The weighted cumulative-distribution form is therefore a consequence of
these measurement principles rather than an assumed functional form. Thus Theorem~\ref{thm:ordinal_metric}
 solves the distributional extension problem once a
threshold-additive ordinal geometry is specified. The remaining measurement
question concerns the threshold weights themselves.

\begin{remark}[Threshold separability and Kantorovich--Rubinstein maximality]
\label{rem:kr-maximality}
The relationship between Theorem~\ref{thm:ordinal_metric}
and the maximality characterization of the Kantorovich--Rubinstein norm can be
made explicit. Let $\widetilde N$ be any norm on $\mathcal M_0^K$ satisfying
local ground-metric consistency,
\(
    \widetilde N(b_k)=w_k,
    \; k=1,\ldots,K-1.
\)
Since every $\sigma\in\mathcal M_0^K$ has the unique decomposition
\(
    \sigma
    =
    \sum_{k=1}^{K-1} C_k(\sigma)b_k,
\)
positive homogeneity and the triangle inequality imply
\begin{align}\label{eq:kr-upper-bound}
    \widetilde N(\sigma)
    \leq
    \sum_{k=1}^{K-1}
    |C_k(\sigma)|\,\widetilde N(b_k)
    =
    \sum_{k=1}^{K-1}
    w_k |C_k(\sigma)|.
\end{align}

The triangle-inequality bound in \eqref{eq:kr-upper-bound} therefore shows
that the norm characterized in Theorem~\ref{thm:ordinal_metric} is the
pointwise largest norm satisfying local ground-metric consistency. Because
every norm compatible with the full ground metric $d_w$ necessarily
satisfies local ground-metric consistency, the same bound implies that the
characterized norm is maximal among all norms compatible with $d_w$,
recovering the usual Kantorovich--Rubinstein maximality property; see
\citet{MellerayPetrovVershik2008}.

Theorem~\ref{thm:ordinal_metric}, however, does not impose maximality as a
primitive requirement. Threshold separability is an internal aggregation
condition on a single norm: distributional changes operating through disjoint
ordinal thresholds must contribute additively. Applied to the canonical
threshold decomposition, it gives
\(
    N(\sigma)
    =
    \sum_{k=1}^{K-1}
    N\!\left(C_k(\sigma)b_k\right)
    =
    \sum_{k=1}^{K-1}
    w_k |C_k(\sigma)|,
\)
and hence makes the upper bound in \eqref{eq:kr-upper-bound} bind. Maximality is therefore
a consequence of the threshold aggregation principle and local consistency,
rather than an assumption used to select the norm. The two characterizations
select the same norm on a weighted path geometry, but they do so from
different primitive requirements: maximality compares a candidate norm with
all other compatible norms, whereas threshold separability specifies how
distinct ordinal components of a distributional change aggregate.
\end{remark}

\begin{example}[Compatibility does not imply threshold separability]
\label{ex:compatible-not-separable}
Consider three ordered categories
\(
    \xi_1 \prec \xi_2 \prec \xi_3
\)
with adjacent threshold costs $w_1,w_2>0$, so that
\(
    d_w(\xi_1,\xi_2)=w_1,\,
    d_w(\xi_2,\xi_3)=w_2,\,
    d_w(\xi_1,\xi_3)=w_1+w_2.
\)
Write any $\sigma\in\mathcal M_0^3$ uniquely as
$\sigma=c_1b_1+c_2b_2$, and define the following norm on $\mathcal M_0^3$:
\(
    N_0(\sigma)
    :=
    \max\left\{
        w_1|c_1|,
        \;w_2|c_2|,
        \;\bigl|w_1c_1+w_2c_2\bigr|
    \right\}.
\)
Moreover,
\(
    N_0(b_1)=w_1,\,
    N_0(b_2)=w_2,\,
    N_0(b_1+b_2)=w_1+w_2.
\)
Since
$b_1=\delta_1-\delta_2$,
$b_2=\delta_2-\delta_3$, and
$b_1+b_2=\delta_1-\delta_3$,
$N_0$ agrees with $d_w$ on every pair of point masses. It is therefore
compatible with the full ground metric.

Nevertheless, $b_1$ and $-b_2$ have disjoint threshold supports, while
\[
    N_0(b_1-b_2)
    =
    \max\{w_1,w_2\}
    <
    w_1+w_2
    =
    N_0(b_1)+N_0(-b_2).
\]
Hence $N_0$ violates threshold separability. Compatibility with the ordinal
ground metric therefore does not by itself impose the aggregation principle
used in Theorem~\ref{thm:ordinal_metric}.

For example, taking
\(
    \mu=\left(\frac12,0,\frac12\right),
    \;
    \nu=(0,1,0),
\)
gives
\(
    \mu-\nu=\frac12(b_1-b_2).
\)
Consequently,
\(
    N_0(\mu-\nu)
    =
    \frac12\max\{w_1,w_2\},
\)
whereas the threshold-separable discrepancy is
\(
    D_w(\mu,\nu)
    =
    \frac12(w_1+w_2).
\)
Thus agreement on all category-to-category distances does not determine 
how population-level distributional changes are measured; threshold 
separability supplies the additional aggregation principle.
\end{example}

Writing
\(
D_w(\mu,\nu)
:=
\sum_{k=1}^{K-1}
w_k|F_\mu(k)-F_\nu(k)|,
\)
Theorem~\ref{thm:ordinal_metric} leaves the threshold weights themselves
unrestricted. This is important: ordinality alone does not determine their
relative magnitudes. If substantive information is available that justifies
a cardinal geometry on the ordered categories, that information can be used
to determine the threshold weights.

For example, suppose the categories are associated with externally justified
cardinal locations
\(
c_1<\cdots<c_K.
\)
Setting
\(
w_k=c_{k+1}-c_k,\; k=1,\ldots,K-1,
\)
gives
\(
d_w(\xi_i,\xi_j)=|c_i-c_j|.
\)
The metric coordinates introduced above then satisfy
\(
a_i=c_i-c_1,
\)
so the isometric representation $r_w(\xi_i)=a_i$ reproduces the externally
specified cardinal geometry up to an irrelevant translation. Thus the
canonical equal-threshold geometry considered below is the special case
appropriate when the available information provides only the ordering of the
categories and no basis for distinguishing the costs of adjacent threshold
crossings.

\begin{example}[Educational attainment with external cardinal spacing]
Suppose the ordered categories are
\[
\text{no schooling}
\prec
\text{primary}
\prec
\text{lower secondary}
\prec
\text{upper secondary}
\prec
\text{bachelor's},
\]
and, in the institutional setting under study, these attainment levels correspond
to the externally given schooling durations
\(
c=(0,6,9,12,16).
\)
The adjacent threshold costs are then
\(
w=(6,3,3,4),
\)
so that
\(
d_w(\xi_i,\xi_j)=|c_i-c_j|.
\)
Theorem~\ref{thm:ordinal_metric} yields
\[
D_w(\mu,\nu)
=
6|F_\mu(1)-F_\nu(1)|
+3|F_\mu(2)-F_\nu(2)|
+3|F_\mu(3)-F_\nu(3)|
+4|F_\mu(4)-F_\nu(4)|.
\]
Here the unequal threshold weights are not inferred from the ordinal labels:
they are supplied by external cardinal information about the outcome space.
When no such information is available, the equal-threshold geometry developed
below provides the canonical purely ordinal specialization.
\end{example}

When no such additional information is available, I impose a separate
neutrality requirement. I say that the ordinal geometry satisfies
\emph{threshold neutrality} if
\(w_k
=w_\ell
\)
for every $k,\ell$. Threshold neutrality does not assert that successive
categories are equally separated on an underlying cardinal scale. Rather,
it treats each observed ordinal threshold symmetrically when the available
ordinal information provides no basis for assigning different costs to
different crossings.

Normalizing the common threshold cost to one yields the canonical ordinal
geometry
\(
d^{\mathrm{can}}(\xi_i,\xi_j)=|i-j|.
\)
Theorem~\ref{thm:ordinal_metric} then immediately gives the canonical specialization appropriate
when no external cardinal information distinguishes the adjacent thresholds.


\begin{corollary}[Canonical ordinal discrepancy]
\label{cor:canonical_discrepancy}
Under threshold neutrality and unit normalization, the unique norm-induced
probability metric that is threshold separable and locally consistent with
the canonical ordinal geometry is
\(
D(\mu,\nu)
:=
\sum_{k=1}^{K-1}|F_\mu(k)-F_\nu(k)|.
\)
\end{corollary}

\begin{proof}
Under threshold neutrality, $w_k=1$ for every $k$, and the result follows
immediately from Theorem~\ref{thm:ordinal_metric}.
\end{proof}

The canonical discrepancy is therefore not obtained by assigning cardinal
values to the substantive categories. It is selected by the canonical
ordinal geometry together with norm inducibility, threshold separability,
and local consistency. Equivalently, it measures the size of a signed
distributional change by adding the absolute net population imbalances
across the observed ordinal thresholds.


The construction is also invariant to order-preserving relabellings that
preserve the ordinal geometry. Let $(\Xi,d_w)$ and $(\Xi',d'_w)$ be finite
ordered category spaces, and let
$h:\Xi\rightarrow\Xi'$ be an order-preserving bijection satisfying
\(
d'_w\!\left(h(\xi_i),h(\xi_j)\right)
=
d_w(\xi_i,\xi_j),
\; i,j=1,\ldots,K.
\)
Equivalently, the corresponding adjacent-threshold weights agree:
$w'_k=w_k$ for $k=1,\ldots,K-1$. If $\mu$ is a probability distribution on
$\Xi$, its pushforward under $h$, denoted by $h_\#\mu$, is the probability
distribution on $\Xi'$ defined by
\(
(h_\#\mu)(A)
=
\mu\!\left(h^{-1}(A)\right),
\; A\subseteq\Xi'.
\)
Thus $h_\#\mu$ is simply the distribution obtained by relabelling each
category $\xi_i$ as $h(\xi_i)$ while leaving its probability mass unchanged.

Let $r_w:\Xi\rightarrow\mathbb{R}$ and
$r'_w:\Xi'\rightarrow\mathbb{R}$ denote the one-dimensional isometric
representations induced by the two threshold-additive geometries, with the
same origin normalization. Since $h$ preserves both the ordering and the
geometry,
\(
r'_w\circ h=r_w.
\)
For any probability distribution $\mu$ on $\Xi$, the pushforward
$(r_w)_\#\mu$ is the probability distribution of the transformed random
variable $r_w(X)$ when $X\sim\mu$. Hence
\(
(r'_w)_\#(h_\#\mu)
=
(r'_w\circ h)_\#\mu
=
(r_w)_\#\mu,
\)
and similarly for $\nu$. It follows that
\(
D'_w(h_\#\mu,h_\#\nu)
=
D_w(\mu,\nu).
\)
Thus the discrepancy depends on the ordering and the specified
threshold-additive geometry, not on the labels assigned to the substantive
categories.

In the canonical threshold-neutral case, $w_k=1$ for all
$k=1,\ldots,K-1$, so every order-preserving bijection between ordered
category spaces with the same number of categories preserves the geometry.
Consequently, arbitrary strictly increasing numerical recodings of the
substantive categories have no effect on the canonical discrepancy.

\begin{example}[Credit ratings]
Consider the ordinal credit-rating scale
\[
\mathrm{CCC}\prec\mathrm{B}\prec\mathrm{BB}\prec
\mathrm{BBB}\prec\mathrm{A}\prec\mathrm{AA}\prec\mathrm{AAA}.
\]
The canonical rank map assigns the integers $1,\ldots,7$ while preserving
the ordering of the substantive categories. Under the canonical geometry,
moving from BBB to A crosses one observed ordinal threshold, whereas moving
from BB to A crosses two. These distances describe displacement on the
ordered category space; they do not assert that the underlying economic
differences between successive credit ratings are cardinally equal.

The same construction applies to verbally defined survey outcomes. If
\[
\mathrm{poor}\prec\mathrm{fair}\prec\mathrm{good}
\prec\mathrm{very\ good}\prec\mathrm{excellent}
\]
denotes self-reported health, the canonical rank representation again
records only the number of observed ordinal thresholds separating two
categories.
\end{example}

The characterization above selects the discrepancy from ordinal geometry
and distributional aggregation principles without appealing to any latent
transition structure. Optimal transport provides a separate representation
that gives the same metric an additional interpretation.

Because only the marginal distributions are observed, many joint
distributions, and hence many latent transition structures, are compatible
with them. A natural benchmark asks for the joint distribution that reconciles the marginals using the least aggregate displacement under the specified ground geometry.

\begin{proposition}[Transport representation]\label{prop:OT_representation}
Let $d_w$ be a threshold-additive metric on
$\Xi=\{\xi_1\prec\cdots\prec\xi_K\}$ with threshold weights
$w_1,\ldots,w_{K-1}$, and let
$r_w:\Xi\to\mathbb R$ denote the isometric representation defined above.
Then, for any probability distributions $\mu$ and $\nu$ on $\Xi$,
\(
D_w(\mu,\nu)
=
\min_{\pi\in\Pi(\mu,\nu)}
\sum_{i=1}^K\sum_{j=1}^K
d_w(\xi_i,\xi_j)\pi_{ij}.
\)
Under threshold neutrality and unit normalization,
$d_w(\xi_i,\xi_j)=|i-j|$ and $r_w(\xi_i)=i-1$.
Hence, up to an irrelevant common translation, the representation reduces
to the canonical rank representation $r(\xi_i)=i$, and
\(
D(\mu,\nu)
=
\min_{\pi\in\Pi(\mu,\nu)}
\sum_{i=1}^K\sum_{j=1}^K |i-j|\pi_{ij}.
\)
\end{proposition}

\begin{proof}
See Appendix~\ref{Proof - Prop}.
\end{proof}

The one-dimensional identity underlying Proposition~\ref{prop:OT_representation} has a
substantially earlier history than the modern Wasserstein terminology.~\cite{Gini1914} introduced an index of dissimilarity that, on the
real line with absolute-distance cost, coincides with the corresponding
minimum-transport functional.~\cite{Salvemini1943} obtained
the integrated-CDF representation for discrete distributions, and~\cite{DallAglio1956} extended the representation to general
distributions; see also \cite{RachevKlebanovStoyanovFabozzi2013} for
historical discussion.~\cite{Vallender} later gave the
identity in the modern probability-metric literature.

In the present setting, Proposition~\ref{prop:OT_representation} applies this
identity to any threshold-additive ground geometry characterized above.
Under threshold neutrality and unit normalization, it specializes to the
canonical rank distributions, showing that the canonical discrepancy
coincides with the Wasserstein--$1$ distance under the canonical ordinal
metric. Consequently, $D(\mu,\nu)$ is the minimum expected number of
ordinal thresholds that must be crossed to reconcile the two marginal
distributions.

\begin{remark}[Relation to the Gini mean difference]\label{Remark: Gini}
The connection with the Gini mean difference is distinct from the
historical identification of the Kantorovich metric with Gini's index
of dissimilarity. If $X\sim\mu$ and $Y\sim\nu$ are represented on their
canonical ordinal ranks, then
\(
D(\mu,\nu)
=
\min_{\pi\in\Pi(\mu,\nu)}
E_\pi|X-Y|.
\)
By contrast, the Gini mean difference evaluates the same
absolute-difference criterion under independence when the two marginals
coincide. Thus, if $X$ and $X'$ are independent draws from $\mu$,
\(
E|X-X'|
\)
is the Gini mean difference of the canonical rank distribution. More generally, evaluating absolute ordinal displacement under
different couplings provides a useful comparison between minimum,
independent, and maximal reassignment. Appendix~\ref{appendix:gini-dependence} develops this
comparison.
\end{remark}

The transport representation provides information beyond the scalar
discrepancy. Any minimizer
\(
\pi^\star
\in
\Pi_w^\star(\mu,\nu)
:=
\arg\min_{\pi\in\Pi(\mu,\nu)}
\sum_{i=1}^K\sum_{j=1}^K d_w(\xi_i,\xi_j)\pi_{ij}
\)
describes one way of organizing this minimum aggregate displacement across
the substantive categories. I interpret such minimizers as
\emph{conservative transition benchmarks}.

These benchmark plans are not estimates of realized individual
transitions, which are not identified from repeated cross-sections. Rather,
they are least-displacement reconstructions compatible with the given
marginals. For every admissible joint distribution
$\pi\in\Pi(\mu,\nu)$,
\(
\sum_{i=1}^K\sum_{j=1}^K d_w(\xi_i,\xi_j)\pi_{ij}
\geq
D_w(\mu,\nu).
\)
Thus any latent transition structure consistent with the marginals must
involve at least $D_w(\mu,\nu)$ units of aggregate ordinal displacement.
Optimal transport plans attain this lower bound and characterize the ways
in which the minimum required restructuring can be distributed across
ordered categories.
\subsection{Illustrative Example}\label{subsec:illustrative-example}

Consider the four-category distributions
\(
\mu=(0.4,0.3,0.2,0.1),
\;
\nu=(0.2,0.3,0.3,0.2).
\)
Their canonical ordinal discrepancy is
\(
D(\mu,\nu)
=
\sum_{k=1}^{3}|F_\mu(k)-F_\nu(k)|
=
0.5.
\)

Figure~\ref{fig:optimal_nonoptimal_couplings} displays two optimal transport plans attaining
this value. Both reconcile the same marginal distributions with the minimum
aggregate ordinal displacement, but they allocate the required reassignment
differently across category pairs. Thus the conservative transition benchmark
need not be unique even when the marginals, and hence the minimum amount of
restructuring, are point identified.
\begin{figure}[htbp]
\centering
\begin{tikzpicture}
\begin{groupplot}[
    group style={
        group size=2 by 1,
        horizontal sep=4cm
    },
    width=0.3\textwidth,
    height=0.3\textwidth,
    colormap={whitered}{
        color(0cm)=(white);
        color(1cm)=(red)
    },
    point meta min=0,
    point meta max=0.3,
    xmin=0.5, xmax=4.5,
    ymin=0.5, ymax=4.5,
    xtick={1,2,3,4},
    ytick={1,2,3,4},
    xlabel={Destination category $j$},
    ylabel={Origin category $i$},
    enlargelimits=false,
    axis on top,
    y dir=reverse,
    nodes near coords,
    nodes near coords style={font=\scriptsize, text=black},
    every node near coord/.append style={
        /pgf/number format/fixed,
        /pgf/number format/precision=1
    },
    every axis plot/.append style={draw=black}
]

\nextgroupplot[
    title={Optimal transport plan $\pi^*$},
    colorbar,
    colorbar style={
        at={(1.18,0.5)},
        anchor=west,
        height=0.26\textwidth,
        width=0.12cm,
        ytick={0,0.1,0.2,0.3},
        y dir=normal,
        nodes near coords=false,
        every node near coord/.append style={opacity=0},
    }
]
\addplot[
    matrix plot*,
    mesh/rows=4,
    mesh/cols=4,
    point meta=explicit,
] coordinates {
    (1,1) [0.2] (2,1) [0.2] (3,1) [0.0] (4,1) [0.0]
    (1,2) [0.0] (2,2) [0.1] (3,2) [0.2] (4,2) [0.0]
    (1,3) [0.0] (2,3) [0.0] (3,3) [0.1] (4,3) [0.1]
    (1,4) [0.0] (2,4) [0.0] (3,4) [0.0] (4,4) [0.1]
};

\nextgroupplot[
    title={Optimal transport plan $\tilde{\pi}$},
    colorbar,
    colorbar style={
        at={(1.18,0.5)},
        anchor=west,
        height=0.26\textwidth,
        width=0.12cm,
        ytick={0,0.1,0.2,0.3},
        scaled ticks=false,
        y dir=normal,
        nodes near coords=false,
        every node near coord/.append style={opacity=0},
    }
]
\addplot[
    matrix plot*,
    mesh/rows=4,
    mesh/cols=4,
    point meta=explicit,
] coordinates {
    (1,1) [0.2] (2,1) [0.0] (3,1) [0.2] (4,1) [0.0]
    (1,2) [0.0] (2,2) [0.3] (3,2) [0.0] (4,2) [0.0]
    (1,3) [0.0] (2,3) [0.0] (3,3) [0.1] (4,3) [0.1]
    (1,4) [0.0] (2,4) [0.0] (3,4) [0.0] (4,4) [0.1]
};

\end{groupplot}
\end{tikzpicture}
\caption{Two optimal transport plans with common marginals. Both attain
$D(\mu,\nu)=0.5$ but organize the minimal reassignment differently.}
\label{fig:optimal_nonoptimal_couplings}
\end{figure}

This distinction motivates the set-valued treatment of benchmark plans below.
Rather than selecting an arbitrary optimizer, the analysis retains all
least-displacement reconstructions compatible with the available marginal
information.

\section{Conservative Transition Benchmarks under Partial Information}
\label{Section PI}

I now allow the ordinal outcomes to be missing for some observations. Let
$X\sim\mu$ and $Y\sim\nu$ denote the latent ordinal outcomes in the two
populations, and let $Z_X$ and $Z_Y$ indicate whether they are observed.
The observable distributions are summarized by
\[
\mathbf q_X=(q_{X0},q_{X1},\ldots,q_{XK})\in\Delta^{K+1},
\qquad
\mathbf q_Y=(q_{Y0},q_{Y1},\ldots,q_{YK})\in\Delta^{K+1},
\]
where
\(
q_{X0}=\Pr(Z_X=0),
\;
q_{Xk}=\Pr(Z_X=1,X=\xi_k),
\)
and analogously for $\mathbf q_Y$.

Under unrestricted missingness, the latent marginal distributions are
partially identified. Writing
$\mathbf q_X^+=(q_{X1},\ldots,q_{XK})$, the sharp identified set for
$\mu$ is
\[
\mathcal M_\mu
=
\left\{
\mathbf q_X^+ + q_{X0}r:r\in\Delta^K
\right\}
=
\left\{
\gamma\in\Delta^K:
q_{Xk}\leq\gamma_k\leq q_{Xk}+q_{X0},
\; k=1,\ldots,K
\right\}.
\]
Similarly,
\(
\mathcal M_\nu
=
\left\{
\eta\in\Delta^K:
q_{Yk}\leq\eta_k\leq q_{Yk}+q_{Y0},
\; k=1,\ldots,K
\right\}
\)
is the sharp identified set for $\nu$. Because missingness is unrestricted separately in the two populations,
there are no cross-population restrictions on the allocation of the
unobserved mass. Hence the sharp joint identified set for
$(\mu,\nu)$ is
\(
\mathcal M_\mu\times\mathcal M_\nu.
\)
Thus, for example,
\(
\sum_{\ell=1}^k q_{X\ell}
\leq
F_\mu(k)
\leq
\sum_{\ell=1}^k q_{X\ell}+q_{X0},
\; k=1,\ldots,K-1,
\)
with analogous sharp bounds for $F_\nu$.
Since the discrepancy depends on the latent marginals, missingness also
partially identifies the minimum ordinal restructuring implied by the two
populations.

\begin{theorem}[Sharp identification of the weighted discrepancy]\label{thm:discrepancy_bounds}
Fix a threshold-additive ground geometry with weights
$w_1,\ldots,w_{K-1}>0$. The sharp identified set for
$D_w(\mu,\nu)$ is the interval
\(
[\underline D_w,\overline D_w],
\)
where both endpoints are attained with
\(
\underline D_w
=
\min_{\gamma\in\mathcal M_\mu,\,
      \eta\in\mathcal M_\nu}
D_w(\gamma,\eta),
\)
and
\(
\overline D_w
=
\max_{\gamma\in\mathcal M_\mu,\,
      \eta\in\mathcal M_\nu}
D_w(\gamma,\eta).
\)

The lower endpoint is the value of the linear program
\(
\underline D_w
=
\min_{\gamma,\eta,t}
\sum_{k=1}^{K-1} w_k t_k,
\)
subject to
\(
\gamma\in\mathcal M_\mu,\;
\eta\in\mathcal M_\nu,
\)
and
\(
-t_k
\leq
F_\gamma(k)-F_\eta(k)
\leq
t_k,
\;
k=1,\ldots,K-1.
\)

For the upper endpoint, let
\(
\mathcal S_K=\{-1,1\}^{K-1},
\)
and, for each $s\in\mathcal S_K$, define
\[
D_{w,s}
=
\max_{\gamma\in\mathcal M_\mu,\,
      \eta\in\mathcal M_\nu}
\sum_{k=1}^{K-1}
w_k s_k
\bigl(F_\gamma(k)-F_\eta(k)\bigr).
\]
Then
\(
\overline D_w
=
\max_{s\in\mathcal S_K}D_{w,s}.
\)
Hence $\underline D_w$ is the value of one linear program, while
$\overline D_w$ is the maximum of the values of $2^{K-1}$ linear
programs.
\end{theorem}
%

\begin{proof}
See Appendix~\ref{Proof - Thm PI}.
\end{proof}

Thus $[\underline D_w,\overline D_w]$ gives the sharp range of minimum ordinal
restructuring across all latent marginals compatible with the observed
data.

\subsection{Endpoint-Conditioned Transport Plans}

The discrepancy interval does not determine how the corresponding minimum
restructuring is allocated across category pairs. For
$(\gamma,\eta)\in\mathcal M_\mu\times\mathcal M_\nu$, let
\[
\Pi_w^*(\gamma,\eta)
:=
\arg\min_{\pi\in\Pi(\gamma,\eta)}
\sum_{i=1}^K\sum_{j=1}^K d_w(\xi_i,\xi_j)\pi_{ij}.
\]
The benchmark-plan sets associated with the two discrepancy endpoints are
\[
\Pi_{w,L}^*
:=
\bigcup_{\substack{
\gamma\in\mathcal M_\mu,\ \eta\in\mathcal M_\nu\\
D_w(\gamma,\eta)=\underline D_w
}}
\Pi_w^*(\gamma,\eta),
\qquad
\Pi_{w,U}^*
:=
\bigcup_{\substack{
\gamma\in\mathcal M_\mu,\ \eta\in\mathcal M_\nu\\
D_w(\gamma,\eta)=\overline D_w
}}
\Pi_w^*(\gamma,\eta).
\]
For $e\in\{L,U\}$ and each cell $(i,j)$, define
\(
\underline\pi^{\,e}_{w,ij}
:=
\inf_{\pi\in\Pi_{w,e}^*}\pi_{ij},
\;
\overline\pi^{\,e}_{w,ij}
:=
\sup_{\pi\in\Pi_{w,e}^*}\pi_{ij}.
\)


\begin{theorem}\label{thm - PI coupling endpoint bounds}
The endpoint-conditioned benchmark sets
$\Pi^\star_{w,L}$ and $\Pi^\star_{w,U}$ are nonempty and compact.
Consequently, for every $(i,j)\in\{1,\ldots,K\}^2$ and
$e\in\{L,U\}$, the endpoint-conditioned cell bounds are attained:
\(
\underline{\pi}^{\,e}_{w,ij}
=
\min_{\pi\in\Pi^\star_{w,e}}\pi_{ij},
\)
\(\overline{\pi}^{\,e}_{w,ij}
=
\max_{\pi\in\Pi^\star_{w,e}}\pi_{ij}.
\)
\end{theorem}

\begin{proof}
See Appendix~\ref{Proof -Thm 2}.
\end{proof}

The cell intervals retain the nonuniqueness of the conservative benchmark
plans. If $\underline\pi^{\,e}_{w,ij}>0$, every benchmark plan associated
with endpoint $e$ assigns positive mass to cell $(i,j)$; if
$\overline\pi^{\,e}_{w,ij}=0$, none does. Their widths measure the remaining
flexibility in how minimum restructuring can be allocated across category
pairs, while comparison of the lower- and upper-endpoint bounds reveals
its sensitivity to the missing-data uncertainty.

\section{Projection-Based Empirical Likelihood Inference}
\label{Section Inference}

Fix a threshold-additive ordinal geometry $d_w$ on $\Xi$, with known
threshold weights given by $w_1,\ldots,w_{K-1}>0$. The identified objects developed in Section~\ref{Section PI} are completely determined by the observable-data distribution. Consequently, sampling uncertainty enters the analysis through the finite-dimensional probabilities governing the observable outcomes and propagates to the parameters of interest through the optimization problems defining the discrepancy bounds and the endpoint-conditioned benchmark transition plans.

I exploit this structure using projection-based inference. I first construct the empirical likelihood ratio for the observable-data parameter and show that, in the present finite categorical setting, it coincides exactly with the ordinary multinomial likelihood ratio. This representation permits finite-sample Monte Carlo calibration of the empirical likelihood-ratio statistic. Under any fixed candidate value of the observable parameter, the multinomial sampling distribution is completely specified and can therefore be simulated exactly. Inverting the resulting Monte Carlo tests yields a confidence region with finite-sample coverage.

I combine this procedure with a projection principle. Any confidence region for the observable-data parameter can be propagated through the optimization problems defining the identified objects. Its coverage guarantee is then inherited by the resulting confidence regions for the discrepancy identified set and the endpoint-conditioned benchmark transition probabilities. This approach avoids requiring differentiability of the relevant value functions or uniqueness of the underlying optimizers.
\subsection{Observable-data parameterization}
\label{subsec:observable-parameterization}
Let
\(
\vartheta
=
(\mathbf q_X,\mathbf q_Y)
\in
\Theta
:=
\Delta^{K+1}\times\Delta^{K+1}
\)
denote the joint observable-data parameter. Each multinomial probability
vector has $K$ free components, so $\vartheta$ has dimension $2K$.
Let
\(
N_X=(N_{X0},\ldots,N_{XK}),
\;
N_Y=(N_{Y0},\ldots,N_{YK})
\)
denote the observable multinomial counts from independent samples of sizes
$n_X$ and $n_Y$, with empirical probability vectors
\(
\widehat{\mathbf q}_X=N_X/n_X,
\;
\widehat{\mathbf q}_Y=N_Y/n_Y.
\)

Throughout this section, let $n:=(n_X,n_Y)$ denote the sample-size index. Furthermore, let $\mathbb P_{\vartheta_0}^{n}$
denote the joint sampling distribution of the independent observable
samples $O_1^X,\ldots,O_{n_X}^X$ and $O_1^Y,\ldots,O_{n_Y}^Y$ under the true observable-data parameter $\vartheta_0\in\Theta$, where
\begin{align*}
O_i^X
=
\begin{cases}
* & \text{if } Z_{X_i}=0,\\
X_i & \text{if } Z_{X_i}=1,
\end{cases}
\qquad
\text{and}
\qquad
O_j^Y
=
\begin{cases}
* & \text{if } Z_{Y_j}=0,\\
Y_j & \text{if } Z_{Y_j}=1,
\end{cases}
\end{align*}
for $i=1,\ldots, n_X$ and $j=1,\ldots,n_Y$, with ``$*$'' denoting the missing-value code.

\subsection{Empirical likelihood and its multinomial representation}
\label{subsec:EL-multinomial}

Inference begins with the empirical likelihood for the observable-data parameter. In the present finite categorical setting, empirical likelihood has a particularly simple representation: its likelihood ratio coincides exactly with the ordinary multinomial likelihood ratio.

For the source sample, the likelihood associated with a candidate probability vector $\mathbf q_X$ is, up to a factor that does not depend on $\mathbf q_X$,
\(
L_X(\mathbf q_X)
=
\prod_{j=0}^K
q_{Xj}^{N_{Xj}},
\)
with the analogous expression
\(
L_Y(\mathbf q_Y)
=
\prod_{j=0}^K
q_{Yj}^{N_{Yj}}
\)
for the target sample. Independence of the two repeated cross-sections gives
\(
L(\vartheta)
=
L_X(\mathbf q_X)L_Y(\mathbf q_Y).
\)
Let
\(
\widehat{\vartheta}
=
(\widehat{\mathbf q}_X,\widehat{\mathbf q}_Y)
\)
denote the unrestricted maximum empirical likelihood estimator and define the empirical likelihood ratio by
\(
R_n(\vartheta)
=
L(\vartheta)/L(\widehat{\vartheta}).
\)

The following result gives the representation used for finite-sample Monte
Carlo inference.

\begin{proposition}[Empirical and multinomial likelihood ratios]
\label{Prop-ELR}
For every $\vartheta=(\mathbf q_X,\mathbf q_Y)\in\Theta$,
\[
-2\log R_n(\vartheta)
=
2n_X
D_{\mathrm{KL}}
\left(
\widehat{\mathbf q}_X
\,\Vert\,
\mathbf q_X
\right)
+
2n_Y
D_{\mathrm{KL}}
\left(
\widehat{\mathbf q}_Y
\,\Vert\,
\mathbf q_Y
\right),
\]
where
\(
D_{\mathrm{KL}}(a\Vert b)
=
\sum_{j=0}^K
a_j\log\left(\frac{a_j}{b_j}\right).
\)
The conventions $0\log(0/b_j)=0$ and
$D_{\mathrm{KL}}(a\Vert b)=+\infty$ whenever $a_j>0$ and $b_j=0$
are used.
\end{proposition}

\begin{proof}
See Appendix~\ref{Proof - Prop ELR}.
\end{proof}

For the remainder of the section, write
$T_n(\vartheta):=-2\log R_n(\vartheta)$. Before developing its finite-sample
Monte Carlo calibration, I state the projection principle that maps a
confidence region for the observable-data parameter into confidence regions
for the optimization-based objects of interest.

\subsection{Projection-based inference through optimization correspondences}
\label{subsec:projection-general}

Let $E$ be a separable metric space and let
\(
\mathcal Z:\Theta\rightrightarrows E
\)
assign to each observable parameter $\vartheta$ a nonempty feasible set $\mathcal Z(\vartheta)$. Let
\(
f:E\rightarrow\mathbb R
\)
be a continuous objective function and consider the parameterized optimization problem
\(
v(\vartheta)
=
\inf_{z\in\mathcal Z(\vartheta)}
f(z),
\)
with solution correspondence
\(
S(\vartheta)
=
\argmin_{z\in\mathcal Z(\vartheta)}
f(z).
\)

The optimization problem generates two types of inferential objects. The first is the value function $v(\vartheta)$, which arises when the parameter of interest is an optimal value. The second is the solution correspondence $S(\vartheta)$, which arises when the entire collection of optimal solutions is economically relevant.

Following \cite{AliprantisBorder2006}, a correspondence
\(
\mathcal Z:\Theta\rightrightarrows E
\)
is said to be \emph{weakly measurable} if, for every open set
$O\subseteq E$,
the inverse image
\(
\left\{
\vartheta\in\Theta:
\mathcal Z(\vartheta)\cap O\neq\varnothing
\right\}
\)
is measurable. I impose the following regularity condition.
\begin{assumption}
\label{ass:projection}
The correspondence
\(
\mathcal Z:\Theta\rightrightarrows E
\)
is weakly measurable with nonempty compact values, and
$f:E\rightarrow\mathbb R$ is continuous.
\end{assumption}

Under Assumption~\ref{ass:projection}, Proposition~\ref{prop:supp-measurable-maximum} establishes that $v$ is measurable and that $S$ is weakly measurable with nonempty compact values. Appendix~\ref{appendix:projection-measurability} verifies these properties for the latent-marginal, identification-endpoint, and endpoint-conditioned benchmark correspondences used in this paper.

Let $\mathcal C_n\subseteq\Theta$ be any confidence region for the observable-data parameter. Define its projection through the value function by
\(
\mathcal V_n(\mathcal C_n)
=
\left\{
v(\vartheta):
\vartheta\in\mathcal C_n
\right\},
\)
and its projection through the solution correspondence by
\(
\mathcal S_n(\mathcal C_n)
=
\bigcup_{\vartheta\in\mathcal C_n}
S(\vartheta).
\)

The following result is the basis of all subsequent inference.

\begin{theorem}[Projection principle]
\label{thm:projection}
Suppose Assumption~\ref{ass:projection} holds and let $\alpha^*\in(0,1)$. Then
\(
\left\{
\vartheta_0\in\mathcal C_n
\right\}
\subseteq
\left\{
v(\vartheta_0)
\in
\mathcal V_n(\mathcal C_n)
\right\},
\)
and
\(
\left\{
\vartheta_0\in\mathcal C_n
\right\}
\subseteq
\left\{
S(\vartheta_0)
\subseteq
\mathcal S_n(\mathcal C_n)
\right\}.
\)
Consequently, for any probability measure $\mathbb{P}$ governing the randomness of $\mathcal{C}_n$, if
\(
\mathbb{P}
\left(
\vartheta_0\in\mathcal C_n
\right)
\geq
1-\alpha^\star,
\)
then
\(
\mathbb{P}
\left(
v(\vartheta_0)
\in
\mathcal V_n(\mathcal C_n)
\right)
\geq
1-\alpha^\star
\)
and
\(
\mathbb{P}
\left(
S(\vartheta_0)
\subseteq
\mathcal S_n(\mathcal C_n)
\right)
\geq
1-\alpha^\star
\)

\end{theorem}

\begin{proof}
On the event $\{\vartheta_0\in\mathcal C_n\}$,
\(
v(\vartheta_0)
\in
\left\{
v(\vartheta):
\vartheta\in\mathcal C_n
\right\}
\)
and
\(
S(\vartheta_0)
\subseteq
\bigcup_{\vartheta\in\mathcal C_n}
S(\vartheta).
\)
The stated event inclusions and probability inequalities follow immediately.
\end{proof}

Theorem~4.1 is a finite-sample result and does not depend on the method used to
construct $\mathcal C_n$. An exact or conservative finite-sample confidence
region for $\vartheta_0$ therefore produces finite-sample-valid projected
inference. No differentiability of the value or solution mappings and no
uniqueness of the optimizers is required.

\subsection{Finite-sample Monte Carlo calibration}
\label{subsec:MC-inference}

The multinomial representation in Proposition~\ref{Prop-ELR} makes it possible to calibrate the empirical likelihood-ratio statistic without using an asymptotic approximation. For each candidate value
\(
\vartheta
=
(\mathbf q_X,\mathbf q_Y)
\in\Theta,
\)
consider the simple null hypothesis
\(
H_0:\vartheta=\vartheta_0.
\)
Under this null, the joint finite-sample distribution of the observable counts is completely specified.

For a fixed candidate $\vartheta$, generate $B$ independent Monte Carlo samples
\[
N_X^{(b)}
\sim
\operatorname{Multinomial}
(n_X,\mathbf q_X),
\qquad
N_Y^{(b)}
\sim
\operatorname{Multinomial}
(n_Y,\mathbf q_Y),
\qquad
b=1,\ldots,B,
\]
independently across $b$ and across the two repeated cross-sections. Let
$\widehat{\mathbf q}_X^{(b)}$ and
$\widehat{\mathbf q}_Y^{(b)}$ denote the corresponding simulated empirical probability vectors and define
\[
T_n^{(b)}(\vartheta)
=
2n_X
D_{\mathrm{KL}}
\left(
\widehat{\mathbf q}_X^{(b)}
\,\Vert\,
\mathbf q_X
\right)
+
2n_Y
D_{\mathrm{KL}}
\left(
\widehat{\mathbf q}_Y^{(b)}
\,\Vert\,
\mathbf q_Y
\right).
\]
Write
\(
T_n^{(0)}(\vartheta)
=
T_n(\vartheta)
\)
for the statistic computed from the observed samples.

Because the multinomial likelihood-ratio statistic is discrete, ties occur with positive probability. I therefore use randomized tie-breaking. Let
\(
U_0,U_1,\ldots,U_B
\)
be independent $\operatorname{Uniform}(0,1)$ random variables, independent of both the observed and simulated samples, and order the pairs
\(
\left(
T_n^{(b)}(\vartheta),U_b
\right),
\) for
$b=0,\ldots,B$,
lexicographically, with larger pairs regarded as more extreme. Equivalently, define
\[
I_b(\vartheta)
=
\mathbf 1
\left\{
T_n^{(b)}(\vartheta)>T_n^{(0)}(\vartheta)
\right\}
+
\mathbf 1
\left\{
T_n^{(b)}(\vartheta)=T_n^{(0)}(\vartheta)
\right\}
\mathbf 1
\left\{
U_b>U_0
\right\},
\]
and let
\[
p_{n,B}(\vartheta)
=
\frac{
1+\sum_{b=1}^B I_b(\vartheta)
}{
B+1
}.
\]
Thus small values of $p_{n,B}(\vartheta)$ correspond to unusually large values of the observed empirical likelihood-ratio statistic.

Under the simple null, the observed and simulated statistics have the same finite-sample distribution. Randomized tie-breaking makes their ranks exactly exchangeable.

\begin{theorem}[Finite-sample Monte Carlo test]
\label{thm:MC-test}
For every
$\vartheta_0\in\Theta$
and every finite pair of sample sizes
$(n_X,n_Y)$,
\(
\mathbb P_{\vartheta_0}^{n,B}
\left(
p_{n,B}(\vartheta_0)\leq\alpha
\right)
=
\frac{
\lfloor\alpha(B+1)\rfloor
}{
B+1
},
\)
where
$\mathbb P_{\vartheta_0}^{n,B}$
denotes probability with respect to both the observed samples and the auxiliary Monte Carlo randomization. Consequently, the test that rejects
$H_0:\vartheta=\vartheta_0$
when
$p_{n,B}(\vartheta_0)\leq\alpha$
has finite-sample level no greater than $\alpha$. If
$\alpha(B+1)$ is an integer, its rejection probability is exactly $\alpha$.
\end{theorem}

\begin{proof}
Under $H_0$, the $B+1$ pairs
\(
\left(
T_n^{(b)}(\vartheta_0),U_b
\right),
\)
for $b=0,\ldots,B$, are exchangeable. Since the $U_b$ are continuously distributed, the lexicographic ordering produces no ties. The rank of the observed pair is therefore uniformly distributed over
$\{1,\ldots,B+1\}$. The result follows immediately.
\end{proof}

Inverting the Monte Carlo tests gives the confidence region
\[
\mathcal C_{n,B}^{\mathrm{MC}}(1-\alpha)
:=
\left\{
\vartheta\in\Theta:
p_{n,B}(\vartheta)>\alpha
\right\}.
\]
Define
\(
\alpha_B
:=
\frac{
\lfloor\alpha(B+1)\rfloor
}{
B+1
}.
\)
Theorem~\ref{thm:MC-test} implies
\(
\mathbb P_{\vartheta_0}^{n,B}
\left(
\vartheta_0
\in
\mathcal C_{n,B}^{\mathrm{MC}}(1-\alpha)
\right)
=
1-\alpha_B
\geq
1-\alpha.
\)
If $\alpha(B+1)$ is an integer, the primitive Monte Carlo confidence region therefore has exact finite-sample coverage $1-\alpha$.

An important feature of Monte Carlo testing is that finite-sample validity
does not require the number of replications $B$ to be large. When the
nominal level lies on the Monte Carlo grid, the randomized test has the
stated finite-sample level for any such $B$. The choice of $B$ nevertheless
matters for other properties of the procedure. As emphasized by
\citet[Section~3]{DUFOUR2006}, the number of Monte Carlo replications affects
power and the probability that the Monte Carlo decision differs from that of
the corresponding infeasible fundamental test. Increasing $B$ can therefore
improve the power and concordance properties of the Monte Carlo test, at the
cost of greater computation.

The Monte Carlo test requires no restriction of $\vartheta_0$ to the relative
interior of $\Theta$. It remains valid at boundary points because, for every
fixed candidate $\vartheta$, the multinomial distribution under the null can
be simulated directly.

In the numerical implementation, common random numbers are used across candidate values of $\vartheta$ when the inverted confidence region is searched numerically. At the true candidate value, this construction preserves the exchangeability argument underlying Theorem~\ref{thm:MC-test}. Details of the simulation and optimization procedure are given in Section~\ref{Section Implementation}.

Combining Theorems~\ref{thm:projection} and~\ref{thm:MC-test} gives the following immediate consequence.

\begin{corollary}[Finite-sample projection-based inference]
\label{cor:MC-projection}
Let
$\mathcal C_{n,B}^{\mathrm{MC}}(1-\alpha)$
be the Monte Carlo confidence region defined above. Then any projection of
$\mathcal C_{n,B}^{\mathrm{MC}}(1-\alpha)$
through an optimization value or solution correspondence satisfying Assumption~\ref{ass:projection} has coverage probability at least
$1-\alpha_B$.
In particular, if
$\alpha(B+1)$
is an integer, its finite-sample coverage probability is at least
$1-\alpha$.
\end{corollary}

\subsection{Projected inference for the ordinal distributional discrepancy}
\label{subsec:projected-D}

Fix the threshold-additive geometry $d_w$ and its associated discrepancy
$D_w$. Section~\ref{Section PI} established that the identified set for the ordinal distributional discrepancy is
$
\left[
\underline D_w(\vartheta),
\overline D_w(\vartheta)
\right],
$
where
$\underline D_w(\vartheta)
=
\min_{(\mu,\nu)\in\mathcal Z(\vartheta)}
D_w(\mu,\nu)
$ and
$
\overline D_w(\vartheta)
=
\max_{(\mu,\nu)\in\mathcal Z(\vartheta)}
D_w(\mu,\nu),
$
with
\(
\mathcal Z(\vartheta)
=
\mathcal M(\mathbf q_X)
\times
\mathcal M(\mathbf q_Y).
\)
Let $\mathcal C_n\subseteq\Theta$ denote a generic confidence region for the observable parameter. Define the projected lower and upper bounds by
$$
\underline D_{w,n}(\mathcal C_n)
=
\inf_{\vartheta\in\mathcal C_n}
\underline D_w(\vartheta)
\quad\text{and}
\quad
\overline D_{w,n}(\mathcal C_n)
=
\sup_{\vartheta\in\mathcal C_n}
\overline D_w(\vartheta).
$$
The corresponding projected confidence interval for the identified set is
\(
\left[
\underline D_{w,n}(\mathcal C_n),
\overline D_{w,n}(\mathcal C_n)
\right].
\)

\begin{corollary}[Projected inference for the discrepancy]
\label{cor:ordinalCI}
For any probability measure $\mathbb{P}$ governing the randomness of $\mathcal{C}_n$, if
\(
\mathbb{P}
\left(
\vartheta_0\in\mathcal C_n
\right)
\geq
1-\alpha^\star,
\)
then
\[
\mathbb{P}
\left(
\left[
\underline D_w(\vartheta_0),
\overline D_w(\vartheta_0)
\right]
\subseteq
\left[
\underline D_{w,n}(\mathcal C_n),
\overline D_{w,n}(\mathcal C_n)
\right]
\right)
\geq
1-\alpha^\star.
\]
\end{corollary}

\begin{proof}
On the event
$\{\vartheta_0\in\mathcal C_n\}$,
\(
\underline D_{w,n}(\mathcal C_n)
\leq
\underline D_w(\vartheta_0)
\leq
\overline D_w(\vartheta_0)
\leq
\overline D_{w,n}(\mathcal C_n).
\)
The result follows by taking probabilities.
\end{proof}

With
\(
\mathcal C_n
=
\mathcal C_{n,B}^{\mathrm{MC}}(1-\alpha),
\) and $\mathbb{P}=\mathbb P_{\vartheta_0}^{n,B}$,
Corollary~\ref{cor:ordinalCI} gives
\[
\mathbb P_{\vartheta_0}^{n,B}
\left(
\left[
\underline D_w(\vartheta_0),
\overline D_w(\vartheta_0)
\right]
\subseteq
\left[
\underline D_{w,n}
\left(
\mathcal C_{n,B}^{\mathrm{MC}}
\right),
\overline D_{w,n}
\left(
\mathcal C_{n,B}^{\mathrm{MC}}
\right)
\right]
\right)
\geq
1-\alpha_B
\]
for every
$\vartheta_0\in\Theta$
and every finite pair of sample sizes. If
$\alpha(B+1)$
is an integer, the right-hand side is $1-\alpha$.

\subsection{Projected inference for endpoint-conditioned benchmark transition probabilities}
\label{subsec:projected-benchmarks}

I next consider the endpoint-conditioned benchmark transition probabilities introduced in Section~\ref{Section PI}. Unlike the preceding subsection, where the inferential objects are optimal values, inference here concerns optimization sets generated by potentially nonunique endpoint-attaining latent marginals and potentially nonunique optimal transport plans.

Let
$
S_{\underline D_w}(\vartheta)
=
\argmin_{(\mu,\nu)\in\mathcal Z(\vartheta)}
D_w(\mu,\nu)
$ and
$
S_{\overline D_w}(\vartheta)
=
\argmax_{(\mu,\nu)\in\mathcal Z(\vartheta)}
D_w(\mu,\nu)
$
denote the solution correspondences associated with the lower and upper identification endpoints. For each pair
$(\mu,\nu)$,
let
\(
\Pi_w^\star(\mu,\nu)
\)
denote the set of optimal transport plans between $\mu$ and $\nu$ under the fixed
threshold-additive ordinal ground metric $d_w$.

Define the lower endpoint-conditioned benchmark correspondence by
\[
\mathcal Z_{\underline D_w}(\vartheta)
=
\left\{
(\mu,\nu,\pi):
(\mu,\nu)\in S_{\underline D_w}(\vartheta),
\;
\pi\in\Pi_w^\star(\mu,\nu)
\right\},
\]
and analogously define the upper endpoint-conditioned benchmark correspondence by
\[
\mathcal Z_{\overline D_w}(\vartheta)
=
\left\{
(\mu,\nu,\pi):
(\mu,\nu)\in S_{\overline D_w}(\vartheta),
\;
\pi\in\Pi_w^\star(\mu,\nu)
\right\}.
\]
Appendix~\ref{appendix:projection-measurability} establishes that these correspondences are suitably measurable with nonempty compact values.

For a generic confidence region
$\mathcal C_n$,
define
\begin{align*}
\mathcal A_{\underline D_w}(\mathcal C_n)
&:=
\left\{
(\vartheta,\mu,\nu,\pi):
\vartheta\in\mathcal C_n,\;
(\mu,\nu,\pi)\in
\mathcal Z_{\underline D_w}(\vartheta)
\right\}
\quad\text{and}\\
\mathcal A_{\overline D_w}(\mathcal C_n)
& :=
\left\{
(\vartheta,\mu,\nu,\pi):
\vartheta\in\mathcal C_n,\;
(\mu,\nu,\pi)\in
\mathcal Z_{\overline D_w}(\vartheta)
\right\}.
\end{align*}

For each cell
$(i,j)\in\{1,\ldots,K\}^2$,
define the projected lower-endpoint bounds
\(\displaystyle
\underline{\pi}^{L}_{w,ij}(\mathcal C_n)
=
\inf_{
(\vartheta,\mu,\nu,\pi)
\in
\mathcal A_{\underline D_w}(\mathcal C_n)
}
\pi_{ij},
\;\text{and}\;
\overline{\pi}^{L}_{w,ij}(\mathcal C_n)
=
\sup_{
(\vartheta,\mu,\nu,\pi)
\in
\mathcal A_{\underline D_w}(\mathcal C_n)
}
\pi_{ij},
\)
and the projected upper-endpoint bounds
\(\displaystyle
\underline{\pi}^{U}_{w,ij}(\mathcal C_n)
=
\inf_{
(\vartheta,\mu,\nu,\pi)
\in
\mathcal A_{\overline D_w}(\mathcal C_n)
}
\pi_{ij},
\;\text{and}\;
\overline{\pi}^{U}_{w,ij}(\mathcal C_n)
=
\sup_{
(\vartheta,\mu,\nu,\pi)
\in
\mathcal A_{\overline D_w}(\mathcal C_n)
}
\pi_{ij}.
\)

These projected bounds cover the entire identified interval for each benchmark transition probability. More strongly, because all of the bounds are generated from the same joint confidence region for $\vartheta$, the coverage statement holds simultaneously across all cells and both discrepancy endpoints.

\begin{corollary}[Projected inference for benchmark transition probabilities]
\label{cor:benchmarkCI}
Let $\alpha^\star\in(0,1)$ and for any probability measure $\mathbb{P}$ governing the randomness of $\mathcal{C}_n$, if
\(
\mathbb{P}
\left(
\vartheta_0\in\mathcal C_n
\right)
\geq
1-\alpha^\star,
\)
then
\begin{align*}
\mathbb{P}
\Bigg(
&
\left[
\underline{\pi}^{L}_{w,ij}(\vartheta_0),
\overline{\pi}^{L}_{w,ij}(\vartheta_0)
\right]
\subseteq
\left[
\underline{\pi}^{L}_{w,ij}(\mathcal C_n),
\overline{\pi}^{L}_{w,ij}(\mathcal C_n)
\right]
\quad\text{for all }(i,j),
\\
&
\left[
\underline{\pi}^{U}_{w,ij}(\vartheta_0),
\overline{\pi}^{U}_{w,ij}(\vartheta_0)
\right]
\subseteq
\left[
\underline{\pi}^{U}_{w,ij}(\mathcal C_n),
\overline{\pi}^{U}_{w,ij}(\mathcal C_n)
\right]
\quad\text{for all }(i,j)
\Bigg)
\geq
1-\alpha^\star.
\end{align*}
\end{corollary}

\begin{proof}
On the event
$\{\vartheta_0\in\mathcal C_n\}$,
the parameter value $\vartheta_0$ is among those over which each projected infimum and supremum is computed. Hence, simultaneously for every $(i,j)$,
\begin{align*}
\underline{\pi}^{L}_{w,ij}(\mathcal C_n)
 &\leq
\underline{\pi}^{L}_{w,ij}(\vartheta_0)
\leq
\overline{\pi}^{L}_{w,ij}(\vartheta_0)
\leq
\overline{\pi}^{L}_{w,ij}(\mathcal C_n), \quad \text{and}\\
\underline{\pi}^{U}_{w,ij}(\mathcal C_n)
 &\leq
\underline{\pi}^{U}_{w,ij}(\vartheta_0)
\leq
\overline{\pi}^{U}_{w,ij}(\vartheta_0)
\leq
\overline{\pi}^{U}_{w,ij}(\mathcal C_n).
\end{align*}
The result follows by taking probabilities.
\end{proof}

Taking
\(
\mathcal C_n
=
\mathcal C_{n,B}^{\mathrm{MC}}(1-\alpha) 
\)
and $\mathbb{P}=\mathbb{P}_{\vartheta_0}^{n,B}$
therefore gives simultaneous finite-sample-valid projected confidence intervals for all endpoint-conditioned benchmark transition probabilities, with coverage probability at least
$1-\alpha_B$.
If
$\alpha(B+1)$
is an integer, this lower bound is $1-\alpha$.

The projection construction therefore accounts jointly for three distinct sources of uncertainty: sampling variability in the observable repeated cross-sections, partial identification of the latent marginal distributions under missingness, and nonuniqueness of the endpoint-attaining optimal transport plans. The finite-sample Monte Carlo calibration accomplishes this without relying on a large-sample approximation.


\section{Computational Implementation}
\label{Section Implementation}

The projection procedures in Section~\ref{Section Inference} require optimization over the finite-sample Monte Carlo confidence region for the observable-data parameter
\(
\vartheta=(\mathbf q_X,\mathbf q_Y)\in\Theta.
\)
For brevity in this section, write
\(
\mathcal C_{n,B}^{\mathrm{MC}}
:=
\mathcal C_{n,B}^{\mathrm{MC}}(1-\alpha)
=
\{\vartheta\in\Theta:p_{n,B}(\vartheta)>\alpha\}.
\)
The confidence region need not be constructed explicitly. Instead, the observable probabilities are retained as decision variables and membership in \(\mathcal C_{n,B}^{\mathrm{MC}}\) is imposed directly as a constraint in the outer optimization problems. Conditional on the observed data and the auxiliary random numbers used to construct the Monte Carlo test, this is a deterministic constraint on \(\vartheta\). Its geometry is generally nonsmooth and locally discontinuous because changes in \(\vartheta\) can change the simulated multinomial counts and hence the Monte Carlo rank. Conditional on \(\vartheta\), however, the missing-data and transport calculations remain finite-dimensional linear programs.

\subsection{Projected discrepancy bounds}
\label{subsec:comp-D}

For \(k=1,\ldots,K-1\), recall that 
\(
C_k(\mu-\nu)
=
F_\mu(k)-F_\nu(k)
=
\sum_{j=1}^k(\mu_j-\nu_j).
\)
The projected lower confidence bound defined in Section~\ref{subsec:projected-D} can be computed as
\begin{equation}
\label{eq:mc-projected-D-lower}
\begin{aligned}
\underline D_{w,n}\!\left(\mathcal C_{n,B}^{\mathrm{MC}}\right)
=
\inf_{\mathbf q_X,\mathbf q_Y,\mu,\nu,t}\quad
& \sum_{k=1}^{K-1} w_k\,t_k\\
\text{subject to}\quad
& (\mathbf q_X,\mathbf q_Y)\in \mathcal C_{n,B}^{\mathrm{MC}}, \quad\mu,\nu\in\Delta^K,\\
& q_{Xj}\leq \mu_j\leq q_{Xj}+q_{X0},
    \; j=1,\ldots,K,\\
& q_{Yj}\leq \nu_j\leq q_{Yj}+q_{Y0},
    \; j=1,\ldots,K,\\
&  C_k(\mu-\nu)\leq t_k,\quad
  - C_k(\mu-\nu)\leq t_k,
    \; k=1,\ldots,K-1,\\
& t_k\geq0,\; k=1,\ldots,K-1.
\end{aligned}
\end{equation}
The auxiliary variables \(t_k\) linearize the absolute values in
\(D_w(\mu,\nu)=\sum_{k=1}^{K-1} w_k\,|C_k(\mu-\nu)|\). Hence, apart from the Monte Carlo confidence-region restriction on \((\mathbf q_X,\mathbf q_Y)\), all restrictions in \eqref{eq:mc-projected-D-lower} are linear.

For the upper bound use
\(
\sum_{k=1}^{K-1}w_k\,|C_k(\mu-\nu)|
=
\max_{s\in\mathcal S_K}
\sum_{k=1}^{K-1}s_k\,\,w_k\,C_k(\mu-\nu),
\;
\mathcal S_K=\{-1,1\}^{K-1}.
\)
For each \(s\in\mathcal S_K\), define
\begin{equation}
\label{eq:mc-projected-D-upper-sign}
\begin{aligned}
\overline D_{w,n,s}\!\left(\mathcal C_{n,B}^{\mathrm{MC}}\right)
=
\sup_{\mathbf q_X,\mathbf q_Y,\mu,\nu}\quad
& \sum_{k=1}^{K-1}s_k\,w_k\,C_k(\mu-\nu)\\
\text{subject to}\quad
& (\mathbf q_X,\mathbf q_Y)\in \mathcal C_{n,B}^{\mathrm{MC}},\\
& \mu,\nu\in\Delta^K,\\
& q_{Xj}\leq \mu_j\leq q_{Xj}+q_{X0},
    \; j=1,\ldots,K,\\
& q_{Yj}\leq \nu_j\leq q_{Yj}+q_{Y0},
    \; j=1,\ldots,K.
\end{aligned}
\end{equation}
Then
\(
\overline D_{w,n}\!\left(\mathcal C_{n,B}^{\mathrm{MC}}\right)
=
\max_{s\in\mathcal S_K}
\overline D_{w,n,s}\!\left(\mathcal C_{n,B}^{\mathrm{MC}}\right).
\)
Thus the projected upper endpoint requires \(2^{K-1}\) sign-specific outer problems; in the four-category application there are eight.

The candidate-specific endpoint functions used in these outer problems are exactly the functions already defined in Section~\ref{subsec:projected-D},
\[
\underline D_w(\vartheta)
=
\min_{(\mu,\nu)\in\mathcal Z(\vartheta)}D_w(\mu,\nu),
\qquad
\overline D_w(\vartheta)
=
\max_{(\mu,\nu)\in\mathcal Z(\vartheta)}D_w(\mu,\nu),
\]
where
\(\mathcal Z(\vartheta)=\mathcal M(\mathbf q_X)\times\mathcal M(\mathbf q_Y)\).
By Proposition~\ref{prop:endpoint-correspondences}, \(\underline D_w\) and \(\overline D_w\) are continuous on \(\Theta\). Consequently,
\(
\underline D_{w,n}\!\left(\mathcal C_{n,B}^{\mathrm{MC}}\right)
=
\inf_{\vartheta\in\mathcal C_{n,B}^{\mathrm{MC}}}\underline D_w(\vartheta),
\;
\overline D_{w,n}\!\left(\mathcal C_{n,B}^{\mathrm{MC}}\right)
=
\sup_{\vartheta\in\mathcal C_{n,B}^{\mathrm{MC}}}\overline D_w(\vartheta).
\)

\subsection{Endpoint-conditioned benchmark transition probabilities}
\label{subsec:comp-benchmarks}

The endpoint-conditioned benchmark bounds have a nested structure. For fixed \(\vartheta=(\mathbf q_X,\mathbf q_Y)\), the admissible latent marginals are determined by \(\mathcal M(\mathbf q_X)\) and \(\mathcal M(\mathbf q_Y)\). Conditional on these observable probabilities, the discrepancy endpoints and associated transport-cell extrema are finite-dimensional linear programs.

For each endpoint \(e\in\{L,U\}\) and each cell \((i,j)\), let
\(\underline\pi_{w,ij}^{\,e}(\vartheta)\) and \(\overline\pi_{w,ij}^{\,e}(\vartheta)\) denote the candidate-specific cell extrema defined by the endpoint-conditioned benchmark correspondences in Section~\ref{subsec:projected-benchmarks}. The Monte Carlo projected cell bounds are therefore
\begin{align*}
\underline\pi_{w,ij}^{L}\!\left(\mathcal C_{n,B}^{\mathrm{MC}}\right)
&=
\inf_{\vartheta\in\mathcal C_{n,B}^{\mathrm{MC}}}
\underline\pi_{w,ij}^{L}(\vartheta),
&
\overline\pi_{w,ij}^{L}\!\left(\mathcal C_{n,B}^{\mathrm{MC}}\right)
&=
\sup_{\vartheta\in\mathcal C_{n,B}^{\mathrm{MC}}}
\overline\pi_{w,ij}^{L}(\vartheta),
\\
\underline\pi_{w,ij}^{U}\!\left(\mathcal C_{n,B}^{\mathrm{MC}}\right)
&=
\inf_{\vartheta\in\mathcal C_{n,B}^{\mathrm{MC}}}
\underline\pi_{w,ij}^{U}(\vartheta),
&
\overline\pi_{w,ij}^{U}\!\left(\mathcal C_{n,B}^{\mathrm{MC}}\right)
&=
\sup_{\vartheta\in\mathcal C_{n,B}^{\mathrm{MC}}}
\overline\pi_{w,ij}^{U}(\vartheta).
\end{align*}
Thus four outer optimization problems are required for each transition cell. Conditional on the observable probabilities, all inner endpoint and transport-cell calculations are linear programs, and the cell calculations are mutually independent and can therefore be performed in parallel.

Unlike \(\underline D_w\) and \(\overline D_w\), the candidate-specific cell-extremal value functions need not be continuous at every observable parameter because the endpoint-conditioned optimal face can change at degenerate parameter values. They are continuous on regions on which the relevant endpoint branch, sign configuration, and optimal linear-programming structure are locally unchanged. Appendix~\ref{appendix:global-search} formalizes the regular-region representation used below.

\subsection{Monte Carlo implementation and persistent global exploration}
\label{subsec:mc-global-search}

The Monte Carlo confidence region is evaluated using common random numbers. Before optimization, I generate and fix the auxiliary random numbers required for all \(B\) simulated samples and for randomized tie-breaking. For each candidate observable parameter \(\vartheta\), these fixed random numbers are transformed into multinomial counts having exactly the distribution implied by \(\vartheta\). I use a sequential conditional-binomial inverse transform with fixed auxiliary uniforms. Consequently, repeated evaluations of the same candidate return the same Monte Carlo \(p\)-value, while for every fixed candidate the simulated counts retain the exact null distribution required by Theorem~\ref{thm:MC-test}.

To separate statistical and computational randomness, let \(\mathcal G_{n,B}\) denote the sigma-field generated by the observed samples together with the auxiliary Monte Carlo random variables used to generate the simulated counts and randomized tie breakers. Conditional on \(\mathcal G_{n,B}\), the realized confidence region \(\mathcal C_{n,B}^{\mathrm{MC}}\) and all candidate-specific projection objectives are deterministic. Let
\(
(\Omega_{\mathrm{alg}},\mathcal F_{\mathrm{alg}},\mathbb P_{\mathrm{alg}})
\)
govern the random numbers used only by the outer optimization algorithm, independently of \(\mathcal G_{n,B}\). All almost-sure convergence statements below are with respect to \(\mathbb P_{\mathrm{alg}}\), for a fixed realization of \(\mathcal G_{n,B}\).

I search the nonsmooth outer problems using simulated annealing augmented by persistent global exploration. Let \(\mathcal F^{\mathrm{alg}}_{m-1}\) denote the algorithmic search history through proposal \(m-1\). Conditional on this history, the \(m\)th proposal has law
\begin{equation}
\label{eq:mixture-proposal}
\mathbb P_{\mathrm{alg}}
\!\left(
\vartheta_m^{\mathrm{prop}}\in A
\mid
\mathcal F^{\mathrm{alg}}_{m-1}
\right)
=
(1-\rho)\,
Q_m^{\mathrm{loc}}
\!\left(
A\mid\mathcal F^{\mathrm{alg}}_{m-1}
\right)
+
\rho\,\Gamma(A),
\qquad
0<\rho<1,
\end{equation}
for every Borel set \(A\subseteq\Theta\). Here \(Q_m^{\mathrm{loc}}\) is the local simulated-annealing proposal law and \(\Gamma\) is a fixed global proposal measure on \(\Theta\). In the implementation, \(\Gamma\) is induced by a uniform global proposal in the bounded simplex parameterization and has full support on \(\Theta\). The local component retains the exponential cooling schedule used for efficient finite-budget search; the global component is used persistently with probability \(\rho\).

The convergence argument requires only that the global component assign positive probability to arbitrarily near-optimal feasible candidates. To state this condition, fix a reported outer objective \(h:\Theta\to\mathbb R\). If \(h\) is minimized, let
\(
h^\star
=
\inf_{\vartheta\in\mathcal C_{n,B}^{\mathrm{MC}}}h(\vartheta),
\;\text{and}\;
A_\varepsilon(h)
=
\left\{
\vartheta\in\mathcal C_{n,B}^{\mathrm{MC}}:
h(\vartheta)<h^\star+\varepsilon
\right\}.
\)
If \(h\) is maximized, let
\(
h^\star
=
\sup_{\vartheta\in\mathcal C_{n,B}^{\mathrm{MC}}}h(\vartheta),
\;\text{and}\;
A_\varepsilon(h)
=
\left\{
\vartheta\in\mathcal C_{n,B}^{\mathrm{MC}}:
h(\vartheta)>h^\star-\varepsilon
\right\}.
\)

\begin{assumption}
\label{ass:near-optimal-accessibility}
\textbf{Near-optimal accessibility.}
Conditional on \(\mathcal G_{n,B}\), for every reported outer projection objective \(h\) and every \(\varepsilon>0\),
\(
\Gamma\!\left(A_\varepsilon(h)\right)>0.
\)
\end{assumption}

Assumption~\ref{ass:near-optimal-accessibility} is a condition on the realized projection problem and the global proposal measure, not on the true observable parameter. It requires neither attainment nor uniqueness of the outer extremum. Rather, it rules out the computational pathology in which all candidates arbitrarily close to the exact projected extremum lie in a set that is null under the global proposal measure. Equivalently, the ordinary infimum or supremum over \(\mathcal C_{n,B}^{\mathrm{MC}}\) agrees with the corresponding \(\Gamma\)-essential extremum. Appendix~\ref{appendix:global-search} gives a simple regular-region condition that is sufficient for near-optimal accessibility and relates it to the structure of the Monte Carlo acceptance rule and the candidate-specific optimization problems.

Starting from any feasible incumbent, let \(\widehat h_N\) denote the best feasible objective value encountered after the first \(N\) proposals generated by \eqref{eq:mixture-proposal}, with minimization and maximization interpreted in the appropriate direction.

\begin{proposition}[Almost-sure global convergence of the outer projection search]
\label{prop:persistent-global-convergence}
Fix a realization of \(\mathcal G_{n,B}\) for which Assumption~\ref{ass:near-optimal-accessibility} holds, and suppose that the proposal law satisfies \eqref{eq:mixture-proposal} with \(\rho>0\). For every reported outer minimization problem,
\[
\mathbb P_{\mathrm{alg}}\!\left(
\lim_{N\to\infty}\widehat h_N
=
\inf_{\vartheta\in\mathcal C_{n,B}^{\mathrm{MC}}}h(\vartheta)
\right)
=1,
\]
and for every reported outer maximization problem,
\[
\mathbb P_{\mathrm{alg}}\!\left(
\lim_{N\to\infty}\widehat h_N
=
\sup_{\vartheta\in\mathcal C_{n,B}^{\mathrm{MC}}}h(\vartheta)
\right)
=1.
\]
The result does not require convexity or differentiability of the Monte Carlo confidence region, continuity of \(h\), uniqueness of any inner optimizer, or a particular cooling schedule for the local simulated-annealing component.
\end{proposition}
\begin{proof}
See  Appendix~\ref{proof:global-search-proposition}.
\end{proof}
For every \(\varepsilon>0\), Assumption~\ref{ass:near-optimal-accessibility} gives \(\Gamma(A_\varepsilon(h))>0\). The persistent global component therefore proposes a point in \(A_\varepsilon(h)\) with probability one as the computational budget increases. Taking a countable sequence \(\varepsilon\downarrow0\) yields convergence of the best feasible value to the exact projected extremum. The argument concerns computational approximation conditional on the realized Monte Carlo confidence region and is separate from the finite-sample coverage result in Section~\ref{Section Inference}.

\begin{corollary}[Projected discrepancy bounds]
\label{cor:global-D}
Under the conditions of Proposition~\ref{prop:persistent-global-convergence},
\(
\widehat{\underline D}_{w,n,N}
\rightarrow
\underline D_{w,n}\!\left(\mathcal C_{n,B}^{\mathrm{MC}}\right)
\;\text{and}\;
\widehat{\overline D}_{w,n,N}
\rightarrow
\overline D_{w,n}\!\left(\mathcal C_{n,B}^{\mathrm{MC}}\right),
\;
\mathbb P_{\mathrm{alg}}\text{-almost surely}.
\)
\end{corollary}
\begin{proof}
Apply Proposition~\ref{prop:persistent-global-convergence} with \(h=\underline D_w\) for the lower projected discrepancy endpoint and \(h=\overline D_w\) for the upper projected discrepancy endpoint. The former is minimized and the latter is maximized over \(\mathcal C_{n,B}^{\mathrm{MC}}\). Both convergence statements therefore hold \(\mathbb P_{\mathrm{alg}}\)-almost surely.
\end{proof}

\begin{corollary}[Endpoint-conditioned benchmark bounds]
\label{cor:global-benchmarks}
Under the conditions of Proposition~\ref{prop:persistent-global-convergence}, for every endpoint \(e\in\{L,U\}\) and every cell \((i,j)\in\{1,\dots,K\}^2\), 
\(
\widehat{\underline\pi}_{w,ij,N}^{\,e}
\rightarrow
\underline\pi_{w,ij}^{e}\!\left(\mathcal C_{n,B}^{\mathrm{MC}}\right)
\;\text{and}\;
\widehat{\overline\pi}_{w,ij,N}^{\,e}
\rightarrow
\overline\pi_{w,ij}^{e}\!\left(\mathcal C_{n,B}^{\mathrm{MC}}\right),
\;
\mathbb P_{\mathrm{alg}}\text{-almost surely}.
\)
Because the collection of cells and endpoints is finite, these convergences hold simultaneously on a single event of \(\mathbb P_{\mathrm{alg}}\)-probability one.
\end{corollary}
\begin{proof}
For each \((i,j)\) and \(e\in\{L,U\}\), apply Proposition~\ref{prop:persistent-global-convergence} to \(h=\underline\pi_{w,ij}^{e}\) for the projected lower cell bound and to \(h=\overline\pi_{w,ij}^{e}\) for the projected upper cell bound. Each convergence event has \(\mathbb P_{\mathrm{alg}}\)-probability one. Because there are finitely many cells and endpoints, the intersection of these events also has \(\mathbb P_{\mathrm{alg}}\)-probability one. Hence all reported endpoint-conditioned benchmark bounds converge simultaneously to their exact Monte Carlo projection extrema as \(N\to\infty\).
\end{proof}

Proposition~\ref{prop:persistent-global-convergence} is an asymptotic
computational guarantee: it does not assert that a finite run has exactly
attained the global optimum. Numerical optimization error is therefore still
assessed separately from statistical validity. I use multiple starting values,
retain the best feasible objective across all starts, and record the stability
of the attained extrema as the computational budget increases. Solver exit
flags, Monte Carlo acceptance, validity of every inner linear program, simplex
and transport feasibility, ordering of lower and upper benchmark bounds, and
membership of reported probabilities in \([0,1]\) are also checked. Increasing
the number of Monte Carlo replications \(B\) has a distinct role: it refines
the Monte Carlo rank resolution and can affect the power and concordance
properties of the randomized finite-sample test, whereas increasing the
optimization budget refines the numerical approximation to the projection
extrema for a fixed realization of the Monte Carlo confidence region.


For a general threshold-additive geometry, discrepancy results are
reported in the units induced by the ground metric $d_w$. Recall that by
Proposition~\ref{prop:OT_representation},
\[
D_w(\mu,\nu)
=
\min_{\pi\in\Pi(\mu,\nu)}
\sum_{i=1}^K\sum_{j=1}^K
d_w(\xi_i,\xi_j)\pi_{ij},
\]
so $D_w(\mu,\nu)$ is the smallest average displacement compatible with
the two marginal distributions under the specified ordinal geometry.
Its maximum possible value is the diameter of the ordered category
space,
\(
\operatorname{diam}(\Xi,d_w)
=
d_w(\xi_1,\xi_K)
=
\sum_{k=1}^{K-1}w_k.
\)
Hence
\(D_w(\mu,\nu)
/(\sum_{k=1}^{K-1}w_k)
\)
expresses the discrepancy as a proportion of the maximum possible
average displacement under the specified geometry.

Under threshold neutrality and unit normalization, $w_k=1$ for every
$k$, so this specialization reduces to the canonical rank
interpretation used in the empirical application: the discrepancy is
the smallest average absolute rank displacement, its maximum possible
value is $K-1$, and $D/(K-1)$ gives the corresponding normalized
measure. The benchmark matrices report transported probability mass
between source and target categories and are not interpreted as
realized individual transitions.

\section{Discussion}
\label{Section Discussion}

The framework characterizes what marginal changes in repeated ordinal cross-sections imply about the
minimum restructuring needed to reconcile them. This section clarifies the interpretation and limitations
of that information and discusses the scope of the inferential and computational approach. Additional interpretations and extensions---including maximal reassignment,
panel-data comparisons, connections to Fr\'echet classes, binary poverty
transitions, Gini mean differences and dependence, and higher-order quantile
discrepancies---are developed in Appendix~\ref{appendix:further-discussion}.

\subsection{Interpretation and Limitations}

As discussed in the Introduction, the measurement characterization connects
two different axiomatic perspectives. The ordinal-measurement literature
motivates cumulative probabilities and adjacent-category movements as natural
ways of representing change on an ordered categorical scale. The
probability-metrics literature, and in particular work on norm extensions of
a prescribed ground metric, asks how distances between probability laws
inherit structure from the underlying state space. The present approach
connects these perspectives: ordinal thresholds determine the elementary
coordinates of distributional change, while norm inducibility, separability
across distinct ordinal cuts, and consistency with the ground geometry
determine how those coordinates aggregate into a probability metric between
two arbitrary marginals.

The measurement characterization also clarifies the scope of
Theorem~\ref{thm:ordinal_metric}. The result does not assert that every
meaningful comparison of ordinal distributions must take a cumulative
absolute-difference form. It characterizes that form within the class of
norm-induced probability metrics that are compatible with a threshold-additive
ordinal geometry and separable across distinct ordinal cuts. The ordinal
structure alone does not determine the relative threshold weights: these remain
unrestricted in the general characterization and become equal only after
imposing threshold neutrality and a normalization. Alternative questions---for
example, directional welfare comparisons, dispersion relative to a reference
distribution, or higher-order notions of displacement---can therefore lead
naturally to different axioms and different functionals.

Familiar distributional metrics provide useful comparisons because they
embody different ways of measuring differences between ordinal
distributions. On the ordered support, the Kolmogorov--Smirnov distance is
\(
D_{\mathrm{KS}}(\mu,\nu)
=
\max_{1\leq k<K}|F_\mu(k)-F_\nu(k)|,
\)
so it uses the same cumulative threshold coordinates but aggregates them by
their largest imbalance rather than additively. It therefore does not satisfy
the threshold-separability principle used here. Total variation,
\(
D_{\mathrm{TV}}(\mu,\nu)
=
\frac{1}{2}\sum_{i=1}^K|\mu_i-\nu_i|,
\)
instead compares category probabilities directly and is invariant to
permutations of the category labels, so it does not exploit their ordering.
Total variation also admits a coupling representation as the minimum
probability of assigning different categories, equivalently as optimal
transport under the zero--one cost. Under that geometry, however, every
category change has the same cost regardless of how many ordinal thresholds
are crossed. The three measures therefore emphasize different features of
distributional change. Kolmogorov--Smirnov records the largest imbalance
across ordinal thresholds, total variation measures the overall discrepancy
in category probabilities without regard to ordinal distance, and the
discrepancy studied here aggregates displacement across the successive
ordinal thresholds. Its optimal plans consequently have a specifically
ordinal interpretation as least-reassignment transition benchmarks under the ordinal geometry

The weighted formulation of the discrepancy also permits additional cardinal information to be
incorporated when it is substantively justified. For example, if the ordered
categories are associated with known cardinal locations, unequal threshold
weights can reflect the corresponding differences in adjacent spacing. What
the category-level framework does not use is information about latent
locations within a category. If such information were available---for
example, distances of individuals from a threshold within an income
interval---the outcome space or information structure could be enriched to
incorporate it. The canonical equal-threshold specification is therefore
appropriate to the setting in which only the ordering of the observed
categories is taken as informative.

The conservative benchmark is not an estimate of the actual transition
structure behind observed marginal change. Many different latent transitions
can produce the same marginals, including those with offsetting movements.
The discrepancy instead measures the minimum aggregate ordinal restructuring
compatible with the marginals, and its optimal transport plans show alternative
ways this minimum can be organized across categories.

Nonuniqueness in optimal plans is substantive, not merely computational. Repeated cross-sections may
reveal the minimum restructuring needed but cannot specify how it is allocated across individual flows.
Endpoint-conditioned cell bounds preserve this ambiguity under partial observation by reporting the full
range of benchmark reallocations associated with each discrepancy endpoint.

The limitation of the approach is clearest when the marginal distributions coincide. In that case,
$D_w(\mu,\nu)=0$, because no restructuring is required to reconcile the marginals, even though potentially
large offsetting individual transitions may have occurred. More generally, the conservative benchmark
identifies movement that is required by marginal change, not all movement that the marginals are capable
of hiding. This distinction is why its optimal plans are interpreted as conservative transition benchmarks
rather than reconstructed transitions.

There is nevertheless complementary information in asking how much aggregate
movement the same marginals can accommodate.
Appendix~\ref{appendix:maximal-reassignment} defines the maximal-reassignment
functional $M_w(\mu,\nu)$ and shows that
\(
[D_w(\mu,\nu),M_w(\mu,\nu)]
\)
is the sharp identified set for the aggregate ordinal reassignment cost when
only the two marginals are known. The two endpoints answer different questions:
$D_w$ measures the restructuring that the marginal change \emph{requires}, whereas
$M_w$ measures the greatest restructuring those marginals \emph{fail to rule out}.
Accordingly, $M_w$ is not an alternative measure of ordinal distributional
change, but an extremal functional of the admissible couplings after the ordinal
geometry has been fixed. Appendix~\ref{appendix:maximal-reassignment} develops
this relationship further, and Appendix~\ref{appendix:frechet} distinguishes these aggregate bounds from classical
cellwise Fr\'echet bounds, and considers the case in which population transition
data are observed directly.
%

The framework can also be applied to counterfactual or covariate-adjusted
marginal distributions. For example, one could construct counterfactual
ordinal distributions that hold fixed the distribution of observed
covariates while allowing the conditional outcome distribution to change,
or vice versa, and then use the ordinal discrepancy and conservative
benchmarks to measure the restructuring associated with these alternative
marginal changes. This suggests a connection with distributional
decomposition methods that separate changes in covariate composition from
changes in conditional outcome relationships, as in the distributional 
decomposition literature; see, for example,
\citet{FirpoFortinLemieux2018}. Developing such
decompositions within the present ordinal framework is left for future work.

\subsection{Inferential and Computational Scope}

The projection approach is useful here because the observable-data
distribution is finite dimensional while the objects of interest are optimal
values and generally nonunique optimizer sets. Rather than relying on local
approximations to these optimization mappings, inference is conducted for the
observable probabilities and then propagated to the discrepancy and benchmark
objects by projection. In the multinomial setting considered here, randomized
Monte Carlo calibration provides the primitive confidence region with
finite-sample validity. The projection step then preserves that validity
without requiring differentiability of the optimization mappings or uniqueness
of their optimizers.

The underlying projection logic is more general than the ordinal transport
application. Whenever a finite-dimensional observable parameter determines a
feasible correspondence and a valid confidence region is available for that
parameter, the same inclusion argument can be used to obtain confidence sets
for optimization values and solution correspondences, subject to the requisite
measurability conditions. What is specific to the present setting is the exact
multinomial Monte Carlo calibration and the tractability of the
candidate-specific missing-data and transport problems as finite-dimensional
linear programs.

Statistical validity and numerical approximation should nevertheless be kept
distinct. The coverage results concern the exact projection of the primitive
confidence region, whereas a finite numerical search produces an approximation
to its extrema. Persistent global exploration provides almost-sure convergence
to those extrema under Near-optimal accessibility, but not a finite-run
certificate of global optimality. The discrepancy projections are comparatively
regular because their candidate-specific endpoint functions are continuous,
whereas the endpoint-conditioned benchmark extrema can change when optimal
linear-programming faces change. Verification of Near-optimal accessibility can
therefore be more problem-specific for the benchmark projections. Detailed
sufficient conditions and the geometry of the realized Monte Carlo acceptance
region are given in Appendix~\ref{appendix:global-search}.

More generally, the practical cost of projection depends on the dimension of
the observable-parameter space and on the complexity of the inner optimization
problems. These computational considerations affect how accurately the exact
projection can be approximated with a finite numerical budget; they do not
alter the finite-sample coverage statement itself.
\section{Empirical Illustration}
\label{Section Empirical Section}

\subsection{Remittances and the repeated-cross-section setting}

I apply the framework to changes in the frequency of remittance receipt among Lebanese survey respondents before and during the economic and financial crisis that
began in 2019. Remittances have long been an important source of household
support and foreign-currency inflows in Lebanon. UNDP reports that nominal
inflows remained comparatively stable at approximately US\$6--7 billion per
year over 2011--2021, even as the domestic economy contracted sharply, and
that remittances reached 37.8\% of GDP in 2022
\citep{UNDP2023Remittances}. Aggregate stability, however,
need not imply stability in the distribution of remittance support across
households.

This distinction became particularly important during the crisis. UNDP
describes remittances as increasingly serving the role of a private social
safety net and emphasizes improved information on the reach and frequency of
remittance flows as an input into policies concerning diaspora resources and
household welfare. The repeated-cross-section setting is therefore well suited
to the question studied here: whether the population distribution shifted
toward broader or more regular remittance receipt, and what minimum
restructuring is necessarily implied by such a shift.

I use repeated cross-sections from the \emph{Arab Barometer} for Lebanon and
focus on question Q1017:

\begin{quote}
``Does your family receive remittances from any immediate or extended family
member living abroad?''
\end{quote}

Responses are recorded as 1 = Yes, monthly; 2 = Yes, a few times a year;
3 = Yes, once a year; and 4 = We do not receive anything. Responses 98 =
Don't know and 99 = Refused to answer are treated as item nonresponse. I
compare Wave 4, conducted in 2016--17, with Wave 7, conducted in 2021--22.
Wave 4 is treated as the source distribution and Wave 7 as the target
distribution.

The analysis concerns the respondent population in each wave. In particular,
the distributions studied below are distributions of reported family
remittance receipt among survey respondents; they should not be interpreted
as distributions for the entire Lebanese household population without
additional assumptions concerning survey participation. The missing-data
analysis allows for unrestricted item nonresponse to Q1017 among respondents,
but does not address unit nonresponse into the survey. In what follows,
``population'' refers to this respondent population unless stated otherwise.

The comparison is descriptive rather than causal: the repeated
cross-sections do not identify the causal effect of the crisis on remittance
behaviour. The objective is instead to monitor how the population distribution
of remittance receipt changed across the two periods, particularly its reach
and frequency. This is closely aligned with the policy motivation emphasized
by UNDP, which highlights the importance of improved information on the
distribution and regularity of remittance flows for understanding their role
in household support.

For the empirical application, I use the canonical threshold-neutral
specialization of the general ordinal geometry developed in Section~\ref{subsec:ordinal-geometry}.
There is no external cardinal information that would justify assigning
different costs to the three adjacent thresholds of the remittance-frequency
scale. I therefore set
\(
w_1=w_2=w_3=1,
\)
so that
\(
d_w(\xi_i,\xi_j)=|i-j|
\)
and the weighted discrepancy $D_w$ reduces to the canonical discrepancy
\(
D(\mu,\nu)
=
\sum_{k=1}^{K-1}
|F_\mu(k)-F_\nu(k)|.
\)
The substantive response categories already coincide with their canonical
ranks. The three nontrivial ordinal thresholds also have direct economic
interpretations. For either wave, the cumulative probability through rank~3
is the respondent-population share receiving any remittances; the cumulative
probability through rank~2 is the share receiving remittances at least a few
times a year; and the cumulative probability through rank~1 is the share
receiving them monthly. The ordering therefore captures both the extensive
margin of remittance receipt and progressively more regular forms of support
without assigning cardinal monetary values to the categories.


The respective sample sizes are 1500 and 2399. There are 12 item-nonresponse
observations in Wave 4 and 6 in Wave 7, corresponding to nonresponse rates of
0.8\% and approximately 0.25\%. Under the unrestricted missing-data model,
the latent categories of these observations may be allocated arbitrarily
across the four ordered outcomes. The resulting identified sets therefore
remain agnostic about remittance receipt among respondents who did not provide
a valid answer to Q1017. Accordingly, the partial-identification analysis addresses uncertainty about
Q1017 within the respondent population; it does not extrapolate to individuals
who did not participate in the survey.

The numerical procedures used in the empirical application were implemented
in MATLAB R2022b; the projection calculations were carried out using parallel
processing on Monash University's M3 high-performance computing facility.

\subsection{Distributional change and minimum restructuring}

Figure~\ref{Fig-CDF WC Bounds} reports the plug-in worst-case CDF bounds implied by
unrestricted item nonresponse; they do not account for sampling uncertainty. The bounds are separated at every nontrivial
ordinal threshold. Even under allocations of nonresponse least favourable to
a difference between the waves, the share receiving remittances monthly is at
least 6.1 percentage points higher in Wave 7, the share receiving them at
least a few times a year is at least 13.3 percentage points higher, and the
share receiving any remittances is at least 15.1 percentage points higher.

\begin{figure}[pt]
\centering
\includegraphics[width=0.75\textwidth]{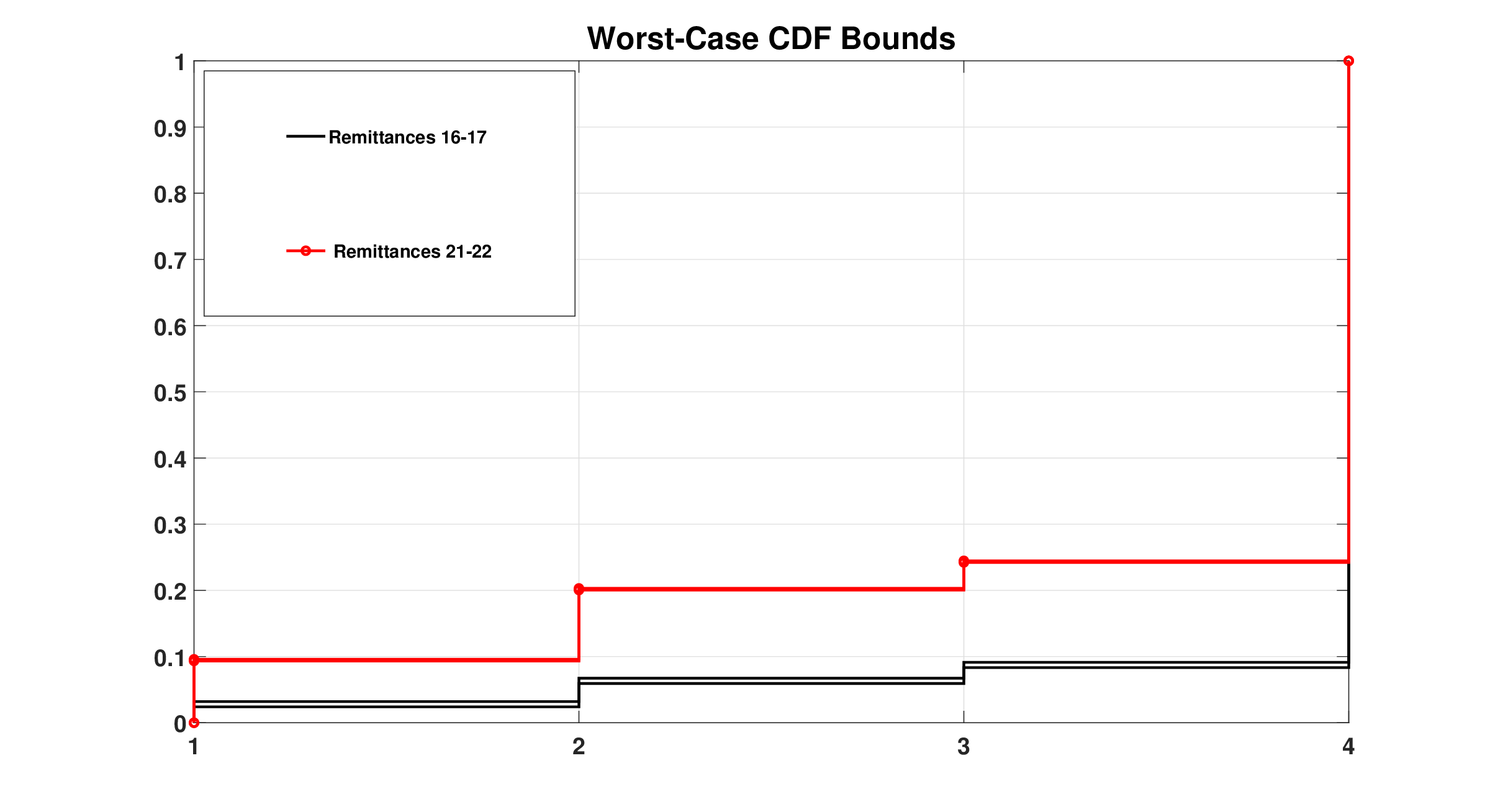}
\caption{Plug-in worst-case cumulative distribution function bounds under
unrestricted item nonresponse for remittance-receipt categories in Lebanon.}
\label{Fig-CDF WC Bounds}
\end{figure}

The respondent-population distribution therefore shifts toward both broader
and more regular remittance support, rather than only along the extensive
margin between receipt and nonreceipt. At every threshold the lower Wave 7
CDF bound exceeds the upper Wave 4 bound, so no admissible allocation of item
nonresponse can reverse the direction of the change. Under the canonical ordering, the separated worst-case CDF bounds therefore
provide graphical evidence that every admissible Wave 7 distribution is
first-order stochastically smaller than every admissible Wave 4 distribution.

The CDF comparison describes the direction and location of the empirical
change but not the amount of ordinal restructuring required to reconcile the
two distributions. The plug-in identified set for the discrepancy $D$ is
\(
[\,0.3454,\;0.3769\,].
\)
By Proposition~\ref{prop:OT_representation}, these values measure minimum
average absolute rank displacement over couplings of the corresponding latent
marginals. Thus, even under the allocation of item nonresponse that makes the
two waves as similar as possible, reconciling the empirical distributions
requires approximately 0.345 ordinal ranks of average displacement. Since the
maximum possible rank displacement is $K-1=3$, the identified interval
corresponds to approximately 11.5--12.6\% of the maximum possible average
rank displacement. This normalization should not be interpreted as the
percentage of households that changed remittance status.

Finite-sample inference additionally accounts for sampling variation in the
observable multinomial probabilities. Using the randomized Monte Carlo
calibration of Section~\ref{subsec:MC-inference} with $B=999$ replications at
the 5\% level gives the 95\%  confidence interval
\(
[\,0.2401,\;0.4768\,],
\)
with unrounded endpoints 0.2400862378 and 0.4767587822. The projection
procedure provides simultaneous 95\% coverage of the entire discrepancy
identified interval, rather than of a particular discrepancy selected from
that set. Consequently, a strictly positive lower confidence bound excludes zero from
the entire discrepancy identified interval at that confidence level. The numerically computed lower endpoint remains well
above zero, providing strong numerical evidence that positive minimum ordinal
restructuring is required even after accounting jointly for sampling
uncertainty and unrestricted item nonresponse.\footnote{The reported numerical
projections are finite-run approximations to the exact projection extrema.
Section~\ref{Section Implementation} gives the persistent-global convergence
guarantee and the numerical diagnostics used to assess the approximation.}

As a sensitivity check, repeating the projection with $B=9999$ Monte Carlo
replications and refining the lower endpoint gives
\(
[\,0.2400,\;0.4778\,].
\)
The lower endpoint is virtually unchanged, and the substantive conclusion is
therefore insensitive to a tenfold increase in the number of replications.
Appendix~\ref{appendix:MC-B-sensitivity} reports the detailed calculation and
discusses the associated power--concordance--computation trade-off.

\subsection{Conservative benchmark restructuring}

The discrepancy $D$ establishes how much minimum ordinal restructuring is
required, but not how it may be organized across categories. The
endpoint-conditioned conservative benchmarks provide this second layer of
information. Their cellwise bounds describe how much probability
mass can or must be assigned to particular reallocations among
least-displacement reconstructions associated with each discrepancy endpoint
and compatible with the Monte Carlo confidence region.

As in the discrepancy analysis, these projections are finite-run numerical
approximations to the exact projection extrema and are interpreted subject to
the computational qualification in Section~\ref{Section Implementation}.

Figure~\ref{fig:benchmark_bounds} reports the confidence bounds and summarizes
their qualitative directional implications. Rows correspond to Wave 4 source
categories and columns to Wave 7 target categories. At the 95\% confidence
level, a positive lower confidence bound implies that every benchmark in the
relevant endpoint-conditioned family assigns positive mass to that cell,
whereas a zero upper confidence bound excludes the cell from every such
benchmark. A zero lower confidence bound with a positive upper confidence
bound means that the reassignment cannot be established as necessary or
excluded at that confidence level. These statements hold simultaneously
across all cells and both discrepancy endpoints. Cellwise upper bounds are
attained separately and therefore need not be jointly attainable by a single
benchmark plan.

Movement toward more frequent receipt corresponds to cells below the main
diagonal because lower category numbers represent more frequent receipt;
movement toward less frequent receipt corresponds to cells above the diagonal.
The lower and upper discrepancy endpoints generate the same qualitative
directional pattern.

\begin{figure}[pt]
\centering

\includegraphics[width=.85\textwidth]
{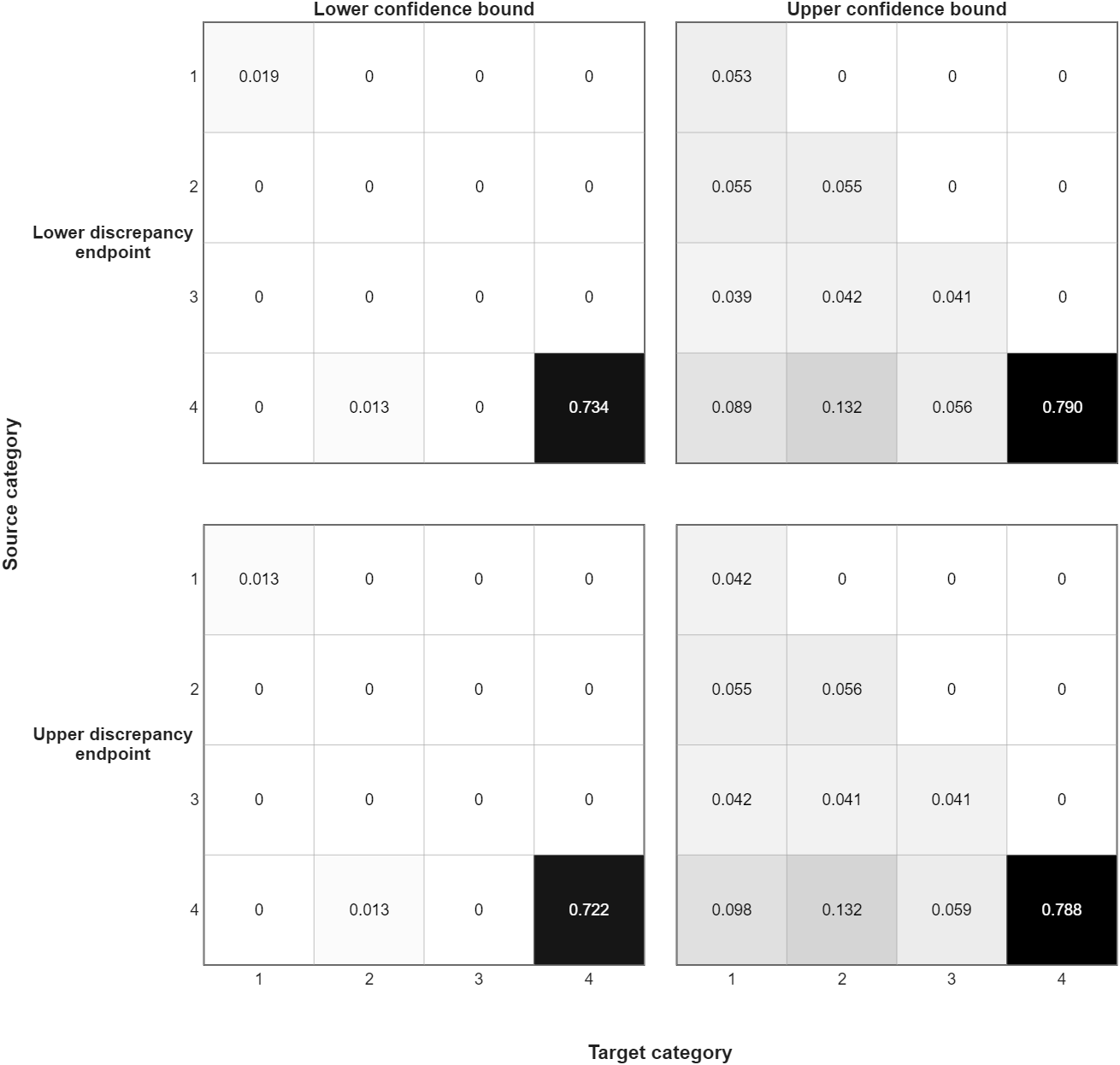}

\vspace{0.35cm}

\begin{minipage}{0.72\textwidth}
\centering
\textbf{Endpoint-conditioned benchmark directions}\\[0.15cm]

{\renewcommand{\arraystretch}{1.30}
\begin{tabular}{c|cccc}
 & \multicolumn{4}{c}{target rank $j$}\\
source $i$ & 1 & 2 & 3 & 4\\
\hline
1 & $\square$ & $\times$ & $\times$ & $\times$\\
2 & $\circ$ & $\square$ & $\times$ & $\times$\\
3 & $\circ$ & $\circ$ & $\square$ & $\times$\\
4 & $\circ$ & $\mathbf{\star}$ & $\circ$ & $\square$\\
\end{tabular}
\qquad
\begin{tabular}{cl}
$\times$ & excluded\\
$\circ$ & permitted but not required\\
$\mathbf{\star}$ & required\\
$\square$ & diagonal mass may remain
\end{tabular}
}
\end{minipage}

\caption{Cellwise confidence bounds and qualitative directional summary for
endpoint-conditioned conservative transition benchmarks. Rows denote Wave 4
source categories and columns denote Wave 7 target categories.}
\label{fig:benchmark_bounds}

\end{figure}

The symbols summarize the off-diagonal implications of the projected bounds;
the squares indicate only that positive diagonal benchmark mass can remain.
All projected upper bounds above the main diagonal are zero, so no
endpoint-conditioned least-displacement benchmark reallocates mass toward
less frequent receipt. By contrast, every off-diagonal cell below the diagonal
has a positive projected upper bound. Minimum restructuring can therefore be
organized through several movements toward more frequent receipt: from a few
times a year to monthly receipt, from annual receipt to either recurrent
category, and from nonreceipt to monthly, annual, or a few-times-a-year
receipt.

The upper bounds also show that this flexibility is quantitatively important.
Among reallocations that are permitted but not required, benchmark mass from
nonreceipt to monthly receipt can reach approximately 8.9 percentage points
at the lower discrepancy endpoint and 9.8 percentage points at the upper
endpoint; the corresponding bounds for reassignment from nonreceipt to annual
receipt are approximately 6.6 and 5.9 percentage points. Movements from receipt
a few times a year to monthly receipt, and from annual receipt to either
recurrent category, can each account for roughly 4--6 percentage points.
Because these cellwise extrema are attained separately, they should not be
added or interpreted as entries of a single benchmark transition matrix.

Cell $(4,2)$ is distinguished by its positive lower confidence bound. It
represents benchmark reassignment from nonreceipt in Wave 4 to receipt a few
times a year in Wave 7, crossing both the receipt/nonreceipt threshold and the
threshold separating annual from more recurrent receipt. Its confidence set is
\(
[\,0.0133,\;0.1322\,]
\)
at the lower discrepancy endpoint and
\(
[\,0.0130,\;0.1318\,]
\)
at the upper endpoint. The positive computed lower bounds provide strong
numerical evidence that this reassignment is required across both
endpoint-conditioned benchmark families, although its contribution to minimum
restructuring remains substantially undetermined.

The fourth row sharpens the distinction between necessity and possibility.
Reassignment from nonreceipt need not occur directly into either monthly or
annual receipt, since the projected lower bounds for cells $(4,1)$ and $(4,3)$
are zero, although both movements are feasible in some least-displacement
reconstructions. By contrast, movement from nonreceipt into receipt a few times
a year is required, while substantial benchmark mass can also remain in cell
$(4,4)$.

The qualitative pattern is stable across the two discrepancy endpoints despite
their allowing different allocations of item nonresponse and potentially
different optimal transport plans. Movements toward less frequent receipt are
excluded, several reallocations toward more frequent receipt remain feasible,
and cell $(4,2)$ retains a positive projected lower bound. Thus the detailed
allocation of benchmark mass remains nonunique, but its direction is sharply
restricted.

Taken together, the benchmark analysis distinguishes reallocations that are
required in every least-displacement reconstruction, those that are permitted
but not individually necessary, and those that are excluded. In the Lebanon
application, repeated cross-sections permit several ways of organizing the
expansion toward more frequent receipt, but rule out least-displacement
restructuring in the opposite direction and require some reassignment from
nonreceipt to recurrent receipt.

The application therefore illustrates the progression of information provided
by the framework. The plug-in CDF bounds provide graphical evidence on the direction
of distributional change; the ordinal discrepancy measures the least aggregate
restructuring it requires; and the endpoint-conditioned benchmarks describe
how that minimum can be organized and which features are unavoidable. These
objects do not identify actual household transitions, but characterize what the
repeated cross-sectional distributions themselves imply.

\section{Conclusion}
\label{Section Conclusion}

This paper develops a framework for learning the minimum ordinal
restructuring required by changes in repeated cross-sectional
distributions when individual transitions are unobserved. An axiomatic
characterization yields a threshold-weighted ordinal discrepancy whose
optimal-transport representation measures minimum aggregate displacement
under the specified threshold-additive geometry, with optimizing plans
providing conservative transition benchmarks. Under threshold neutrality
and unit normalization, this reduces to the canonical rank discrepancy,
which measures the minimum average number of ordinal thresholds crossed.
With missing outcomes, the discrepancy and endpoint-conditioned benchmark
plans are partially identified; I derive their sharp identified sets and
finite-sample-valid projection-based inference.

Applied to remittance receipt in Lebanon, the framework documents a shift
toward broader and more regular support. The computed confidence interval for
the discrepancy remains well separated from zero after accounting for missing
outcomes and sampling uncertainty, providing strong numerical evidence of
positive minimum ordinal restructuring. The conservative benchmark bounds
further provide strong numerical evidence that restructuring toward less
frequent receipt is excluded and that some reassignment from nonreceipt to
recurrent receipt is required, while allowing several alternative ways of
organizing the minimum change. They do not recover actual household
transitions, but instead characterize what restructuring the observed marginal
changes require and how that restructuring can be organized across categories.

\bibliographystyle{chicago}
\bibliography{mcgilletd}

\begin{thebibliography}{}

\bibitem[\protect\citeauthoryear{Aliprantis and Border}{Aliprantis and
  Border}{2006}]{AliprantisBorder2006}
Aliprantis, C.~D. and K.~C. Border (2006).
\newblock {\em Infinite Dimensional Analysis: A Hitchhiker's Guide\/} (2nd
  ed.).
\newblock Berlin: Springer.

\bibitem[\protect\citeauthoryear{Blair and Lacy}{Blair and
  Lacy}{1996}]{BlairLacy1996}
Blair, J. and M.~G. Lacy (1996).
\newblock Measures of variation for ordinal data as functions of the cumulative
  distribution.
\newblock {\em Perceptual and Motor Skills\/}~{\em 82\/}(2), 411--418.

\bibitem[\protect\citeauthoryear{Blair and Lacy}{Blair and
  Lacy}{2000}]{BlairLacy2000}
Blair, J. and M.~G. Lacy (2000).
\newblock Statistics of ordinal variation.
\newblock {\em Sociological Methods \& Research\/}~{\em 28\/}(3), 251--280.

\bibitem[\protect\citeauthoryear{Chernozhukov, Hong, and Tamer}{Chernozhukov
  et~al.}{2007}]{Chernozhukov-Hong-Tamer}
Chernozhukov, V., H.~Hong, and E.~Tamer (2007).
\newblock Estimation and confidence regions for parameter sets in econometric
  models.
\newblock {\em Econometrica\/}~{\em 75\/}(5), 1243--1284.

\bibitem[\protect\citeauthoryear{Cowell}{Cowell}{1985}]{Cowell1985}
Cowell, F.~A. (1985).
\newblock Measures of distributional change: An axiomatic approach.
\newblock {\em The Review of Economic Studies\/}~{\em 52\/}(1), 135--151.

\bibitem[\protect\citeauthoryear{Cowell and Flachaire}{Cowell and
  Flachaire}{2018}]{Cowell_Flachaire_2018}
Cowell, F.~A. and E.~Flachaire (2018).
\newblock Measuring mobility.
\newblock {\em Quantitative Economics\/}~{\em 9\/}(2), 865--901.

\bibitem[\protect\citeauthoryear{Cuhadaroglu}{Cuhadaroglu}{2023}]{Cuhadaroglu2023}
Cuhadaroglu, T. (2023).
\newblock Evaluating ordinal inequalities between groups.
\newblock {\em The Journal of Economic Inequality\/}~{\em 21\/}(1), 219--231.

\bibitem[\protect\citeauthoryear{D'Agostino and Dardanoni}{D'Agostino and
  Dardanoni}{2009}]{DAgostinoDardanoni2009}
D'Agostino, M. and V.~Dardanoni (2009).
\newblock The measurement of rank mobility.
\newblock {\em Journal of Economic Theory\/}~{\em 144\/}(4), 1783--1803.

\bibitem[\protect\citeauthoryear{Dall'Aglio}{Dall'Aglio}{1956}]{DallAglio1956}
Dall'Aglio, G. (1956).
\newblock Sugli estremi dei momenti delle funzioni di ripartizione doppia.
\newblock {\em Annali della Scuola Normale Superiore di Pisa, Classe di
  Scienze\/}~{\em 10\/}(1--2), 35--74.

\bibitem[\protect\citeauthoryear{Dang and Lanjouw}{Dang and
  Lanjouw}{2023}]{DangLanjouw2023}
Dang, H.-A.~H. and P.~F. Lanjouw (2023).
\newblock Measuring poverty dynamics with synthetic panels based on
  cross-sections.
\newblock {\em World Bank Economic Review\/}~{\em 37\/}(1), 1--26.

\bibitem[\protect\citeauthoryear{Dang, Lanjouw, Luoto, and McKenzie}{Dang
  et~al.}{2014}]{DangEtAl2014}
Dang, H.-A.~H., P.~F. Lanjouw, J.~Luoto, and D.~McKenzie (2014).
\newblock Using repeated cross-sections to explore movements into and out of
  poverty.
\newblock {\em Journal of Development Economics\/}~{\em 107}, 112--128.

\bibitem[\protect\citeauthoryear{Deaton}{Deaton}{1985}]{Deaton1985}
Deaton, A. (1985).
\newblock Panel data from time series of cross-sections.
\newblock {\em Journal of Econometrics\/}~{\em 30\/}(1-2), 109--126.

\bibitem[\protect\citeauthoryear{Dufour}{Dufour}{2006}]{DUFOUR2006}
Dufour, J.-M. (2006).
\newblock Monte {C}arlo tests with nuisance parameters: A general approach to
  finite-sample inference and nonstandard asymptotics.
\newblock {\em Journal of Econometrics\/}~{\em 133\/}(2), 443--477.

\bibitem[\protect\citeauthoryear{Fakih, Makdissi, Marrouch, Tabri, and
  Yazbeck}{Fakih et~al.}{2022}]{Tabri2021}
Fakih, A., P.~Makdissi, W.~Marrouch, R.~V. Tabri, and M.~Yazbeck (2022).
\newblock A stochastic dominance test under survey nonresponse with an
  application to comparing trust levels in {L}ebanese public institutions.
\newblock {\em Journal of Econometrics\/}~{\em 228\/}(2), 342--358.

\bibitem[\protect\citeauthoryear{Fang and Santos}{Fang and
  Santos}{2018}]{Fang-Santos-2018}
Fang, Z. and A.~Santos (2018, 09).
\newblock {Inference on Directionally Differentiable Functions}.
\newblock {\em The Review of Economic Studies\/}~{\em 86\/}(1), 377--412.

\bibitem[\protect\citeauthoryear{Firpo, Fortin, and Lemieux}{Firpo
  et~al.}{2018}]{FirpoFortinLemieux2018}
Firpo, S.~P., N.~M. Fortin, and T.~Lemieux (2018).
\newblock Decomposing wage distributions using recentered influence function
  regressions.
\newblock {\em Econometrics\/}~{\em 6\/}(2), 28.

\bibitem[\protect\citeauthoryear{Fr{\'e}chet}{Fr{\'e}chet}{1935}]{Frechet1935}
Fr{\'e}chet, M. (1935).
\newblock G{\'e}n{\'e}ralisations du th{\'e}or{\`e}me des probabilit{\'e}s
  totales.
\newblock {\em Fundamenta Mathematicae\/}~{\em 25}, 379--387.

\bibitem[\protect\citeauthoryear{Fr{\'e}chet}{Fr{\'e}chet}{1951}]{Frechet1951}
Fr{\'e}chet, M. (1951).
\newblock Sur les tableaux de corr{\'e}lation dont les marges sont donn{\'e}es.
\newblock {\em Annales de l'Universit{\'e} de Lyon. Section A: Sciences
  math{\'e}matiques et astronomie\/}~{\em 9}, 53--77.

\bibitem[\protect\citeauthoryear{Galichon and Henry}{Galichon and
  Henry}{2026}]{GalichonHenry2026}
Galichon, A. and M.~Henry (2026).
\newblock An econometrician's guide to optimal transport.
\newblock Technical Report arXiv:2604.04227, arXiv.
\newblock Working paper.

\bibitem[\protect\citeauthoryear{Gini}{Gini}{1914}]{Gini1914}
Gini, C. (1914).
\newblock Di una misura della dissomiglianza tra due gruppi di quantit{\`a} e
  delle sue applicazioni allo studio delle relazioni statistiche.
\newblock {\em Atti del Reale Istituto Veneto di Scienze, Lettere ed
  Arti\/}~{\em 74\/}(2), 185--213.

\bibitem[\protect\citeauthoryear{Goff and Mbakop}{Goff and
  Mbakop}{2025}]{GoffMbakop2025}
Goff, L. and E.~Mbakop (2025).
\newblock Testing the solvability of systems of linear inequalities.
\newblock {\em arXiv preprint arXiv:2506.06776\/}.

\bibitem[\protect\citeauthoryear{H{\'e}rault and Jenkins}{H{\'e}rault and
  Jenkins}{2019}]{HeraultJenkins2019}
H{\'e}rault, N. and S.~P. Jenkins (2019).
\newblock How valid are synthetic panel estimates of poverty dynamics?
\newblock {\em Journal of Economic Inequality\/}~{\em 17\/}(1), 51--76.

\bibitem[\protect\citeauthoryear{Jenkins}{Jenkins}{2020}]{Jenkins}
Jenkins, S.~P. (2020, March).
\newblock Comparing distributions of ordinal data.
\newblock IZA Discussion Paper No.13057.

\bibitem[\protect\citeauthoryear{Kaido, Molinari, and Stoye}{Kaido
  et~al.}{2019}]{KaidoMolinariStoye2019}
Kaido, H., F.~Molinari, and J.~Stoye (2019).
\newblock Confidence intervals for projections of partially identified
  parameters.
\newblock {\em Econometrica\/}~{\em 87\/}(4), 1397--1432.

\bibitem[\protect\citeauthoryear{K{\k e}ska}{K{\k e}ska}{2017}]{Keska2017}
K{\k e}ska, A. (2017).
\newblock Axiomatic determination of a class of ordinal variation measures.
\newblock {\em Studies in Logic, Grammar and Rhetoric\/}~{\em 50\/}(1), 45--65.

\bibitem[\protect\citeauthoryear{Manski}{Manski}{2005}]{manski2005}
Manski, C.~F. (2005).
\newblock Partial identification with missing data: concepts and findings.
\newblock {\em International Journal of Approximate Reasoning\/}~{\em
  39\/}(2-3), 151--165.

\bibitem[\protect\citeauthoryear{McGee}{McGee}{2026}]{McGee2026}
McGee, R. (2026, May).
\newblock Sufficient statistics for economic mobility: When do measures agree?
\newblock IFS Working Paper W26/38, Institute for Fiscal Studies.

\bibitem[\protect\citeauthoryear{Melleray, Petrov, and Vershik}{Melleray
  et~al.}{2008}]{MellerayPetrovVershik2008}
Melleray, J., F.~V. Petrov, and A.~M. Vershik (2008).
\newblock Linearly rigid metric spaces and the embedding problem.
\newblock {\em Fundamenta Mathematicae\/}~{\em 199\/}(2), 177--194.

\bibitem[\protect\citeauthoryear{Rachev, Klebanov, Stoyanov, and
  Fabozzi}{Rachev et~al.}{2013}]{RachevKlebanovStoyanovFabozzi2013}
Rachev, S.~T., L.~B. Klebanov, S.~V. Stoyanov, and F.~J. Fabozzi (2013).
\newblock {\em The Methods of Distances in the Theory of Probability and
  Statistics}.
\newblock New York: Springer.

\bibitem[\protect\citeauthoryear{Ray and Genicot}{Ray and
  Genicot}{2023}]{RayGenicot2023}
Ray, D. and G.~Genicot (2023).
\newblock Measuring upward mobility.
\newblock {\em American Economic Review\/}~{\em 113\/}(11), 3044--3089.

\bibitem[\protect\citeauthoryear{Romano and Shaikh}{Romano and
  Shaikh}{2010}]{RomanoShaikh2010}
Romano, J.~P. and A.~M. Shaikh (2010).
\newblock Inference for the identified set in partially identified econometric
  models.
\newblock {\em Econometrica\/}~{\em 78\/}(1), 169--211.

\bibitem[\protect\citeauthoryear{R{\"u}schendorf}{R{\"u}schendorf}{1991}]{Ruschendorf1991}
R{\"u}schendorf, L. (1991).
\newblock Fr{\'e}chet bounds and their applications.
\newblock In G.~Dall'Aglio, S.~Kotz, and G.~Salinetti (Eds.), {\em Advances in
  Probability Distributions with Given Marginals}, pp.\  151--187. Dordrecht:
  Kluwer Academic Publishers.

\bibitem[\protect\citeauthoryear{Salvemini}{Salvemini}{1943}]{Salvemini1943}
Salvemini, T. (1943).
\newblock Sul calcolo degli indici di concordanza tra due caratteri
  quantitativi.
\newblock In {\em Atti della VI Riunione della Societ{\`a} Italiana di
  Statistica}, Roma.

\bibitem[\protect\citeauthoryear{{United Nations Development
  Programme}}{{United Nations Development
  Programme}}{2023}]{UNDP2023Remittances}
{United Nations Development Programme} (2023, May).
\newblock The increasing role and importance of remittances in {L}ebanon.
\newblock Technical report, United Nations Development Programme, {L}ebanon.

\bibitem[\protect\citeauthoryear{Vallender}{Vallender}{1974}]{Vallender}
Vallender, S.~S. (1974).
\newblock Calculation of the {W}asserstein distance between probability
  distributions on the line.
\newblock {\em Theory of Probability \& Its Applications\/}~{\em 18\/}(4),
  784--786.

\bibitem[\protect\citeauthoryear{Villani}{Villani}{2009}]{villani2009optimal}
Villani, C. (2009).
\newblock {\em Optimal Transport: Old and New}.
\newblock Grundlehren der mathematischen Wissenschaften. Springer Berlin
  Heidelberg.

\bibitem[\protect\citeauthoryear{Zolotarev}{Zolotarev}{1984}]{Zolotarev1983}
Zolotarev, V.~M. (1984).
\newblock Probability metrics.
\newblock {\em Theory of Probability and Its Applications\/}~{\em 28\/}(1),
  278--302.

\end{thebibliography}

\appendix
\section{Proofs of Results}

\subsection{Proposition~\ref{prop:threshold_metric}}\label{Proof - prop -- threshold metric}
\begin{proof}
Suppose first that $d$ is threshold additive. For $i<j$, repeated
application of threshold additivity gives
\(
d(\xi_i,\xi_j)
=
\sum_{k=i}^{j-1}d(\xi_k,\xi_{k+1})
=
\sum_{k=i}^{j-1}w_k.
\)
Symmetry gives the stated representation for $i>j$.
Conversely, if the representation~(\ref{eq - Threshhold add metric}) holds, then for
$i<j<\ell$,
\(
d(\xi_i,\xi_\ell)
=
\sum_{k=i}^{\ell-1}w_k
=
d(\xi_i,\xi_j)+d(\xi_j,\xi_\ell),
\)
so $d$ is threshold additive.
\end{proof}

\subsection{Theorem~\ref{thm:ordinal_metric}}\label{Proof - Thm - ordinal-metric charac}
\begin{proof}
Recall that every $\sigma\in\mathcal M_0^K$ admits the decomposition
\(
\sigma=\sum_{k=1}^{K-1} C_k(\sigma)b_k,
\)
where $\mathcal{T}(b_k)=\{k\}$. Hence the nonzero summands
$C_k(\sigma)b_k$ have pairwise disjoint threshold supports. Repeated
application of threshold separability therefore gives the form
\(
N(\sigma)
=
\sum_{k=1}^{K-1}N\!\left(C_k(\sigma)b_k\right).
\)
By absolute homogeneity of the norm and local ground-metric consistency,
\(
N\!\left(C_k(\sigma)b_k\right)
=
|C_k(\sigma)|N(b_k)
=
w_k|C_k(\sigma)|.
\)
Thus
\(
N(\sigma)
=
\sum_{k=1}^{K-1}w_k|C_k(\sigma)|.
\)
Because
\(
C_k(\mu-\nu)=F_\mu(k)-F_\nu(k),
\)
the stated expression for $\Delta_N$ follows.

For $i<j$,
\(
C_k(\delta_i-\delta_j)
=
\mathbf 1\{i\le k<j\}.
\)
Hence
\(
\Delta_N(\delta_i,\delta_j)
=
\sum_{k=i}^{j-1}w_k
=
d_w(\xi_i,\xi_j),
\)
where the final equality follows from Proposition~\ref{prop:threshold_metric}. Symmetry gives
the result for $j<i$, while the case $i=j$ is immediate.

For the converse, consider the probability metric
\(
\Delta(\mu,\nu)
=
\sum_{k=1}^{K-1}w_k|F_\mu(k)-F_\nu(k)|.
\)
Since
\(
F_\mu(k)-F_\nu(k)=C_k(\mu-\nu)\;\forall k,
\)
it can be written as
\(
\Delta(\mu,\nu)
=
\sum_{k=1}^{K-1}w_k|C_k(\mu-\nu)|.
\)
Now define the cumulative-coordinate map
\(
C:\mathcal M_0^K\to\mathbb R^{K-1},
\;\text{with}\;
C(\sigma):=
\bigl(C_1(\sigma),\ldots,C_{K-1}(\sigma)\bigr).
\)
By the adjacent-threshold basis representation of elements of
$\mathcal M_0^K$, $C$ is a linear bijection. Define
\(
N(\sigma)
:=
\sum_{k=1}^{K-1}w_k|C_k(\sigma)|,
\;\text{for }
\sigma\in\mathcal M_0^K.
\)
Because $C$ is linear and injective and $w_k>0$ for every $k$,
$N(\sigma)=0$ if and only if $\sigma=0$, while absolute homogeneity
and the triangle inequality follow from the corresponding properties
of the weighted $\ell_1$ norm on $\mathbb R^{K-1}$. Thus $N$ is a norm
on $\mathcal M_0^K$.

Moreover,
\(
N(\mu-\nu)=\Delta(\mu,\nu),
\)
so $\Delta$ is induced by $N$. Its weighted $\ell_1$ form implies
threshold separability, and
\(
N(b_k)=w_k,
\; k=1,\ldots,K-1,
\)
so $N$ is locally consistent with $d_w$.
\end{proof}

\subsection{Proposition~\ref{prop:OT_representation}}\label{Proof - Prop}
\begin{proof}
Let $r_w(\xi_i)=a_i$ denote the isometric representation introduced above.
Since the ground metric is given by
\(
d_w(\xi_i,\xi_j)=|a_i-a_j|,
\)
the transport problem on $(\Xi,d_w)$ is equivalent to the one-dimensional
transport problem on the support $a_1<\cdots<a_K$. By the characterization of the Wasserstein--$1$ distance as the $L_1$ distance between cumulative distribution functions (see, e.g.,~\citealp{Vallender}),
\[
\min_{\pi\in\Pi(\mu,\nu)}
\sum_{i=1}^K\sum_{j=1}^K
d_w(\xi_i,\xi_j)\pi_{ij}
=
\sum_{k=1}^{K-1}
(a_{k+1}-a_k)
|F_\mu(k)-F_\nu(k)|.
\]
Because $a_{k+1}-a_k=w_k$, the right-hand side equals
\(
\sum_{k=1}^{K-1}
w_k|F_\mu(k)-F_\nu(k)|
=
D_w(\mu,\nu)
\)
by Theorem~\ref{thm:ordinal_metric}. The canonical representation follows
by setting $w_k=1$ for every $k$.
\end{proof}

\subsection{Theorem~\ref{thm:discrepancy_bounds}}\label{Proof - Thm PI}
\begin{proof}
Since $\mathcal M_\mu\times\mathcal M_\nu$ is the sharp joint
identified set for $(\mu,\nu)$, the sharp identified set for
$D_w(\mu,\nu)$ is
\(
\mathcal D_w
=
\left\{
D_w(\gamma,\eta):
(\gamma,\eta)\in
\mathcal M_\mu\times\mathcal M_\nu
\right\}.
\)
The sets $\mathcal M_\mu$ and $\mathcal M_\nu$ are nonempty compact
convex polytopes, so their product is nonempty, compact, and connected.
Since $(\gamma,\eta)\mapsto D_w(\gamma,\eta)$ is continuous,
$\mathcal D_w$ is a compact and connected subset of $\mathbb R$ and
therefore a closed interval. Compactness and continuity imply that both
endpoints are attained, giving
\(
\mathcal D_w=[\underline D_w,\overline D_w].
\)

For the lower endpoint, for each $k$,
\(
|C_k(\gamma-\eta)|
=
\min_{t_k\in\mathbb R}
\left\{
t_k:
-t_k\leq C_k(\gamma-\eta)\leq t_k
\right\}.
\)
Since
\(
D_w(\gamma,\eta)
=
\sum_{k=1}^{K-1}
w_k|C_k(\gamma-\eta)|,
\)
applying this representation simultaneously over
$k=1,\ldots,K-1$ and minimizing over
$(\gamma,\eta)\in\mathcal M_\mu\times\mathcal M_\nu$
gives the stated linear program.

For the upper endpoint, for every
$z=(z_1,\ldots,z_{K-1})\in\mathbb R^{K-1}$, I can write
\(
\sum_{k=1}^{K-1}w_k|z_k|
=
\max_{s\in\mathcal S_K}
\sum_{k=1}^{K-1}w_k s_k z_k.
\)
Applying this identity with
$z_k=C_k(\gamma-\eta)$ gives
\[
\overline D_w
=
\max_{s\in\mathcal S_K}
\max_{\gamma\in\mathcal M_\mu,\,
      \eta\in\mathcal M_\nu}
\sum_{k=1}^{K-1}
w_k s_k C_k(\gamma-\eta)
=
\max_{s\in\mathcal S_K}D_{w,s},
\]
where the interchange of the two maxima is because $\mathcal S_K$ is finite.
For each fixed $s$, the objective is linear in $(\gamma,\eta)$ and the
feasible set is a polytope, so $D_{w,s}$ is the value of a linear
program. Since $\mathcal S_K$ contains $2^{K-1}$ sign vectors, $\overline D_w$ is obtained as the maximum of the optimal values of $2^{K-1}$ linear programs.
\end{proof}

\subsection{Theorem~\ref{thm - PI coupling endpoint bounds}}\label{Proof -Thm 2}
\begin{proof}
For $e\in\{L,U\}$, write
\(
D_{w,L}:=\underline D_w,
\;
D_{w,U}:=\overline D_w,
\)
and define the auxiliary set
\(
\mathcal C_e
:=
\left\{
(\gamma,\eta,\pi):
\gamma\in\mathcal M_\mu,\;
\eta\in\mathcal M_\nu,\;
D_w(\gamma,\eta)=D_{w,e},\;
\pi\in\Pi(\gamma,\eta),\;
c_w(\pi)=D_{w,e}
\right\},
\)
where \\
\(
c_w(\pi)
:=
\sum_{r=1}^K\sum_{s=1}^K d_w(\xi_r,\xi_s)\pi_{rs}.
\)

By Theorem~\ref{thm:discrepancy_bounds}, for each $e\in\{L,U\}$ there exists
$(\gamma^e,\eta^e)\in\mathcal M_\mu\times\mathcal M_\nu$
such that
\(
D_w(\gamma^e,\eta^e)=D_{w,e}.
\)
The transport polytope $\Pi(\gamma^e,\eta^e)$ is nonempty and compact,
and $c_w$ is continuous. Hence there exists
$\pi^e\in\Pi(\gamma^e,\eta^e)$ attaining its minimum. By
Proposition~\ref{prop:OT_representation},
\(
c_w(\pi^e)
=
D_w(\gamma^e,\eta^e)
=
D_{w,e},
\)
so $(\gamma^e,\eta^e,\pi^e)\in\mathcal C_e$. Thus
$\mathcal C_e$ is nonempty.

We next show that $\mathcal C_e$ is compact. Since every transport
plan belongs to $\Delta^{K^2}$,
\(
\mathcal C_e
\subseteq
\mathcal M_\mu\times\mathcal M_\nu\times\Delta^{K^2},
\)
and the latter set is compact. Let
\(
(\gamma^m,\eta^m,\pi^m)\in\mathcal C_e,
\;
(\gamma^m,\eta^m,\pi^m)
\to
(\gamma,\eta,\pi).
\)
Because $\mathcal M_\mu$ and $\mathcal M_\nu$ are closed,
$\gamma\in\mathcal M_\mu$ and $\eta\in\mathcal M_\nu$.
Passing to the limit in the row- and column-sum equalities gives
$\pi\in\Pi(\gamma,\eta)$. Continuity of $D_w$ and $c_w$ gives
\(
D_w(\gamma,\eta)
=
\lim_{m\to\infty}D_w(\gamma^m,\eta^m)
=
D_{w,e}
\)
and
\(
c_w(\pi)
=
\lim_{m\to\infty}c_w(\pi^m)
=
D_{w,e}.
\)
Hence $(\gamma,\eta,\pi)\in\mathcal C_e$, so
$\mathcal C_e$ is closed and therefore compact.

By Proposition~\ref{prop:OT_representation} and the definition of $\Pi_{w,e}^\ast$,
\(
\Pi_{w,e}^\ast
=
\left\{
\pi:
(\gamma,\eta,\pi)\in\mathcal C_e
\text{ for some }(\gamma,\eta)
\right\}.
\)
Thus $\Pi_{w,e}^\ast$ is the projection of the nonempty compact set
$\mathcal C_e$ onto the $\pi$-coordinate. Since coordinate projection
is continuous, $\Pi_{w,e}^\ast$ is nonempty and compact.

Finally, for every $(i,j)$ the coordinate map
$\pi\mapsto\pi_{ij}$ is continuous. It therefore attains its minimum
and maximum on each compact set $\Pi_{w,e}^\ast$. Hence
\(
\underline{\pi}^{\,e}_{w,ij}
=
\min_{\pi\in\Pi_{w,e}^\ast}\pi_{ij},
\;
\overline{\pi}^{\,e}_{w,ij}
=
\max_{\pi\in\Pi_{w,e}^\ast}\pi_{ij},
\)
which proves the result.
\end{proof}

\subsection{Proposition~\ref{Prop-ELR}}\label{Proof - Prop ELR}
\begin{proof}
By independence of the two repeated cross-sections, the joint empirical
likelihood factorizes as
\(
L(\vartheta)
=
L_X(\mathbf q_X)L_Y(\mathbf q_Y),
\)
where $L_X(\mathbf q_X)=\prod_{j=0}^{K}q_{Xj}^{N_{Xj}}$ and $L_Y(\mathbf q_Y)=\prod_{j=0}^{K}q_{Yj}^{N_{Yj}}$. The unrestricted maximum empirical likelihood estimators are the empirical
cell-probability vectors
\(
\widehat{\mathbf q}_X
=
\left(
\frac{N_{X0}}{n_X},\ldots,\frac{N_{XK}}{n_X}
\right)\;\text{and}\; \widehat{\mathbf q}_Y
=
\left(
\frac{N_{Y0}}{n_Y},\ldots,\frac{N_{YK}}{n_Y}
\right).
\)
Consequently, $R_n(\vartheta)=\frac{L_X(\mathbf q_X)}{L_X(\widehat{\mathbf q}_X)}\frac{L_Y(\mathbf q_Y)}{L_Y(\widehat{\mathbf q}_Y)}$. Taking the logarithm of $R_n(\vartheta)$ yields
$
\log R_n(\vartheta)
=
\sum_{j=0}^{K}
N_{Xj}
\log\left(
\frac{q_{Xj}}{\widehat q_{Xj}}
\right)
+
\sum_{j=0}^{K}
N_{Yj}
\log\left(
\frac{q_{Yj}}{\widehat q_{Yj}}
\right).
$
Since $N_{Xj}=n_X\widehat q_{Xj}$ and $N_{Yj}=n_Y\widehat q_{Yj}$, it follows that
\begin{align*}
-2\log R_n(\vartheta)
&=
2n_X
\sum_{j=0}^{K}
\widehat q_{Xj}
\log\left(
\frac{\widehat q_{Xj}}{q_{Xj}}
\right)
+2n_Y
\sum_{j=0}^{K}
\widehat q_{Yj}
\log\left(
\frac{\widehat q_{Yj}}{q_{Yj}}
\right)
\\
&=
2n_X
D_{\mathrm{KL}}
\left(
\widehat{\mathbf q}_X\,\Vert\,\mathbf q_X
\right)
+
2n_Y
D_{\mathrm{KL}}
\left(
\widehat{\mathbf q}_Y\,\Vert\,\mathbf q_Y
\right).
\end{align*}
The stated convention handles cells for which the empirical frequency is
zero.
\end{proof}


\subsection{Proposition~\ref{prop:persistent-global-convergence}}
\label{proof:global-search-proposition}
\begin{proof}
Fix a realization of $\mathcal G_{n,B}$ satisfying Assumption~2.
Consider first a minimization problem and let
\(
h^\star
:=
\inf_{\vartheta\in\mathcal C^{MC}_{n,B}} h(\vartheta).
\)
The reported objectives are bounded, so $h^\star$ is finite. The
best feasible value $\widehat h_N$ is nonincreasing in $N$ and satisfies
\(
\widehat h_N\ge h^\star
\;\text{for every }N.
\)

For each $r\in\mathbb N$, let
\(
A_{1/r}(h)
=
\left\{
\vartheta\in\mathcal C^{MC}_{n,B}:
h(\vartheta)<h^\star+\frac1r
\right\}.
\)
By Assumption~2, observe that
$\Gamma(A_{1/r}(h))>0$, and Lemma~\ref{lem:persistent-hit} implies
\(
\mathbb P_{\mathrm{alg}}
\left(
\vartheta_m^{\mathrm{prop}}\in A_{1/r}(h)
\ \text{infinitely often}
\right)
=1.
\)
Let $E_r$ denote this probability-one event and set
\(
E:=\bigcap_{r=1}^{\infty}E_r.
\)
Then $\mathbb P_{\mathrm{alg}}(E)=1$.

On $E$, for every $r$ there exists a finite $N_r$ such that a feasible
candidate in $A_{1/r}(h)$ has been encountered by proposal $N_r$.
Hence, for every $N\ge N_r$,
\(
h^\star
\le
\widehat h_N
<
h^\star+\frac1r.
\)
Since this holds for every $r\in\mathbb N$,
\(
\widehat h_N\longrightarrow h^\star
\)
on $E$. Therefore
\(
\mathbb P_{\mathrm{alg}}
\left(
\lim_{N\to\infty}\widehat h_N
=
\inf_{\vartheta\in\mathcal C^{MC}_{n,B}}h(\vartheta)
\right)
=1.
\)

For a maximization problem, let
\(
h^\star
:=
\sup_{\vartheta\in\mathcal C^{MC}_{n,B}} h(\vartheta).
\)
Then $\widehat h_N$ is nondecreasing and bounded above by $h^\star$.
Now applying the same argument to the event
\(
A_{1/r}(h)
=
\left\{
\vartheta\in\mathcal C^{MC}_{n,B}:
h(\vartheta)>h^\star-\frac1r
\right\}
\)
gives
\(
\mathbb P_{\mathrm{alg}}
\left(
\lim_{N\to\infty}\widehat h_N
=
\sup_{\vartheta\in\mathcal C^{MC}_{n,B}}h(\vartheta)
\right)
=1.
\)
\end{proof}

\section{Regularity of the Optimization Correspondences}
\label{appendix:projection-measurability}

This section collects the measurability results used in
Section~\ref{Section Inference}. The projection arguments themselves
are deterministic: coverage of the primitive observable-data parameter
implies coverage of the corresponding optimization values and solution
sets. The results below ensure that the value functions and
correspondences appearing in those probability statements are
measurable.

\subsection{A general measurable optimization result}

Let $\Theta$ be equipped with its Borel $\sigma$-field, let $E$ be a
separable metric space, and let  
\(
\mathcal Z:\Theta\rightrightarrows E
\)
be a correspondence. Consider the parameterized optimization problem
\[
v(\vartheta)
=
\inf_{z\in\mathcal Z(\vartheta)}f(z),
\qquad
S(\vartheta)
=
\argmin_{z\in\mathcal Z(\vartheta)}f(z).
\]

\begin{proposition}
\label{prop:supp-measurable-maximum}
Suppose Assumption~\ref{ass:projection} holds. Then:
\begin{enumerate}[(i)]
\item
the value function $v:\Theta\rightarrow\mathbb R$ is measurable;
\item
the solution correspondence
$S:\Theta\rightrightarrows E$ is weakly measurable with nonempty
compact values.
\end{enumerate}
\end{proposition}

\begin{proof}
Define
\(
g(\vartheta,z)=-f(z).
\)
For each $z\in E$, the mapping
$\vartheta\mapsto g(\vartheta,z)$ is measurable, since it is constant
in $\vartheta$, while for each $\vartheta\in\Theta$, the mapping
$z\mapsto g(\vartheta,z)$ is continuous. Hence $g$ is a
Carath\'eodory function.

By Assumption~\ref{ass:projection},
$\mathcal Z$ is weakly measurable with nonempty compact values.
The Measurable Maximum Theorem (\citealp{AliprantisBorder2006}, Theorem~18.19) therefore implies that
\(
\vartheta
\longmapsto
\max_{z\in\mathcal Z(\vartheta)}g(\vartheta,z)
\)
is measurable and that its maximizing correspondence is weakly
measurable with nonempty compact values. Since
\(
v(\vartheta)
=
-
\max_{z\in\mathcal Z(\vartheta)}g(\vartheta,z),
\)
and the maximizing correspondence for $g$ coincides with
$S(\vartheta)$, the result follows.
\end{proof}

\subsection{The latent-marginal correspondence}

Recall that an observable probability vector is
\(
\mathbf q
=
(q_0,q_1,\ldots,q_K)
\in\Delta^{K+1},
\)
where $q_0$ is the nonresponse probability. Define
\(
\mathbf q_+
:=
(q_1,\ldots,q_K).
\)
The identified set for the corresponding latent marginal distribution
is 
\(
\mathcal M(\mathbf q)
=
\left\{
\mu\in\Delta^K:
q_k\leq\mu_k\leq q_k+q_0,
\; k=1,\ldots,K
\right\}.
\)

The following alternative representation will be useful:
\[
\mathcal M(\mathbf q)
=
\mathbf q_+
+
q_0\Delta^K
=
\left\{
\mathbf q_+ + q_0 r:
r\in\Delta^K
\right\}.
\]
Indeed, if $\mu\in\mathcal M(\mathbf q)$ and $q_0>0$, then
\(
r_k
=
\frac{\mu_k-q_k}{q_0},
\; k=1,\ldots,K,
\)
defines an element of $\Delta^K$. When $q_0=0$,
$\mathcal M(\mathbf q)=\{\mathbf q_+\}$.

\begin{lemma}
\label{lem:latent-correspondence-continuous}
The correspondence
\(
\mathcal M:
\Delta^{K+1}
\rightrightarrows
\Delta^K
\)
is continuous and has nonempty compact values.
\end{lemma}

\begin{proof}
Nonemptiness and compactness follow immediately from
\(
\mathcal M(\mathbf q)
=
\mathbf q_+ + q_0\Delta^K
\)
and compactness of $\Delta^K$.

I establish upper and lower hemicontinuity sequentially. Let
$\mathbf q^m\rightarrow\mathbf q$ and take
$\mu^m\in\mathcal M(\mathbf q^m)$ such that
$\mu^m\rightarrow\mu$. For every $m$, there exists
$r^m\in\Delta^K$ satisfying
\(
\mu^m
=
\mathbf q_+^m+q_0^m r^m.
\)
Since $\Delta^K$ is compact, every subsequence of $\{r^m\}$ has a
further convergent subsequence with limit $r\in\Delta^K$. Along such
a subsequence,
\(
\mu
=
\mathbf q_+ + q_0 r
\in
\mathcal M(\mathbf q).
\)
Hence $\mathcal M$ has a closed graph and, because its range is
contained in the compact set $\Delta^K$, it is upper
hemicontinuous.

For lower hemicontinuity, fix
$\mu\in\mathcal M(\mathbf q)$. There exists $r\in\Delta^K$ such that
\(
\mu=\mathbf q_+ + q_0r.
\)
For any sequence $\mathbf q^m\rightarrow\mathbf q$, define
\(
\mu^m
=
\mathbf q_+^m+q_0^m r.
\)
Then
$\mu^m\in\mathcal M(\mathbf q^m)$ for every $m$ and
$\mu^m\rightarrow\mu$. Therefore $\mathcal M$ is lower
hemicontinuous. Combining the two properties proves continuity.
\end{proof}

For
\(
\vartheta=(\mathbf q_X,\mathbf q_Y)\in\Theta,
\)
define
\(
\mathcal Z(\vartheta)
=
\mathcal M(\mathbf q_X)
\times
\mathcal M(\mathbf q_Y).
\)

\begin{corollary}
\label{cor:latent-product-continuous}
The correspondence
\(
\mathcal Z:
\Theta
\rightrightarrows
\Delta^K\times\Delta^K
\)
is continuous with nonempty compact values.
\end{corollary}

\begin{proof}
The result follows immediately from
Lemma~\ref{lem:latent-correspondence-continuous} and continuity of
finite Cartesian products of compact-valued correspondences.
\end{proof}

\subsection{Identification-bound solution correspondences}

Recall that the ordinal discrepancy is a function
\(
D_w:
\Delta^K\times\Delta^K
\rightarrow\mathbb R.
\)
Because it is continuous, define
\(
\underline D_w(\vartheta)
=
\min_{(\mu,\nu)\in\mathcal Z(\vartheta)}
D_w(\mu,\nu),
\)
and
\(
\overline D_w(\vartheta)
=
\max_{(\mu,\nu)\in\mathcal Z(\vartheta)}
D_w(\mu,\nu).
\)
Let
\(
S_{\underline D_w}(\vartheta)
=
\argmin_{(\mu,\nu)\in\mathcal Z(\vartheta)}
D_w(\mu,\nu)
\)
and
\(
S_{\overline D_w}(\vartheta)
=
\argmax_{(\mu,\nu)\in\mathcal Z(\vartheta)}
D_w(\mu,\nu).
\)

\begin{proposition}
\label{prop:endpoint-correspondences}
The functions
$\underline D_w,\overline D_w:\Theta\rightarrow\mathbb R$
are continuous. Moreover,
\(
S_{\underline D_w},
S_{\overline D_w}:
\Theta
\rightrightarrows
\Delta^K\times\Delta^K
\)
have nonempty compact values and are upper hemicontinuous and weakly
measurable.
\end{proposition}

\begin{proof}
By Corollary~\ref{cor:latent-product-continuous},
$\mathcal Z$ is continuous with nonempty compact values, and
$D_w$ is continuous. Berge's Maximum Theorem therefore implies that
$\underline D_w$ and $\overline D_w$ are continuous and that the
correspondences $S_{\underline D_w}$ and $S_{\overline D_w}$ are
nonempty, compact valued, and upper hemicontinuous.

Since the domain and range are finite-dimensional metric spaces and
the solution correspondences are compact valued and upper
hemicontinuous, they are weakly measurable.
\end{proof}

\subsection{Endpoint-conditioned benchmark correspondences}

For $\mu,\nu\in\Delta^K$, let
\(
\Pi(\mu,\nu)
=
\left\{
\pi\in\mathbb R_+^{K\times K}:
\sum_{j=1}^K\pi_{ij}=\mu_i,\quad
\sum_{i=1}^K\pi_{ij}=\nu_j
\right\}
\)
denote the transport polytope, and define
\(
c_w(\pi)
=
\sum_{i=1}^K\sum_{j=1}^K
d_w(\xi_i,\xi_j)\pi_{ij}.
\)
All feasible transport plans belong to the compact simplex
\(
\Delta^{K^2}
=
\left\{
\pi\in\mathbb R_+^{K\times K}:
\sum_{i=1}^K\sum_{j=1}^K\pi_{ij}=1
\right\}.
\)

Rather than relying on a general composition result for measurable
correspondences, I verify measurability of the endpoint-conditioned
benchmark sets directly.

Define
\[
\mathcal C_L(\vartheta)
=
\left\{
(\mu,\nu,\pi):
\begin{array}{l}
\mu\in\mathcal M(\mathbf q_X),\;
\nu\in\mathcal M(\mathbf q_Y),\\
\pi\in\Pi(\mu,\nu),\\
c_w(\pi)=\underline D_w(\vartheta)
\end{array}
\right\},
\]
and
\[
\mathcal C_U(\vartheta)
=
\left\{
(\mu,\nu,\pi):
\begin{array}{l}
\mu\in\mathcal M(\mathbf q_X),\;
\nu\in\mathcal M(\mathbf q_Y),\\
\pi\in\Pi(\mu,\nu),\\
D_w(\mu,\nu)=\overline D_w(\vartheta),\\
c_w(\pi)=\overline D_w(\vartheta)
\end{array}
\right\}.
\]

The first definition indeed selects the lower-endpoint benchmark triples. To see this, for every feasible triple,
optimal transport implies
\(
D_w(\mu,\nu)\leq c_w(\pi),
\)
while by definition of the lower identification endpoint,
\(
\underline D_w(\vartheta)
\leq D_w(\mu,\nu).
\)
Hence
\(
c_w(\pi)=\underline D_w(\vartheta)
\)
forces
\(
D_w(\mu,\nu)
=
c_w(\pi)
=
\underline D_w(\vartheta),
\)
so that $(\mu,\nu)$ attains the lower endpoint and $\pi$ is an
optimal transport plan.

For the upper endpoint, the separate restriction
\(
D_w(\mu,\nu)=\overline D_w(\vartheta)
\)
is required because the equality
$c_w(\pi)=\overline D_w(\vartheta)$ alone would not imply that
$(\mu,\nu)$ attains the upper discrepancy endpoint.

\begin{proposition}
\label{prop:endpoint-benchmark-measurable}
The correspondences
\(
\mathcal C_L,\mathcal C_U:
\Theta
\rightrightarrows
\Delta^K\times\Delta^K\times\Delta^{K^2}
\)
are weakly measurable with nonempty compact values.
\end{proposition}

\begin{proof}
I give the argument for $\mathcal C_L$ first. Nonemptiness follows from attainment of
$\underline D_w(\vartheta)$. By
Proposition~\ref{prop:endpoint-correspondences}, there exists
$(\mu^\star,\nu^\star)$ attaining
$\underline D_w(\vartheta)$. The transport polytope
$\Pi(\mu^\star,\nu^\star)$ is nonempty and compact, and the linear
transport objective attains its minimum. Hence there exists
$\pi^\star$ such that
\(
c_w(\pi^\star)
=
D_w(\mu^\star,\nu^\star)
=
\underline D_w(\vartheta),
\)
so that
$(\mu^\star,\nu^\star,\pi^\star)\in\mathcal C_L(\vartheta)$.

The values of $\mathcal C_L$ are compact because they are closed
subsets of the compact space
\(
\Delta^K\times\Delta^K\times\Delta^{K^2}.
\)

I next show that the graph of $\mathcal C_L$ is closed. Let
\(
\vartheta^m\rightarrow\vartheta
\)
and
\(
(\mu^m,\nu^m,\pi^m)
\rightarrow
(\mu,\nu,\pi),
\)
with
\(
(\mu^m,\nu^m,\pi^m)
\in
\mathcal C_L(\vartheta^m)
\)
for every $m$. Continuity of the latent-marginal correspondence
implies
\(
\mu\in\mathcal M(\mathbf q_X),
\)
\(
\nu\in\mathcal M(\mathbf q_Y).
\)
The row- and column-sum restrictions defining the transport polytope
are linear and hence closed, so
\(
\pi\in\Pi(\mu,\nu).
\)
Finally, continuity of $c_w$ and of
$\underline D_w$ gives
\(
c_w(\pi)
=
\lim_{m\rightarrow\infty}c_w(\pi^m)
=
\lim_{m\rightarrow\infty}\underline D_w(\vartheta^m)
=
\underline D_w(\vartheta).
\)
Thus
\(
(\mu,\nu,\pi)\in\mathcal C_L(\vartheta),
\)
and the graph is closed.

Since $\Theta$ and
\(
E_B
:=
\Delta^K\times\Delta^K\times\Delta^{K^2}
\)
are compact metric spaces, closed graph and compact values imply weak
measurability here directly. Indeed, let $F\subseteq E_B$ be closed.
Then
\(
\left\{
\vartheta:
\mathcal C_L(\vartheta)\cap F\neq\varnothing
\right\}
\)
is the projection onto $\Theta$ of
\(
\operatorname{Gr}(\mathcal C_L)
\cap
(\Theta\times F).
\)
This set is compact, and its projection is therefore compact and
hence closed. Hit sets of closed subsets of $E_B$ are consequently
Borel measurable. For an open set $O\subseteq E_B$, write
\[
O
=
\bigcup_{m=1}^{\infty}
\left\{
z\in E_B:
d(z,E_B\setminus O)\geq\frac{1}{m}
\right\}.
\]
Each set on the right is closed, so the corresponding hit set for
$O$ is a countable union of measurable sets. Hence
$\mathcal C_L$ is weakly measurable.

The argument for $\mathcal C_U$ is identical. Nonemptiness follows
from attainment of $\overline D_w(\vartheta)$ and existence of an
optimal transport plan for an upper-endpoint pair. Compactness again
follows from the compact ambient space. Its graph is closed because
the defining restrictions are continuous, using in particular the
continuity of
\(
\vartheta\mapsto\overline D_w(\vartheta)
\)
established in Proposition~\ref{prop:endpoint-correspondences}.
The same closed-graph argument therefore establishes weak
measurability.
\end{proof}

The endpoint-conditioned optimal-plan correspondences can now be
defined directly as the projections $\Pi_{w,L}^\star(\vartheta)=
\left\{
\pi:
(\mu,\nu,\pi)\in\mathcal C_L(\vartheta)
\text{ for some }(\mu,\nu)
\right\}$ and \\ $\Pi_{w,U}^\star(\vartheta)
=
\left\{
\pi:
(\mu,\nu,\pi)\in\mathcal C_U(\vartheta)
\text{ for some }(\mu,\nu)
\right\}$.

\begin{corollary}
\label{cor:optimal-plan-correspondence-measurable}
The correspondences
\(
\Pi_{w,L}^\star,\Pi_{w,U}^\star:
\Theta
\rightrightarrows
\Delta^{K^2}
\)
are weakly measurable with nonempty compact values.
\end{corollary}

\begin{proof}
Nonemptiness and compactness follow because
$\Pi_{w,L}^\star(\vartheta)$ and $\Pi_{w,U}^\star(\vartheta)$ are continuous
projections of the nonempty compact sets
$\mathcal C_L(\vartheta)$ and $\mathcal C_U(\vartheta)$,
respectively.

For measurability, consider the lower-endpoint correspondence. Its
graph is the projection of
$\operatorname{Gr}(\mathcal C_L)$ onto the
$(\vartheta,\pi)$ coordinates. Since
$\operatorname{Gr}(\mathcal C_L)$ is a closed subset of a compact
finite-dimensional space, it is compact, and its projection is
compact and therefore closed. Thus
$\Pi_{w,L}^\star$ has a closed graph and compact values. The same
closed-hit-set argument used in the proof of
Proposition~\ref{prop:endpoint-benchmark-measurable} implies weak
measurability. The proof for $\Pi_{w,U}^\star$ is identical.
\end{proof}

\subsection{Measurability of the endpoint-conditioned cell bounds}

For each $(i,j)\in\{1,\ldots,K\}^2$, define
\begin{align*}
\underline{\pi}_{w,ij}^{L}(\vartheta)
&=
\min_{(\mu,\nu,\pi)\in\mathcal C_L(\vartheta)}
\pi_{ij},
&
\overline{\pi}_{w,ij}^{L}(\vartheta)
&=
\max_{(\mu,\nu,\pi)\in\mathcal C_L(\vartheta)}
\pi_{ij},\\
\underline{\pi}_{w,ij}^{U}(\vartheta)
&=
\min_{(\mu,\nu,\pi)\in\mathcal C_U(\vartheta)}
\pi_{ij},
&
\overline{\pi}_{w,ij}^{U}(\vartheta)
&=
\max_{(\mu,\nu,\pi)\in\mathcal C_U(\vartheta)}
\pi_{ij}.
\end{align*}

\begin{proposition}
\label{prop:cell-bounds-measurable}
For every $(i,j)$, the four functions
\(
\underline{\pi}_{w,ij}^{L},
\;
\overline{\pi}_{w,ij}^{L},
\;
\underline{\pi}_{w,ij}^{U},
\;
\overline{\pi}_{w,ij}^{U}
\)
are measurable, and their defining extrema are attained.
\end{proposition}

\begin{proof}
By Proposition~\ref{prop:endpoint-benchmark-measurable},
$\mathcal C_L$ and $\mathcal C_U$ are weakly measurable with
nonempty compact values. The coordinate mapping
\(
(\mu,\nu,\pi)
\mapsto
\pi_{ij}
\)
is continuous. Applying the Measurable Maximum Theorem to this
mapping and to its negative shows that the corresponding maximum and
minimum value functions are measurable. Compactness of the feasible
sets implies attainment of all four extrema.
\end{proof}

Propositions~\ref{prop:endpoint-correspondences}--
\ref{prop:cell-bounds-measurable} verify the measurability conditions
required for the projection results in
Section~\ref{Section Inference}. In particular, no uniqueness of the
endpoint-attaining latent marginals or of the optimal transport plans
is required.


\section{Computational Details and Global Search}
\label{appendix:global-search}

This appendix provides the computational details and supporting regularity
arguments used in Section~\ref{Section Implementation}. I first give the
linear-programming formulations used to compute the endpoint-conditioned
benchmark transition probabilities conditional on the observable-data
parameter. I then describe the exact sequential conditional-binomial
simulation and common-random-number construction used to evaluate the
finite-sample Monte Carlo confidence region. The following subsections study
the local geometry of the realized Monte Carlo acceptance rule, show how the
structure of the simulation cells and an inward-rank accessibility condition
yield a useful sufficient condition for Near-optimal accessibility of the
discrepancy projections, and establish the persistent-exploration result used
in the proof of Proposition~\ref{prop:persistent-global-convergence}. The
appendix concludes with a numerical sensitivity analysis of the projected
discrepancy with respect to the number of Monte Carlo replications.

\subsection{Linear Programs for the Endpoint-Conditioned Benchmark Bounds}
\label{appendix:benchmark-LPs}

Fix an observable-data parameter
\(
\vartheta
=
(\mathbf q_X,\mathbf q_Y)
\in\Theta.
\)
Conditional on $\vartheta$, the identified sets for the two latent
marginal distributions are
\begin{align*}
\mathcal M(\mathbf q_X)
& =
\left\{
\mu\in\Delta^K:
q_{Xj}\leq\mu_j\leq q_{Xj}+q_{X0},
\quad j=1,\ldots,K
\right\},\quad\text{and}\\
\mathcal M(\mathbf q_Y)
&=
\left\{
\nu\in\Delta^K:
q_{Yj}\leq\nu_j\leq q_{Yj}+q_{Y0},
\quad j=1,\ldots,K
\right\}.
\end{align*}

Let
\begin{equation}
\begin{aligned}
\mathcal T(\mathbf q_X,\mathbf q_Y)
=
\biggl\{
(\mu,\nu,\pi):\;&
\mu,\nu\in\Delta^K,\quad
\pi\in\mathbb R_+^{K\times K},
\\
&
q_{Xj}\leq\mu_j\leq q_{Xj}+q_{X0},
\quad j=1,\ldots,K,
\\
&
q_{Yj}\leq\nu_j\leq q_{Yj}+q_{Y0},
\quad j=1,\ldots,K,
\\
&
\sum_{s=1}^{K}\pi_{rs}=\mu_r,
\quad r=1,\ldots,K,
\\
&
\sum_{r=1}^{K}\pi_{rs}=\nu_s,
\quad s=1,\ldots,K
\biggr\}.
\end{aligned}
\label{eq:supp-T}
\end{equation}
Writing
\(
c_{rs}:=d_w(\xi_r,\xi_s),
\)
the lower discrepancy endpoint conditional on
$(\mathbf q_X,\mathbf q_Y)$ can equivalently be computed as the
optimal transport problem
\begin{equation}
\underline D_w(\mathbf q_X,\mathbf q_Y)
=
\min_{(\mu,\nu,\pi)\in
\mathcal T(\mathbf q_X,\mathbf q_Y)}
\sum_{r=1}^{K}\sum_{s=1}^{K}
c_{rs}\pi_{rs}.
\label{eq:supp-lower-D}
\end{equation}
This formulation jointly selects admissible latent marginals and an
optimal transport plan between them.

For a transition cell $(i,j)$, the lower-endpoint candidate-specific cell bounds are
obtained from
\begin{align}
\underline\pi^L_{w,ij}(\vartheta)
&=
\min_{\mu,\nu,\pi}\;
\pi_{ij}
\label{eq:supp-lower-cell-min}\\
\overline\pi^L_{w,ij}(\vartheta)
&=
\max_{\mu,\nu,\pi}\;
\pi_{ij},
\label{eq:supp-lower-cell-max}
\end{align}
where in both problems
\begin{align}
(\mu,\nu,\pi)
&\in
\mathcal T(\mathbf q_X,\mathbf q_Y),
\label{eq:supp-lower-cell-T}\\
\sum_{r=1}^{K}\sum_{s=1}^{K}
c_{rs}\pi_{rs}
&=
\underline D_w(\vartheta).
\label{eq:supp-lower-cell-cost}
\end{align}
All restrictions in
\eqref{eq:supp-lower-cell-T}--\eqref{eq:supp-lower-cell-cost}
are linear once the conditional endpoint
$\underline D_w(\vartheta)$ has been computed. Hence
\eqref{eq:supp-lower-cell-min} and
\eqref{eq:supp-lower-cell-max} are linear programs.

The equality in \eqref{eq:supp-lower-cell-cost} performs two roles
simultaneously. For every feasible pair of latent marginals,
\(
D_w(\mu,\nu)
\leq
\sum_{r,s}c_{rs}\pi_{rs},
\)
while, by definition,
\(
\underline D_w(\vartheta)
\leq
D_w(\mu,\nu).
\)
Consequently,
\(
\sum_{r,s}c_{rs}\pi_{rs}
=
\underline D_w(\vartheta)
\)
forces
\(
D_w(\mu,\nu)
=
\underline D_w(\vartheta)
\)
and simultaneously establishes optimality of $\pi$ for those
marginals.

The upper discrepancy endpoint requires a different formulation
because maximizing $D_w(\mu,\nu)$ is not itself a linear program.
Recall that
\(
D_w(\mu,\nu)
=
\max_{s\in\mathcal S_K}
\sum_{k=1}^{K-1}
s_k\,w_k\,C_k(\mu-\nu),
\) with
\(\mathcal S_K
=
\{-1,1\}^{K-1},
\)
where
\(
C_k(\mu-\nu)
=
\sum_{\ell=1}^{k}(\mu_\ell-\nu_\ell).
\)
For a fixed $s\in\mathcal S_K$, define
\begin{equation}
\begin{aligned}
\mathcal T^U_s(\mathbf q_X,\mathbf q_Y)
=
\biggl\{
(\mu,\nu,\pi)\in
\mathcal T(\mathbf q_X,\mathbf q_Y):\;&
s_k\,C_k(\mu-\nu)\geq0,
\quad k=1,\ldots,K-1,
\\
&
\sum_{k=1}^{K-1}
s_k\,w_k\,C_k(\mu-\nu)
=
\overline D_w(\vartheta),
\\
&
\sum_{r=1}^{K}\sum_{t=1}^{K}
c_{rt}\pi_{rt}
=
\overline D_w(\vartheta)
\biggr\}.
\end{aligned}
\label{eq:supp-TU}
\end{equation}
The sign restrictions imply
\(
s_k\,C_k(\mu-\nu)
=
|C_k(\mu-\nu)|
\)
for every $k$. Hence the first equality in
\eqref{eq:supp-TU} implies
\(
D_w(\mu,\nu)
=
\overline D_w(\vartheta),
\)
while the transport-cost equality implies that $\pi$ is optimal
between those endpoint-attaining marginals.

The upper-endpoint candidate-specific cell bounds can therefore be written as
\begin{align}
\underline\pi^U_{w,ij}(\vartheta)
&=
\min_{s\in\mathcal S_K}
\;
\min_{(\mu,\nu,\pi)\in
\mathcal T^U_s(\mathbf q_X,\mathbf q_Y)}
\pi_{ij},
\label{eq:supp-upper-cell-min}\\
\overline\pi^U_{w,ij}(\vartheta)
&=
\max_{s\in\mathcal S_K}
\;
\max_{(\mu,\nu,\pi)\in
\mathcal T^U_s(\mathbf q_X,\mathbf q_Y)}
\pi_{ij}.
\label{eq:supp-upper-cell-max}
\end{align}
For each fixed sign vector $s$, the inner problems in
\eqref{eq:supp-upper-cell-min}--\eqref{eq:supp-upper-cell-max}
are linear programs. Sign vectors for which
$\mathcal T^U_s(\mathbf q_X,\mathbf q_Y)$ is empty are discarded.

Finally, for a generic primitive confidence region
$\mathcal C\subseteq\Theta$, the projected benchmark bounds are
\begin{align}
\underline{\pi}^{L}_{w,ij}(\mathcal C)
&=
\inf_{\vartheta\in\mathcal C}
\underline\pi^L_{w,ij}(\vartheta),
&
\overline{\pi}^{L}_{w,ij}(\mathcal C)
&=
\sup_{\vartheta\in\mathcal C}
\overline\pi^L_{w,ij}(\vartheta),
\label{eq:supp-projected-L}\\
\underline{\pi}^{U}_{w,ij}(\mathcal C)
&=
\inf_{\vartheta\in\mathcal C}
\underline\pi^U_{w,ij}(\vartheta),
&
\overline{\pi}^{U}_{w,ij}(\mathcal C)
&=
\sup_{\vartheta\in\mathcal C}
\overline\pi^U_{w,ij}(\vartheta).
\label{eq:supp-projected-U}
\end{align}
Thus the endpoint-conditioned calculations have a nested structure:
an outer optimization over the observable probabilities and inner
linear programs determining the endpoint-attaining benchmark plans.
The inner problems depend only on the candidate observable parameter;
the primitive confidence region enters through the outer projection.

\subsection{Exact Multinomial Simulation and Common Random Numbers}
\label{appendix:MC-simulation}

The finite-sample confidence region
$\mathcal C_{n,B}^{\mathrm{MC}}(1-\alpha)$
requires repeated simulation from the multinomial distribution
specified by each candidate observable parameter. This subsection
describes the simulation procedure used in the numerical
implementation.

Consider first a generic multinomial probability vector
\(
\mathbf q
=
(q_0,\ldots,q_K)
\in\Delta^{K+1}
\)
and sample size $n$. A draw
\(
N=(N_0,\ldots,N_K)
\sim
\operatorname{Multinomial}(n,\mathbf q)
\)
can be generated sequentially from conditional binomial
distributions.

Define
\(
m_0=n,
\qquad
r_0=1.
\)
For $k=0,\ldots,K-1$, let
\(
m_k
=
n-\sum_{\ell=0}^{k-1}N_\ell
\)
and
\(
r_k
=
1-\sum_{\ell=0}^{k-1}q_\ell.
\)
Whenever $r_k>0$, generate
\begin{equation}
N_k
\sim
\operatorname{Binomial}
\left(
m_k,
\frac{q_k}{r_k}
\right).
\label{eq:conditional-binomial}
\end{equation}
The final count is then determined by
\(
N_K
=
n-\sum_{k=0}^{K-1}N_k.
\)
If $r_k=0$, all remaining category probabilities are zero and the
remaining counts are set equal to zero. This convention allows the
simulation algorithm to be used on the boundary of the probability
simplex.

The sequential construction is exact. To see this, conditional on
the first $k$ cell counts, the remaining $m_k$ observations are
distributed across categories $k,\ldots,K$ with probabilities
renormalized by $r_k$. Hence the conditional distribution of $N_k$
is precisely \eqref{eq:conditional-binomial}. Iterating these
conditional distributions yields the multinomial joint law.

For common-random-number evaluation, let
\(
U_k^{(b)},
\;
b=1,\ldots,B,\;
k=0,\ldots,K-1,
\)
be independent $\operatorname{Uniform}(0,1)$ random variables. These
uniform draws are generated once and held fixed throughout the outer
optimization. The conditional binomial count in
\eqref{eq:conditional-binomial} is generated by
\(
N_k^{(b)}(\mathbf q)
=
F_{\operatorname{Bin}
\left(m_k^{(b)},q_k/r_k\right)}^{-1}
\left(
U_k^{(b)}
\right),
\)
where
\(
F^{-1}(u)
=
\inf\{z:F(z)\geq u\}
\)
denotes the generalized inverse of the relevant binomial distribution
function.

Separate independent uniform arrays are used for the source and
target repeated cross-sections. Thus, for each fixed candidate
\(
\vartheta=(\mathbf q_X,\mathbf q_Y),
\)
the resulting simulated vectors satisfy
\[
N_X^{(b)}
\sim
\operatorname{Multinomial}(n_X,\mathbf q_X),
\qquad
N_Y^{(b)}
\sim
\operatorname{Multinomial}(n_Y,\mathbf q_Y),
\]
independently across the two samples and across
$b=1,\ldots,B$.

The randomized tie-breaking variables
\(
U_0^{R},U_1^{R},\ldots,U_B^{R}
\)
used in the Monte Carlo rank statistic are also generated once,
independently of the multinomial simulation uniforms, and held fixed
during optimization. For each candidate $\vartheta$, the simulated
statistics $T_n^{(b)}(\vartheta)$ for $b=1,\ldots,B$, and the observed statistic
$T_n^{(0)}(\vartheta)$ are then ranked as described in
Section~\ref{subsec:MC-inference}.

Holding the auxiliary uniforms fixed has two computational
advantages. First, repeated evaluation of the same candidate
$\vartheta$ gives exactly the same Monte Carlo $p$-value. Second,
nearby candidate parameters are evaluated using the same underlying
simulation draws, avoiding additional numerical noise that would
arise if an independent Monte Carlo experiment were performed at
every objective-function evaluation.

This common-random-number construction does not alter the
finite-sample validity of the Monte Carlo test. For every fixed
candidate $\vartheta$, the simulated samples have exactly the
multinomial distribution specified by that candidate, and at the true
parameter the observed and simulated statistics, together with their
independent randomized tie breakers, satisfy the exchangeability
argument in Theorem~\ref{thm:MC-test}. Dependence of the simulated
statistics across different candidate values of $\vartheta$ is
irrelevant to this pointwise argument.

Conditional on the fixed auxiliary draws, the mapping
\(
\vartheta
\longmapsto
p_{n,B}(\vartheta)
\)
is deterministic. It is not generally smooth. The inverse binomial
transform changes only when a candidate probability crosses a value
at which one of the simulated counts changes. On a region on which
all simulated count vectors are fixed, the Monte Carlo rank can change
only when a simulated likelihood-ratio statistic ties the observed
statistic. Thus the nonsmoothness of the realized Monte Carlo
confidence region is generated by simulation-count boundaries and
rank-tie boundaries. The next subsection describes this local
structure more precisely.

\subsection{Geometry of the Realized Monte Carlo Acceptance Rule}
\label{subsec:MC-geometry}

Let \(\mathcal G_{n,B}\) be the sigma-field generated by the observed samples and the Monte Carlo auxiliary random variables, as defined in Section~\ref{subsec:mc-global-search}. Throughout the remainder of this appendix, a realization of \(\mathcal G_{n,B}\) is held fixed. The realized Monte Carlo confidence region \(\mathcal C_{n,B}^{\mathrm{MC}}\), the candidate-specific projection objectives, and their exact outer extrema are therefore deterministic. Probability statements concern only the algorithmic randomization and are taken under \(\mathbb P_{\mathrm{alg}}\).

Conditional on the fixed auxiliary uniforms, each simulated multinomial count is locally constant except at parameter values at which the relevant binomial distribution function evaluated at an integer count equals the corresponding auxiliary uniform. Because the sample sizes, number of categories, and number of replications are finite, only finitely many such simulation-change boundaries arise after refinement over the finitely many possible preceding count configurations.

The structure is especially transparent on the full-dimensional simplex interior. For one multinomial probability vector define the sequential conditional probabilities
\(
\lambda_k
=
\frac{q_k}{1-\sum_{\ell=0}^{k-1}q_\ell},
\qquad
k=0,\ldots,K-1.
\)
These are the stick-breaking coordinates associated with the conditional-binomial construction. Conditional on a fixed complete bank of simulated counts, each inverse-binomial requirement restricts the corresponding \(\lambda_k\) to an interval, possibly with one endpoint included. Intersecting these restrictions over the finitely many Monte Carlo replications therefore gives a Cartesian product of intervals in the stick-breaking coordinates. Thus the simulation-count regions have a structured cell geometry; the remaining changes in the Monte Carlo rank within a fixed count cell arise only from likelihood-ratio ties.

For later use, it is convenient to make this cell structure explicit. For one sample, let
\(
\Psi_1:[0,1]^K\to\Delta^{K+1}
\)
denote the stick-breaking map defined by
\[
q_0=\lambda_0,
\qquad
q_k
=
\lambda_k\prod_{\ell=0}^{k-1}(1-\lambda_\ell),
\quad k=1,\ldots,K-1,
\qquad
q_K
=
\prod_{\ell=0}^{K-1}(1-\lambda_\ell).
\]
Let
\(
\Psi:[0,1]^{2K}\to\Theta
\)
be the product of the corresponding maps for the source and target samples. On \((0,1)^{2K}\), \(\Psi\) is a smooth one-to-one parameterization of the relative interior of \(\Theta\).

For \(m\geq0\) and \(z\in\{0,\ldots,m\}\), write
\(
F_{m,z}(p)
\)
for the distribution function of a \(\operatorname{Binomial}(m,p)\) random variable evaluated at \(z\), with the convention
\(
F_{m,-1}(p)=0.
\)

\begin{lemma}[Fixed-count cells]
\label{lem:fixed-count-cells}
Fix the observed data and all Monte Carlo simulation uniforms on the probability-one event that every simulation uniform lies in \((0,1)\), and let
\(
\lambda^\star\in[0,1]^{2K}
\)
be a stick-breaking representation of a candidate
\(
\vartheta^\star=\Psi(\lambda^\star).
\)
Let \(\mathcal S(\lambda^\star)\) denote the set of all stick-breaking coordinate vectors that, under the fixed common random numbers, generate exactly the same complete bank of simulated multinomial counts as 
 \(\lambda^\star\). Then
\(
\mathcal S(\lambda^\star)
=
I_1\times\cdots\times I_{2K},
\)
where each \(I_\ell\subseteq[0,1]\) is a nondegenerate interval. Moreover, with interior and closure in the following display taken in \(\mathbb R^{2K}\),
\[
\lambda^\star
\in
\operatorname{cl}
\left(
\operatorname{int}\mathcal S(\lambda^\star)
\cap
(0,1)^{2K}
\right).
\]
Thus every realized simulated-count configuration can be approached arbitrarily closely by strictly interior observable probabilities without changing any simulated count.
\end{lemma}

\begin{proof}
Consider first one sample and one sequential step \(k\). Conditional on the preceding simulated counts, the remaining sample size \(m_k^{(b)}\) is fixed. If the realized count in replication \(b\) is \(z_k^{(b)}\), the generalized-inverse convention implies that the same count is generated at conditional probability \(p\) exactly when
\(
F_{m_k^{(b)},z_k^{(b)}-1}(p)
<
U_k^{(b)}
\leq
F_{m_k^{(b)},z_k^{(b)}}(p).
\)
When \(m_k^{(b)}>0\), the binomial distribution functions appearing here are continuous in \(p\), and for \(z<m_k^{(b)}\) they are strictly decreasing in \(p\). Hence the set of \(p\in[0,1]\) satisfying the displayed inequalities is an interval of the form \([0,u]\) or \((\ell,u]\), with the natural endpoint modifications when the realized count is \(0\) or \(m_k^{(b)}\). If \(m_k^{(b)}=0\), the realized count is zero for every \(p\), so this replication imposes no restriction.

Intersecting over \(b=1,\ldots,B\) gives an interval \(I_k\) containing \(\lambda_k^\star\). The lower restrictions are strict except at the endpoint zero, whereas the upper restrictions are weak. Since every simulation uniform belongs to \((0,1)\), the resulting interval cannot collapse to a singleton: if \(\lambda_k^\star\in(0,1)\), its lower endpoint is strictly below \(\lambda_k^\star\); at \(\lambda_k^\star=0\) the interval contains sufficiently small positive values; and at \(\lambda_k^\star=1\) it contains values sufficiently close to one from below.

Conditional on preservation of the counts at the preceding steps, the remaining sample sizes are unchanged. Proceeding recursively over \(k=0,\ldots,K-1\) therefore shows that the set of parameters preserving the entire count vector in every replication is the Cartesian product of these coordinate intervals. Applying the same argument independently to the source and target samples gives
\(
\mathcal S(\lambda^\star)=I_1\times\cdots\times I_{2K}.
\)
Because every factor is nondegenerate and has interior points in \((0,1)\) arbitrarily close to \(\lambda_\ell^\star\), the final closure statement follows by taking coordinatewise approximating sequences.
\end{proof}

Lemma~\ref{lem:fixed-count-cells} shows that simulation-count boundaries do
not themselves obstruct regular approximation. Any candidate can be
approached arbitrarily closely by strictly interior parameters that generate
the same complete bank of simulated counts. Thus, within a fixed count cell,
the only remaining obstruction to approximating an accepted candidate by
interior accepted points comes from the Monte Carlo rank comparisons. I now
make this remaining rank condition explicit.

Fix a relatively open region on which all simulated count vectors are constant. For replication \(b=1,\ldots,B\), write
\(
G_b(\vartheta)
=
T_n^{(b)}(\vartheta)-T_n^{(0)}(\vartheta).
\)
On the relative interior of a fixed simplex face, wherever the likelihood-ratio statistics are finite, \(G_b\) is continuous; on the full-dimensional simplex interior it is real analytic. With the simulated and observed empirical vectors fixed, it can be written, up to a constant \(c_b\) depending only on those empirical vectors, as
\[
G_b(\vartheta)
=
c_b
-
2n_X\sum_{j=0}^K
\left(
\widehat q^{(b)}_{Xj}-\widehat q_{Xj}
\right)\log q_{Xj}
-
2n_Y\sum_{j=0}^K
\left(
\widehat q^{(b)}_{Yj}-\widehat q_{Yj}
\right)\log q_{Yj}.
\]
If the simulated empirical vectors coincide with the observed empirical vectors, the two likelihood-ratio statistics tie identically and their ordering is fixed by the randomized tie breakers. Otherwise, on an analytic region on which \(G_b\) is not identically zero, the zero set \(\{G_b=0\}\) has empty relative interior. Consequently, away from simulation-change boundaries and nontrivial rank-tie boundaries, all comparisons entering the randomized Monte Carlo rank are locally unchanged and the Monte Carlo \(p\)-value is locally constant.

Call an accepted candidate \(\vartheta\in\mathcal C_{n,B}^{\mathrm{MC}}\) \emph{MC-regular} if there exists a nonempty relatively open neighborhood \(U\) of \(\vartheta\) in \(\Theta\) such that
\(
U\subseteq\mathcal C_{n,B}^{\mathrm{MC}}.
\)
In particular, an accepted candidate away from the simulation-change and nontrivial rank-tie boundaries is MC-regular whenever the preceding local constancy holds in the relative topology of \(\Theta\).

The fixed-count-cell lemma isolates the remaining possible obstruction to regular approximation: the randomized rank itself. Fix an accepted candidate
\(
\vartheta^\star=\Psi(\lambda^\star)
\)
and its count cell
\(
\mathcal S^\star=\mathcal S(\lambda^\star).
\)
On
\(
\operatorname{int}\mathcal S^\star\cap(0,1)^{2K},
\)
define
\(
H_b(\lambda)=G_b(\Psi(\lambda)).
\)
Let \(P(\vartheta^\star)\) denote the replications that are strictly more extreme than the observed statistic at \(\vartheta^\star\), and let \(J(\vartheta^\star)\) denote the replications that tie the observed statistic at \(\vartheta^\star\) but for which \(H_b\) is not identically zero on the interior of the count cell. Ties for which \(H_b\equiv0\) are excluded from \(J(\vartheta^\star)\), because their ordering is fixed throughout the cell by the randomized tie breakers.

Call \(\vartheta^\star\) \emph{inward-rank accessible} if there exists a sequence
\(
\lambda^m
\in
\operatorname{int}\mathcal S^\star\cap(0,1)^{2K},
\;
\lambda^m\longrightarrow\lambda^\star,
\)
such that, for all sufficiently large \(m\),
\(
H_b(\lambda^m)>0
\;
\text{for every }
b\in P(\vartheta^\star)\cup J(\vartheta^\star),
\)
and
\(
H_b(\lambda^m)\neq0
\;
\text{for every }b
\text{ for which }H_b\not\equiv0
\text{ on }
\operatorname{int}\mathcal S^\star.
\)
Thus an inward-rank accessible accepted point can be approached from the interior of the same simulated-count cell while preserving every strict comparison that already contributes to the Monte Carlo rank, turning every nontrivial active tie in favor of the simulated statistic, and avoiding all nontrivial rank-tie surfaces.

\begin{proposition}[Inward rank accessibility]
\label{prop:inward-rank-closure}
Fix the observed data and Monte Carlo auxiliary draws. Suppose every accepted candidate
\(
\vartheta^\star\in\mathcal C_{n,B}^{\mathrm{MC}}
\)
is inward-rank accessible for at least one stick-breaking representation \(\lambda^\star\). Then
\(
\mathcal C_{n,B}^{\mathrm{MC}}
\subseteq
\operatorname{cl}_{\Theta}
\left(
\operatorname{int}_{\Theta}
\mathcal C_{n,B}^{\mathrm{MC}}
\right).
\)
\end{proposition}

\begin{proof}
Fix
\(
\vartheta^\star\in\mathcal C_{n,B}^{\mathrm{MC}}
\)
and let \(\lambda^m\) be a sequence supplied by inward rank accessibility. Set
\(
\vartheta^m=\Psi(\lambda^m).
\)
By construction, \(\vartheta^m\to\vartheta^\star\) and every \(\vartheta^m\) generates the same simulated count bank as \(\vartheta^\star\).

Consider the Monte Carlo rank contributions. Every replication that is strictly more extreme at \(\vartheta^\star\) remains strictly more extreme at \(\vartheta^m\). Every nontrivial tie at \(\vartheta^\star\) becomes strictly more extreme at \(\vartheta^m\), so its rank contribution changes from either zero or one to one. An identically tied replication retains the same randomized ordering throughout the count cell. Replications that contribute zero at \(\vartheta^\star\) cannot reduce the rank. Hence, for all sufficiently large \(m\),
\(
p_{n,B}(\vartheta^m)
\geq
p_{n,B}(\vartheta^\star)
>
\alpha,
\)
so \(\vartheta^m\in\mathcal C_{n,B}^{\mathrm{MC}}\).

The sequence also lies in the interior of the fixed-count cell and avoids every nontrivial rank-tie surface. Therefore all simulated counts and all nontrivial rank comparisons are locally constant around \(\lambda^m\); identically tied comparisons remain ordered by the fixed randomized tie breakers. The Monte Carlo \(p\)-value is consequently locally constant around \(\lambda^m\). Since \(\lambda^m\in(0,1)^{2K}\), the stick-breaking map is locally a homeomorphism onto the relative interior of \(\Theta\). It follows that
\(
\vartheta^m\in
\operatorname{int}_{\Theta}\mathcal C_{n,B}^{\mathrm{MC}}.
\)
As \(\vartheta^m\to\vartheta^\star\), the desired inclusion follows.
\end{proof}

At an interior candidate, inward rank accessibility has a simple differential sufficient condition. By Lemma~\ref{lem:fixed-count-cells}, write the count cell coordinatewise as
\(
\mathcal S^\star=I_1\times\cdots\times I_{2K}.
\)
Because the generalized-inverse rule gives open lower and weak upper restrictions away from the endpoint zero, an interior candidate can lie on a count-cell facet only at an upper endpoint. Let \(A(\lambda^\star)\) be the coordinates for which \(\lambda_\ell^\star\) equals that upper endpoint. Suppose the functions \(H_b\) are differentiable at \(\lambda^\star\). If there exists \(v\in\mathbb R^{2K}\) such that
\[
v_\ell<0,
\qquad
\ell\in A(\lambda^\star),
\quad\text{and}\quad
\nabla H_b(\lambda^\star)^\top v>0,
\qquad
b\in J(\vartheta^\star),
\]
then \(\lambda^\star+t v\) lies in the interior of the same count cell and makes every active nontrivial tie strictly more extreme for all sufficiently small \(t>0\). Continuity preserves the already strict rank comparisons, so \(\vartheta^\star\) is inward-rank accessible. A convenient stronger sufficient condition is linear independence of the oriented active count-cell normals
\(
\{-e_\ell:\ell\in A(\lambda^\star)\}
\)
and the active rank-tie gradients
\(
\{\nabla H_b(\lambda^\star):b\in J(\vartheta^\star)\}.
\)
Indeed, if these vectors are the rows of a full-row-rank matrix \(M\), one may solve \(Mv=\mathbf 1\), which gives all the required strict inequalities. This linear-independence condition is sufficient rather than necessary; the primitive requirement is simply existence of an inward rank-improving direction.

Figure~\ref{fig:MC-accessibility} summarizes the geometric roles of
Lemma~\ref{lem:fixed-count-cells} and Proposition~\ref{prop:inward-rank-closure}

\begin{figure}[H]
\centering

\begin{tikzpicture}[
    x=1cm,
    y=1cm,
    >=Latex,
    every node/.style={font=\footnotesize}
]


\node[font=\small\bfseries] at (3.0,4.55)
{Panel A: Fixed-count cell};

\draw[thick] (0,0) rectangle (6.0,4.1);

\node[anchor=north west] at (0.18,3.92)
{Stick-breaking space $[0,1]^{2K}$};

\fill[gray!15] (1.05,0.75) rectangle (4.95,3.15);
\draw[thick] (1.05,0.75) rectangle (4.95,3.15);

\node at (3.0,2.78)
{fixed-count cell $S(\lambda^\star)$};

\fill (3.72,1.72) circle (1.8pt);
\node[anchor=south] at (3.72,1.84)
{$\lambda^m$};

\fill (4.95,1.72) circle (1.8pt);
\node[anchor=west] at (5.05,1.72)
{$\lambda^\star$};

\draw[-{Latex[length=2.3mm]},thick,dashed]
      (3.86,1.72) -- (4.80,1.72);

\node at (4.34,1.28)
{$\lambda^m\to\lambda^\star$};


\begin{scope}[xshift=7.0cm]

\node[font=\small\bfseries] at (3.0,4.55)
{Panel B: Inward-rank accessibility};

\draw[thick] (0,0) rectangle (6.0,4.1);

\node[anchor=north west] at (0.18,3.92)
{Observable-parameter space $\Theta$};

\filldraw[
    fill=gray!15,
    draw=black,
    thick,
    rounded corners=11pt
]
(0.95,0.78)
-- (4.75,0.78)
.. controls (5.25,0.82) and (5.40,1.30) .. (5.18,1.72)
.. controls (4.93,2.10) and (5.42,2.62) .. (4.85,3.12)
-- (1.35,3.12)
.. controls (0.85,2.98) and (0.70,2.40) .. (0.82,1.98)
.. controls (0.64,1.48) and (0.70,1.02) .. (0.95,0.78)
-- cycle;

\node at (2.75,2.72)
{$C^{MC}_{n,B}$};

\fill (3.75,1.68) circle (1.8pt);
\node[anchor=south] at (3.75,1.80)
{$\vartheta^m$};

\fill (5.12,1.90) circle (1.8pt);
\node[anchor=west] at (5.22,1.90)
{$\vartheta^\star$};

\draw[-{Latex[length=2.3mm]},thick,dashed]
      (3.90,1.70) -- (4.96,1.87);

\node at (4.42,1.27)
{$\vartheta^m\to\vartheta^\star$};

\end{scope}


\node[
    align=center,
    text width=12.5cm,
    font=\footnotesize
] at (6.5,-0.55)
{Lemma~C.1 provides same-count interior approximation;
inward-rank accessibility strengthens this to interior
\emph{accepted} approximation.};

\end{tikzpicture}

\caption{Geometry underlying Proposition~C.1.
Panel A illustrates Lemma~C.1: a candidate
$\lambda^\star$ can be approximated by points
$\lambda^m$ in the interior of the same fixed-count cell, so the complete
bank of simulated counts is unchanged.
Panel B illustrates the additional role of inward-rank accessibility:
an accepted candidate $\vartheta^\star$ can be approximated by relatively
interior accepted candidates $\vartheta^m$, yielding
$C^{MC}_{n,B}\subseteq
\operatorname{cl}_{\Theta}
\bigl(\operatorname{int}_{\Theta}C^{MC}_{n,B}\bigr)$.}
\label{fig:MC-accessibility}
\end{figure}

\subsection{A Regular-Region Sufficient Condition for Near-Optimal Accessibility}
\label{subsec:regular-route-accessibility}

For a reported outer objective \(h\), say that a candidate is \emph{\(h\)-regular} if it is MC-regular and \(h\) is continuous at that candidate. The following lemma gives a simple route to Assumption~\ref{ass:near-optimal-accessibility}.

\begin{lemma}
\label{lem:regular-implies-accessibility}
Suppose \(\Gamma\) has full support on \(\Theta\). Consider a reported outer objective \(h\).

If \(h\) is minimized and, for every \(\varepsilon>0\), there exists an \(h\)-regular candidate \(\vartheta_\varepsilon\in\mathcal C_{n,B}^{\mathrm{MC}}\) satisfying
\(
h(\vartheta_\varepsilon)<h^\star+\varepsilon,
\)
then
\(
\Gamma\!\left(A_\varepsilon(h)\right)>0
\;
\text{for every }\varepsilon>0.
\)
The analogous conclusion holds for a maximization objective if, for every \(\varepsilon>0\), there exists an \(h\)-regular accepted candidate satisfying
\(
h(\vartheta_\varepsilon)>h^\star-\varepsilon.
\)
\end{lemma}

\begin{proof}
Consider minimization and fix \(\varepsilon>0\). Apply the stated condition with \(\varepsilon/2\) to obtain an \(h\)-regular accepted candidate \(\vartheta_{\varepsilon/2}\) such that
\(
h(\vartheta_{\varepsilon/2})<h^\star+\frac{\varepsilon}{2}.
\)
MC regularity gives a relatively open neighborhood \(U_1\) of \(\vartheta_{\varepsilon/2}\) contained in \(\mathcal C_{n,B}^{\mathrm{MC}}\). Continuity of \(h\) at \(\vartheta_{\varepsilon/2}\) gives a relatively open neighborhood \(U_2\) such that
\(
h(\vartheta)<h^\star+\varepsilon
\;
\text{for every }\vartheta\in U_2.
\)
Thus \(U_1\cap U_2\) is a nonempty relatively open subset of \(A_\varepsilon(h)\). Full support of \(\Gamma\) implies
\[
\Gamma\!\left(A_\varepsilon(h)\right)
\geq
\Gamma(U_1\cap U_2)
>0.
\]
The maximization case is identical with the inequalities reversed.
\end{proof}

For the discrepancy projections, Proposition~\ref{prop:endpoint-correspondences} establishes that \(\underline D_w(\vartheta)\) and \(\overline D_w(\vartheta)\) are continuous on \(\Theta\). Hence the regularity requirement in Lemma~\ref{lem:regular-implies-accessibility} reduces to the existence of arbitrarily near-optimal MC-regular accepted candidates. For the endpoint-conditioned benchmark functions \(\underline\pi^e_{w,ij}(\vartheta)\) and \(\overline\pi^e_{w,ij}(\vartheta)\), Appendix~\ref{appendix:projection-measurability} establishes measurability and attainment of the conditional cell extrema. Local continuity is an additional sufficient property: for example, it holds on neighborhoods on which the relevant endpoint branch, sign configuration, and optimal linear-programming basis remain valid. These regularity observations provide sufficient routes to Near-optimal accessibility; they are not assumptions required by Proposition~\ref{prop:persistent-global-convergence} itself.

\subsection{A Geometric Sufficient Condition for the Discrepancy Projections}
\label{subsec:D-geometric-accessibility}

For the discrepancy projections, continuity permits a particularly simple sufficient condition for Near-optimal accessibility. Proposition~\ref{prop:inward-rank-closure} gives a primitive condition under which the geometric property used below is satisfied. Fix the realized Monte Carlo confidence region and write
\(
\mathcal C
=
\mathcal C_{n,B}^{\mathrm{MC}}.
\)
Throughout this subsection, interior and closure are taken relative to the observable-parameter space \(\Theta\).

\begin{proposition}[Geometric accessibility of the discrepancy projections]
\label{prop:D-closure-accessibility}
Suppose that \(\mathcal C\) is nonempty, that the global proposal measure
\(\Gamma\) has full support on \(\Theta\), and that
\begin{equation}
\label{eq:MC-closure-interior}
\mathcal C
\subseteq
\operatorname{cl}_{\Theta}
\left(
\operatorname{int}_{\Theta}\mathcal C
\right).
\end{equation}
Then Assumption~\ref{ass:near-optimal-accessibility} holds for both
discrepancy objectives \(h=\underline D_w\) and
\(h=\overline D_w\).

Consequently,
\(
\inf_{\vartheta\in\mathcal C}\underline D_w(\vartheta)
=
\inf_{\vartheta\in\operatorname{int}_{\Theta}\mathcal C}
\underline D_w(\vartheta)
\)
and
\(
\sup_{\vartheta\in\mathcal C}\overline D_w(\vartheta)
=
\sup_{\vartheta\in\operatorname{int}_{\Theta}\mathcal C}
\overline D_w(\vartheta),
\)
and every \(\varepsilon\)-optimal feasible set for either discrepancy
projection has strictly positive \(\Gamma\)-probability.
\end{proposition}

\begin{proof}
By Proposition~\ref{prop:endpoint-correspondences},
\(\underline D_w\) and \(\overline D_w\) are continuous on \(\Theta\).

Consider first the lower discrepancy projection and write
\(
d_L^\star
=
\inf_{\vartheta\in\mathcal C}\underline D_w(\vartheta).
\)
Fix \(\varepsilon>0\). By the definition of the infimum, there exists
\(\vartheta^\star\in\mathcal C\) such that
\(
\underline D_w(\vartheta^\star)
<
d_L^\star+\frac{\varepsilon}{4}.
\)
Condition~\eqref{eq:MC-closure-interior} implies that
\(\vartheta^\star\) can be approximated by points in
\(\operatorname{int}_{\Theta}\mathcal C\). Continuity of
\(\underline D_w\) at \(\vartheta^\star\) therefore allows us to choose
\(
\widetilde\vartheta
\in
\operatorname{int}_{\Theta}\mathcal C
\)
sufficiently close to \(\vartheta^\star\) that
\(
\underline D_w(\widetilde\vartheta)
<
\underline D_w(\vartheta^\star)
+\frac{\varepsilon}{4}
<
d_L^\star+\frac{\varepsilon}{2}.
\)

Because \(\widetilde\vartheta\) is an interior point of \(\mathcal C\)
relative to \(\Theta\), there exists a relatively open neighborhood
\(U_1\) of \(\widetilde\vartheta\) such that
\(
U_1\subseteq\mathcal C.
\)
Continuity of \(\underline D_w\) at \(\widetilde\vartheta\) gives a
relatively open neighborhood \(U_2\) such that
\(
\underline D_w(\vartheta)
<
d_L^\star+\varepsilon
\;
\text{for every }\vartheta\in U_2.
\)
Consequently,
\(
\varnothing
\neq
U_1\cap U_2
\subseteq
A_\varepsilon(\underline D_w).
\)
Since \(\Gamma\) has full support on \(\Theta\),
\(
\Gamma\!\left(A_\varepsilon(\underline D_w)\right)
\geq
\Gamma(U_1\cap U_2)
>0.
\)
As \(\varepsilon>0\) was arbitrary,
Assumption~\ref{ass:near-optimal-accessibility} holds for
\(h=\underline D_w\).

For the upper discrepancy projection, write
\(
d_U^\star
=
\sup_{\vartheta\in\mathcal C}\overline D_w(\vartheta).
\)
Given \(\varepsilon>0\), choose
\(\vartheta^\star\in\mathcal C\) such that
\(
\overline D_w(\vartheta^\star)
>
d_U^\star-\frac{\varepsilon}{4}.
\)
Condition~\eqref{eq:MC-closure-interior} and continuity of
\(\overline D_w\) give
\(\widetilde\vartheta\in
\operatorname{int}_{\Theta}\mathcal C\) satisfying
\(
\overline D_w(\widetilde\vartheta)
>
d_U^\star-\frac{\varepsilon}{2}.
\)
Using interiority of \(\widetilde\vartheta\) relative to \(\mathcal C\)
and continuity once more yields a nonempty relatively open set
\(
U
\subseteq
\left\{
\vartheta\in\mathcal C:
\overline D_w(\vartheta)>d_U^\star-\varepsilon
\right\}
=
A_\varepsilon(\overline D_w).
\)
Full support of \(\Gamma\) therefore implies
\(
\Gamma\!\left(A_\varepsilon(\overline D_w)\right)>0.
\)
Thus Near-optimal accessibility also holds for
\(h=\overline D_w\).

Finally, the two equalities between extrema over
\(\mathcal C\) and over
\(\operatorname{int}_{\Theta}\mathcal C\) follow directly from
\eqref{eq:MC-closure-interior} and continuity: every point of
\(\mathcal C\) is the limit of a sequence of interior points whose
discrepancy values converge to its discrepancy value.
\end{proof}

Condition~\eqref{eq:MC-closure-interior} has a direct interpretation:
every accepted parameter can be approximated arbitrarily closely by
parameters lying in a relatively open accepted region. It therefore excludes
accepted lower-dimensional strata that are isolated from all such regions.
The condition is sufficient rather than necessary; Near-optimal
accessibility requires only positive global-proposal probability for the
near-optimal feasible sets themselves. Because the topology is relative to
\(\Theta\), condition~\eqref{eq:MC-closure-interior} is not an interiority
restriction on the true observable parameter and allows probability vectors
on the boundary of the simplex.

\subsection{Persistent Global Exploration}
\label{subsec:appendix-persistent-exploration}
\begin{proposition}[Borel measurability of the realized acceptance region]
Fix the observed data and the Monte Carlo auxiliary draws. Then the mapping
\(
\vartheta\longmapsto p_{n,B}(\vartheta)
\)
is Borel measurable on $\Theta$. Consequently,
\(
\mathcal C^{MC}_{n,B}(1-\alpha)
=
\{\vartheta\in\Theta:p_{n,B}(\vartheta)>\alpha\}
\)
is a Borel subset of $\Theta$.
\end{proposition}

\begin{proof}
For each Monte Carlo replication, the sequential inverse-binomial
construction maps the candidate probability vector into a
finite-valued simulated count vector. At each sequential step, the
event that a particular count $z$ is generated can be written as the inequalities
\(
F_{m,z-1}(p)<u\leq F_{m,z}(p),
\)
where $u$ is the fixed auxiliary uniform and $F_{m,z}$ is the binomial
distribution function evaluated at $z$. Since $F_{m,z}(p)$ is
continuous in $p$, these sets are Borel. Iterating over the finitely
many sequential steps therefore shows that each simulated count vector,
and hence each simulated empirical probability vector, is a Borel
function of $\vartheta$.

For every $b=0,\ldots,B$, the likelihood-ratio statistic
$T_n^{(b)}(\vartheta)$ is consequently Borel measurable, using the
usual extended-real convention for the Kullback--Leibler divergence.
It follows that the comparison sets
\(
\{T_n^{(b)}(\vartheta)>T_n^{(0)}(\vartheta)\}
\)
and
\(
\{T_n^{(b)}(\vartheta)=T_n^{(0)}(\vartheta)\}
\)
are Borel. Since the randomized tie breakers are fixed, each rank
indicator $I_b(\vartheta)$ is Borel measurable. Therefore
\(
p_{n,B}(\vartheta)
=
\frac{1+\sum_{b=1}^B I_b(\vartheta)}{B+1}
\)
is Borel measurable. The stated confidence region is the inverse image
of $(\alpha,1]$ under this mapping and is therefore Borel.
\end{proof}

Let \(\mathcal F^{\mathrm{alg}}_{m-1}\) denote the algorithmic search history through proposal \(m-1\). Under \eqref{eq:mixture-proposal},
\(
\mathbb P_{\mathrm{alg}}
\!\left(
\vartheta_m^{\mathrm{prop}}\in A
\mid
\mathcal F^{\mathrm{alg}}_{m-1}
\right)
\geq
\rho\,\Gamma(A)
\)
for every Borel set \(A\subseteq\Theta\). No independence assumption is required for the local proposal sequence; it may depend on the complete algorithmic history.

The reported outer objectives are Borel measurable: the discrepancy
endpoint functions are continuous by Proposition~\ref{prop:endpoint-correspondences}, while the
endpoint-conditioned cell-extremal functions are measurable by
Proposition~\ref{prop:cell-bounds-measurable}. Hence, for every reported objective $h$ and every
$\varepsilon>0$, the near-optimal feasible set $A_\varepsilon(h)$ is
Borel measurable.

\begin{lemma}\label{lem:persistent-hit}
Fix the realization of $\mathcal G_{n,B}$. If
$A\subseteq\Theta$ is Borel measurable and $\Gamma(A)>0$, then, for
every deterministic $M\geq 0$ and $N\geq 1$,
\(
\mathbb P_{\mathrm{alg}}
\left(
\bigcap_{m=M+1}^{M+N}
\{\vartheta_m^{\mathrm{prop}}\notin A\}
\right)
\leq
\{1-\rho\Gamma(A)\}^N,
\)
holds.
Consequently,
\(
\mathbb P_{\mathrm{alg}}
\left(
\vartheta_m^{\mathrm{prop}}\in A
\ \text{infinitely often}
\right)
=1.
\)
\end{lemma}

\begin{proof}
For every $m$,
\(
\mathbb P_{\mathrm{alg}}
\left(
\vartheta_m^{\mathrm{prop}}\notin A
\mid
\mathcal F^{\mathrm{alg}}_{m-1}
\right)
\leq
1-\rho\Gamma(A).
\)
Repeated conditioning therefore gives, for every deterministic
$M\geq0$ and $N\geq1$,
\[
\mathbb P_{\mathrm{alg}}
\left(
\bigcap_{m=M+1}^{M+N}
\{\vartheta_m^{\mathrm{prop}}\notin A\}
\right)
\leq
\{1-\rho\Gamma(A)\}^N.
\]
Letting $N\to\infty$ and using continuity of probability from above,
\[
\mathbb P_{\mathrm{alg}}
\left(
\bigcap_{m=M+1}^{\infty}
\{\vartheta_m^{\mathrm{prop}}\notin A\}
\right)
=0
\]
for every $M\geq0$.

Now
\(
\left\{
\vartheta_m^{\mathrm{prop}}\in A
\ \text{infinitely often}
\right\}^{c}
=
\bigcup_{M=0}^{\infty}
\bigcap_{m=M+1}^{\infty}
\{\vartheta_m^{\mathrm{prop}}\notin A\}.
\)
The right-hand side is a countable union of probability-zero events.
Hence
\(
\mathbb P_{\mathrm{alg}}
\left(
\vartheta_m^{\mathrm{prop}}\in A
\ \text{infinitely often}
\right)
=1.
\)
\end{proof}

Assumption~\ref{ass:near-optimal-accessibility} applies Lemma~\ref{lem:persistent-hit} directly to the \(\varepsilon\)-optimal feasible sets \(A_\varepsilon(h)\). Hence continuity of the outer objective, convexity or smoothness of the Monte Carlo confidence region, and local regularity of the acceptance rule are not required for Proposition~\ref{prop:persistent-global-convergence}. The preceding regularity and geometric conditions are instead convenient sufficient conditions for verifying Near-optimal accessibility in the projection problems considered here. 

\subsection{Sensitivity to the Number of Monte Carlo Replications}
\label{appendix:MC-B-sensitivity}

The number of Monte Carlo replications $B$ plays a different role from the
computational budget used to solve the outer projection problems. As discussed
in Section~\ref{subsec:MC-inference}, finite-sample validity of the randomized
Monte Carlo test does not require $B$ to be large. At $\alpha=0.05$, both
$B=999$ and $B=9999$ satisfy
\(
\alpha(B+1)\in\mathbb N,
\)
so both yield the exact finite-sample calibration of
Theorem~\ref{thm:MC-test}. Increasing $B$ instead affects power and the
probability that the randomized Monte Carlo decision differs from that of the
infeasible fundamental test; away from the boundary at which the fundamental
p-value equals the nominal level, this disagreement probability vanishes as
$B$ increases \citep[Section~3]{DUFOUR2006}. Thus larger $B$ can improve power
and concordance at greater computational cost.

To assess the sensitivity of the projected discrepancy, I repeated the
calculation with $B=9999$. The main analysis with $B=999$ gives
\(
[0.2400862378,\,0.4767587822].
\)
The $B=9999$ upper endpoint was $0.4778451191$. The initial lower-endpoint
search produced $0.2405950120$ but failed the search-stability diagnostic, so
I conducted 39 additional simulated-annealing restarts. The refined lower
endpoint was $0.2400263098$, with Monte Carlo p-value $0.0501$, and satisfied
the stability diagnostic. The resulting interval is therefore
\(
[0.2400263098,\,0.4778451191].
\)
Table~\ref{tab:B-sensitivity-D} summarizes the comparison.

\begin{table}[htbp]
\centering
\caption{Sensitivity of the projected discrepancy to the number of
Monte Carlo replications}
\label{tab:B-sensitivity-D}
\begin{tabular}{lccc}
\hline\hline
$B$
& Lower endpoint
& Upper endpoint
& Width
\\
\hline
999
& 0.2400862
& 0.4767588
& 0.2366725
\\
9999
& 0.2400263
& 0.4778451
& 0.2378188
\\
\hline
Difference
& $-0.0000599$
& 0.0010863
& 0.0011463
\\
\hline\hline
\end{tabular}
\end{table}

The effect of increasing $B$ tenfold is small relative to the scale of the
projected interval. On the normalized scale $D/(K-1)=D/3$, the endpoint
changes correspond in absolute value to approximately 0.002 and 0.036
percentage points of the maximum possible average rank displacement. Most
importantly, the lower projection remains well separated from zero and the
range of minimum restructuring allowed by sampling uncertainty and unrestricted
item nonresponse is essentially unchanged.

This comparison is not an empirical power calculation: the Lebanon data are
held fixed and only the Monte Carlo resolution is changed. It instead assesses
the sensitivity of the final projected object to a substantial increase in
$B$, motivated by the power and concordance considerations in
\citet[Section~3]{DUFOUR2006}.

The higher-resolution calculation is computationally expensive. The original
$B=9999$ projection used 13 parallel workers and required approximately
36.6 hours of wall-clock time; the targeted lower-endpoint refinement required
a further 45.1 hours using the same number of workers. Each of the 39 restarts
allowed up to 36,000 objective evaluations. Although the refined search
satisfied the stability diagnostic, it remains a finite-run approximation
rather than a certificate of global optimality.

The endpoint-conditioned benchmark analysis is considerably more costly,
because lower and upper bounds are projected for each of 16 cells at both
discrepancy endpoints. I therefore retain $B=999$ for the benchmark projections
in the main text. This choice preserves exact finite-sample calibration at the
5\% level and keeps the aggregate and cellwise analyses on a common Monte
Carlo specification. The $B=9999$ exercise shows that the primary scalar
discrepancy is highly insensitive to a tenfold increase in Monte Carlo
resolution, but it does not establish the same insensitivity for every
endpoint-conditioned cell bound.

\section{Further Discussion}
\label{appendix:further-discussion}
\subsection{Measurement Geometry and Maximal Reassignment Benchmarks}
\label{appendix:maximal-reassignment}

The conservative transition benchmark studied in the main text minimizes
aggregate reassignment under the ordinal geometry selected in
Section~\ref{subsec:ordinal-geometry}. To distinguish the measurement choice
from the subsequent optimization over couplings, consider the general threshold-weighted geometry on $\Xi=\{\xi_1\prec\cdots\prec\Xi_K\}$
with
\[
d_w(\xi_i,\xi_j)
=
\sum_{k=\min\{i,j\}}^{\max\{i,j\}-1}w_k,
\quad\text{where}\quad i\neq j,
\]
with $w_k>0$, and define for $\pi\in\Pi(\mu,\nu)$
\(
C_w(\pi)
=
\sum_{i=1}^K\sum_{j=1}^K d_w(\xi_i,\xi_j)\pi_{ij}.
\)
Section~\ref{subsec:ordinal-geometry} shows that the corresponding probability
metric is
\(
D_w(\mu,\nu)
=
\sum_{k=1}^{K-1}w_k|F_\mu(k)-F_\nu(k)|
=
\min_{\pi\in\Pi(\mu,\nu)} C_w(\pi).
\)
Under threshold neutrality and unit normalization, $d_w(\xi_i,\xi_j)=|i-j|$ and
$D_w$ reduces to the canonical discrepancy $D$.

Once the ordinal geometry is fixed, one may also consider the opposite
extremum,
\(
M_w(\mu,\nu)
:=
\max_{\pi\in\Pi(\mu,\nu)} C_w(\pi),
\)
which is attained because $\Pi(\mu,\nu)$ is compact and $C_w$ is linear.
Moreover,
\(
\mathcal C_w(\mu,\nu)
:=
\{C_w(\pi):\pi\in\Pi(\mu,\nu)\}
=
[D_w(\mu,\nu),M_w(\mu,\nu)],
\)
since $\Pi(\mu,\nu)$ is convex and $C_w$ is linear. Hence, when only the two
marginals are observed, this interval is the sharp identified set for the
aggregate ordinal reassignment cost of a latent transition structure. If
$\pi^0\in\Pi(\mu,\nu)$ is the actual population transition structure, then
\(
D_w(\mu,\nu)
\leq
C_w(\pi^0)
\leq
M_w(\mu,\nu),
\)
and every intermediate value is attainable. Its width therefore measures the
ambiguity about aggregate ordinal movement left unresolved by the marginals.

The two endpoints have different conceptual status. The lower endpoint is the
probability metric characterized by the measurement principles in
Section~\ref{subsec:ordinal-geometry}; its optimal-transport representation is
a consequence of that characterization. The upper endpoint is instead an
extremal functional of the admissible couplings after the ordinal cost has
been fixed. In particular, $M_w$ is not a probability metric: generally
$M_w(\mu,\mu)>0$ even though $D_w(\mu,\mu)=0$. Thus $D_w$ measures the
restructuring required by the marginal change, whereas $M_w$ measures the
greatest restructuring the same marginals can accommodate. The latter is
informative about latent movement that repeated cross-sections cannot rule
out, but is not an alternative measure of ordinal distributional change.

The interval $[D_w(\mu,\nu),M_w(\mu,\nu)]$ concerns aggregate reassignment
cost, not individual transition probabilities. Minimizing and maximizing
plans therefore do not generally provide cellwise bounds on $\pi_{ij}$;
sharp bounds based only on the marginals are instead the classical Fr\'echet
bounds discussed in Appendix~\ref{appendix:frechet}.

Let
\(
\Pi_w^{\max}(\mu,\nu)
=
\arg\max_{\pi\in\Pi(\mu,\nu)} C_w(\pi).
\)
As with the conservative benchmark, the maximizer need not be unique. For the
canonical geometry and the four-category example in
Section~\ref{subsec:illustrative-example},
\(
\mu=(0.4,0.3,0.2,0.1),
\;
\nu=(0.2,0.3,0.3,0.2),
\)
the conservative benchmark is $D(\mu,\nu)=0.5$, whereas the maximal
reassignment benchmark is $M(\mu,\nu)=2.4$.
Figure~\ref{fig:maximal_mobility_example} displays one representative maximal
coupling. Whereas conservative benchmark plans minimize aggregate threshold
crossings, a maximizing coupling places mass across more distant categories
subject to preserving the same marginals.

\begin{figure}[htbp]
\centering
\begin{tikzpicture}
\begin{axis}[
    width=0.34\textwidth,
    height=0.34\textwidth,
    colormap={whitered}{
        color(0cm)=(white);
        color(1cm)=(red)
    },
    point meta min=0,
    point meta max=0.2,
    xmin=0.5, xmax=4.5,
    ymin=0.5, ymax=4.5,
    xtick={1,2,3,4},
    ytick={1,2,3,4},
    xlabel={Destination category $j$},
    ylabel={Origin category $i$},
    enlargelimits=false,
    axis on top,
    y dir=reverse,
    nodes near coords,
    nodes near coords style={font=\tiny, text=black},
    every node near coord/.append style={
        /pgf/number format/fixed,
        /pgf/number format/precision=1
    },
    colorbar,
    colorbar style={
        at={(1.12,0.5)},
        anchor=west,
        height=0.22\textwidth,
        width=0.10cm,
        ytick={0,0.05,0.10,0.15,0.20},
        scaled ticks=false,
        tick label style={font=\tiny}
    },
    title={Maximal reassignment coupling $\pi^{\max}$}
]
\addplot[
    matrix plot*,
    mesh/rows=4,
    mesh/cols=4,
    point meta=explicit,
] coordinates {
    (1,1) [0.0] (2,1) [0.0] (3,1) [0.2] (4,1) [0.2]
    (1,2) [0.0] (2,2) [0.0] (3,2) [0.1] (4,2) [0.2]
    (1,3) [0.1] (2,3) [0.2] (3,3) [0.0] (4,3) [0.0]
    (1,4) [0.1] (2,4) [0.1] (3,4) [0.0] (4,4) [0.0]
};
\end{axis}
\end{tikzpicture}
\caption{Representative maximal reassignment coupling for the illustrative
example.}
\label{fig:maximal_mobility_example}
\end{figure}

Under partial observation, one could similarly optimize $M_w(\mu,\nu)$ and
its maximizing plans over the identified latent marginals and apply the
projection-based inference machinery of Section~\ref{Section Inference}. I do not pursue this
extension because the maximal benchmark concerns movement that the marginals
fail to rule out, rather than movement the observed marginal change makes
unavoidable.

\subsection{Panel-Data Interpretation}

Although the framework is motivated by settings in which only repeated cross-sections are observed, the
minimum- and maximum-reassignment benchmarks also provide useful reference points when population
transition data are available directly. Let $\pi^0\in\Pi(\mu_t,\mu_{t+1})$ denote the realized population
transition matrix between two periods and define its aggregate ordinal reassignment cost by
\(
C_w^0
=
\sum_{i=1}^K\sum_{j=1}^K
d_w(\xi_i,\xi_j)\pi^0_{ij}.
\)
Since $D_w(\mu_t,\mu_{t+1})$ and $M_w(\mu_t,\mu_{t+1})$ are respectively the minimum and maximum of this
cost over $\Pi(\mu_t,\mu_{t+1})$, it follows immediately that
\(
D_w(\mu_t,\mu_{t+1})
\leq
C_w^0
\leq
M_w(\mu_t,\mu_{t+1}).
\)

The lower endpoint therefore measures the aggregate restructuring required by the marginal evolution,
whereas the upper endpoint measures the greatest aggregate restructuring compatible with the same
marginals. The gap
\(
C_w^0-D_w(\mu_t,\mu_{t+1})
\)
measures how much realized movement exceeds the minimum required restructuring. Such excess movement
may arise, for example, from offsetting transitions that are invisible in the comparison of the marginals.
At the same time,
\(
M_w(\mu_t,\mu_{t+1})-D_w(\mu_t,\mu_{t+1})
\)
measures the full range of aggregate reassignment costs left unresolved by the marginal information alone.

Whenever $M_w(\mu_t,\mu_{t+1})>D_w(\mu_t,\mu_{t+1})$, the realized transition structure can therefore be
located within this range using the normalized diagnostic
\[
S_w=
\frac{C_w^0-D_w(\mu_t,\mu_{t+1})}
{M_w(\mu_t,\mu_{t+1})-D_w(\mu_t,\mu_{t+1})}.
\]
By construction, $S_w\in[0,1]$. Values close to zero indicate that realized movement is close to the
least-displacement benchmark implied by the marginal evolution, whereas values close to one indicate
movement approaching the greatest aggregate reassignment compatible with those marginals. If
$M_w(\mu_t,\mu_{t+1})=D_w(\mu_t,\mu_{t+1})$, the aggregate reassignment cost is already determined by the
marginals and no normalization is needed.

The conservative benchmark contains additional information that is not summarized by the scalar
quantity $D_w(\mu_t,\mu_{t+1})$. Define the least-displacement benchmark set
\[
\Pi_w^\star(\mu_t,\mu_{t+1})
=
\arg\min_{\pi\in\Pi(\mu_t,\mu_{t+1})}
\sum_{i=1}^K\sum_{j=1}^K d_w(\xi_i,\xi_j)\pi_{ij}.
\]
This set is not an estimator of the realized transition matrix $\pi^0$. Rather, it is a nonstructural
extremal reference set containing all transition patterns that reconcile the two marginals using only
the ordinal movement required by their change. It therefore provides a benchmark for studying not only
the amount of realized movement, but also its composition across category pairs.

For example, for a metric $\rho$ on the space of transition matrices, consider
\[
R_w^0
=
\operatorname{dist}\!\left(
\pi^0,\Pi_w^\star(\mu_t,\mu_{t+1})
\right)
:=
\inf_{\pi^\star\in\Pi_w^\star(\mu_t,\mu_{t+1})}
\rho(\pi^0,\pi^\star).
\]
Whereas $C_w^0-D_w(\mu_t,\mu_{t+1})$ measures realized movement in excess of the minimum aggregate
restructuring, $R_w^0$ measures how far the pattern of realized transitions lies from every
least-displacement transition pattern. The two diagnostics therefore capture different features of
mobility. Realized transition structures can have similar aggregate reassignment costs while allocating
that movement very differently across category pairs. Retaining the entire optimizer set is important
for this comparison, since measuring distance from an arbitrarily selected optimal plan would make the
diagnostic depend on an arbitrary choice among potentially nonunique conservative benchmarks.

The same reference-set interpretation can be useful when transition matrices are constructed rather
than observed directly. Suppose, for example, that a synthetic-panel method produces a transition
matrix $\widehat{\pi}^{\,SP}\in\Pi(\mu_t,\mu_{t+1})$. Then
\[
\sum_{i=1}^K\sum_{j=1}^K d_w(\xi_i,\xi_j)\widehat{\pi}^{\,SP}_{ij}
-
D_w(\mu_t,\mu_{t+1})
\]
measures the additional aggregate ordinal movement implied by the synthetic-panel reconstruction,
while
\(
\operatorname{dist}\!\left(
\widehat{\pi}^{\,SP},
\Pi^\star(\mu_t,\mu_{t+1})
\right)
\)
measures how far its transition pattern lies from the entire minimum-restructuring benchmark set.
Comparing these quantities across synthetic-panel procedures provides a diagnostic for the additional
economic dynamics introduced by their identifying assumptions, beyond the movement required by the
observed marginal evolution. Such comparisons are descriptive diagnostics rather than criteria for
ranking synthetic-panel methods: additional or offsetting mobility may be a substantive implication of
the assumptions used by a particular procedure rather than evidence against it.

When the marginal distributions are themselves only partially identified, the same idea applies to the
endpoint-conditioned conservative benchmark sets developed in Section~\ref{Section PI}.
Distances from an observed or reconstructed transition matrix to the lower- and upper-endpoint
benchmark families quantify departures from conservative transition patterns associated with,
respectively, the smallest and largest minimum restructuring compatible with the observable law.
Thus the endpoint-conditioned couplings provide not only identified sets for benchmark transition
probabilities, but also extremal reference sets against which richer transition information can be
assessed.

This interpretation does not change the main purpose of the framework. When panel transitions are
unavailable, the conservative benchmark characterizes what restructuring the marginal evolution itself
requires, while the maximal benchmark describes the additional aggregate movement that the marginals
cannot rule out. When panel transitions or model-based transition reconstructions are available, the
same conservative benchmark can instead be used as a reference point for separating the dynamics
required by the marginal evolution from additional transition structure. In this sense, the ordinal
geometry provides a common link between marginal distributional change, realized transition behaviour,
and model-based reconstructions of economic dynamics.
\subsection{Relation to Fr\'echet Classes and Extremal Dependence}
\label{appendix:frechet}

For fixed marginals $\mu$ and $\nu$, the set
\(
\Pi(\mu,\nu)
\)
is the Fr\'echet class of all joint distributions having those marginals
\citep{Frechet1935,Frechet1951}. It is determined entirely by marginal
information and is therefore independent of any ordering or metric imposed on
the outcome categories.

Classical Fr\'echet inequalities give the sharp marginal-information bounds
on an individual cell probability:
\(
\max\{0,\mu_i+\nu_j-1\}
\leq
\pi_{ij}
\leq
\min\{\mu_i,\nu_j\},
\;
i,j=1,\ldots,K.
\)
These extrema are cellwise and need not be attained simultaneously by a
single coupling.

The transport benchmarks begin from the same Fr\'echet class but add a
criterion for comparing its elements. Under the threshold-additive ordinal
geometry $d_w$ of Section~\ref{subsec:ordinal-geometry}, let
\(
C_w(\pi)
=
\sum_{i=1}^K\sum_{j=1}^K
d_w(\xi_i,\xi_j)\pi_{ij},
\)
and define
\[
\Pi_w^{\min}(\mu,\nu)
=
\arg\min_{\pi\in\Pi(\mu,\nu)} C_w(\pi),
\qquad
\Pi_w^{\max}(\mu,\nu)
=
\arg\max_{\pi\in\Pi(\mu,\nu)} C_w(\pi).
\]
Thus the marginals determine the feasible set, while the ordinal geometry
selects its least- and greatest-reassignment faces.

There is an additional simplification for the least-reassignment face. For
each threshold $k=1,\ldots,K-1$, define the total mass crossing that threshold
from below to above and from above to below by
\(
A_k(\pi)
=
\sum_{i\leq k<j}\pi_{ij},
\;
B_k(\pi)
=
\sum_{j\leq k<i}\pi_{ij}.
\)
The marginal restrictions imply
\(
A_k(\pi)-B_k(\pi)
=
F_\mu(k)-F_\nu(k),
\)
while threshold additivity gives
\(
C_w(\pi)
=
\sum_{k=1}^{K-1}
w_k\{A_k(\pi)+B_k(\pi)\}.
\)
Hence
\(
C_w(\pi)
\geq
\sum_{k=1}^{K-1}
w_k|F_\mu(k)-F_\nu(k)|
=
D_w(\mu,\nu).
\)
Because $w_k>0$ for every $k$, equality holds if and only if
\[
A_k(\pi)+B_k(\pi)
=
|A_k(\pi)-B_k(\pi)|,
\qquad k=1,\ldots,K-1,
\]
or equivalently,
\[
A_k(\pi)B_k(\pi)=0,
\qquad k=1,\ldots,K-1.
\]
Thus, at every ordinal threshold, a least-displacement coupling permits
transport in at most one direction. Importantly, this characterization does
not depend on the magnitudes of the positive threshold weights. Consequently,
for fixed marginals, the set $\Pi_w^{\min}(\mu,\nu)$ is invariant to the
particular choice of positive threshold weights, even though the minimum
value $D_w(\mu,\nu)$ generally changes with the geometry. By contrast, the
greatest-reassignment face $\Pi_w^{\max}(\mu,\nu)$ can depend on the relative
threshold weights.

The cell bounds studied in the main text can therefore be viewed as
Fr\'echet bounds sharpened by the least-reassignment criterion. With known
marginals, define
\[
\underline\pi^{\min}_{ij}
=
\min_{\pi\in\Pi_w^{\min}(\mu,\nu)}\pi_{ij},
\qquad
\overline\pi^{\min}_{ij}
=
\max_{\pi\in\Pi_w^{\min}(\mu,\nu)}\pi_{ij}.
\]
Since
\(
\Pi_w^{\min}(\mu,\nu)\subseteq\Pi(\mu,\nu),
\)
it follows that
\(
\max\{0,\mu_i+\nu_j-1\}
\leq
\underline\pi^{\min}_{ij}
\leq
\overline\pi^{\min}_{ij}
\leq
\min\{\mu_i,\nu_j\}.
\)
The same inclusion holds for cell bounds calculated over
$\Pi_w^{\max}(\mu,\nu)$.

The distinction is therefore one of conditioning. Classical Fr\'echet bounds
ask how large or small a joint cell probability can be given the marginals
alone. Conservative transport bounds ask the same question among couplings
that achieve the least aggregate ordinal reassignment, while maximal
transport bounds condition instead on greatest aggregate reassignment.
Because these are global restrictions on the joint distribution, their
cellwise extrema need not coincide with the unrestricted Fr\'echet bounds.

Under partial observation, the same distinction remains, but the role of the
ordinal geometry becomes richer. Missingness allows the marginals to vary
over $\mathcal M(\mathbf q_X)$ and $\mathcal M(\mathbf q_Y)$. Unrestricted
Fr\'echet-type cell bounds would optimize over all admissible marginals and
all couplings between them. The endpoint-conditioned benchmark bounds in the
main text instead restrict attention to marginals attaining a specified
weighted discrepancy endpoint and to transport plans attaining the
corresponding minimum cost. Although the least-displacement coupling set is
invariant to the positive threshold weights for any fixed pair of marginals,
the weights generally affect the discrepancy values and may affect which
admissible marginal pairs attain the lower or upper identification endpoint.
The endpoint-conditioned benchmark sets can therefore depend on the chosen
ordinal geometry through this marginal-selection step.

Fr\'echet bounds and transport benchmarks are thus complementary rather than
competing notions of extremality. Fr\'echet theory describes the dependence
structures permitted by the marginals. The conservative transport criterion
selects, among those structures, the couplings that avoid offsetting movement
across every ordinal threshold, while the maximal criterion selects
greatest-reassignment structures relative to the specified ordinal geometry.
Under partial observation, the same geometry additionally determines which
latent marginal configurations are associated with the discrepancy
endpoints.

\subsection{Binary Outcomes and Poverty Transitions}

The binary case provides a particularly transparent specialization of the
Fr\'echet-class interpretation above. Suppose the two ordered categories are
\(
\text{poverty}\prec\text{nonpoverty},
\)
and write the marginal distributions as
\(
\mu=(p,1-p),
\;
\nu=(q,1-q),
\)
where $p$ and $q$ denote the poverty rates in the two periods. Every coupling
in $\Pi(\mu,\nu)$ can be written as
\[
\pi(x)
=
\begin{pmatrix}
x & p-x\\
q-x & 1-p-q+x
\end{pmatrix},
\qquad
\max\{0,p+q-1\}
\leq
x
\leq
\min\{p,q\}.
\]
Thus the Fr\'echet class is a line segment indexed by the single transition
probability $x$. Its endpoints are the two Fr\'echet--Hoeffding extremal
couplings.

There is only one ordinal threshold, with cost
\(
w_1=d_w(\xi_1,\xi_2)>0.
\)
Hence the aggregate reassignment cost of $\pi(x)$ is
\[
C_w(\pi(x))
=
w_1\{(p-x)+(q-x)\}
=
w_1(p+q-2x).
\]
Because $w_1>0$, this cost is minimized by the largest feasible value of
$x$, namely $x=\min\{p,q\}$, and maximized by the smallest feasible value,
namely $x=\max\{0,p+q-1\}$. Consequently, the upper Fr\'echet endpoint is
the unique conservative transport plan, while the lower Fr\'echet endpoint
is the unique maximal-reassignment plan. The corresponding aggregate
benchmarks are
\(
D_w(\mu,\nu)
=
w_1|p-q|,
\)
and
\(
M_w(\mu,\nu)
=
w_1\bigl\{1-|p+q-1|\bigr\}.
\)

The binary case therefore makes especially transparent the distinction
between the ordinal geometry and the selected coupling. Changing the
positive threshold cost $w_1$ rescales both aggregate benchmarks but does
not change either extremal transition plan. This is the binary counterpart
of the fixed-marginal invariance of the least-reassignment face discussed
in Appendix~\ref{appendix:frechet}. Under threshold neutrality and unit
normalization, $w_1=1$, giving the canonical formulas
\(
D(\mu,\nu)=|p-q|,
\;
M(\mu,\nu)=1-|p+q-1|.
\)

The poverty-transition interpretation is immediate. If poverty falls, so
that $p>q$, the conservative plan assigns mass $p-q$ to exits from poverty
and zero mass to entries into poverty. If poverty rises, so that $q>p$, the
analogous statement holds with the direction reversed. Thus the conservative
benchmark attributes only the net change in the poverty rate to transitions
across the poverty threshold and contains no offsetting entry and exit flows.
The maximal plan instead assigns as much mass as possible to movements
between poverty and nonpoverty while preserving the same marginal poverty
rates.

The binary case therefore gives an unusually sharp link between repeated
cross-sectional marginals and admissible transition structures. The full
Fr\'echet class is the line segment joining the conservative and maximal
benchmark plans, and every feasible transition matrix lies between these
two extremes. Because there is only one threshold, the choice of its
positive weight affects the scale of aggregate ordinal reassignment but
not the ordering of feasible couplings by reassignment cost.
\subsection{Dependence, Gini Mean Differences, and Maximal Reassignment}
\label{appendix:gini-dependence}

As noted in Remark~\ref{Remark: Gini}, the Kantorovich metric with
absolute-distance cost has historically also been called Gini's index of
dissimilarity. A distinct connection with Gini arises by evaluating the same
absolute-difference loss under other couplings in the Fr\'echet class.

Let $Q_\mu$ and $Q_\nu$ denote the quantile functions of the canonical rank
distributions and let $U,V\sim\operatorname{Unif}(0,1)$ be independent. The
one-dimensional transport representation gives
\(
D(\mu,\nu)
=
E|Q_\mu(U)-Q_\nu(U)|,
\)
where $(Q_\mu(U),Q_\nu(U))$ is the comonotone coupling and minimizes expected
absolute rank displacement. Under independence, define
\(
A(\mu,\nu)
=
E|Q_\mu(U)-Q_\nu(V)|.
\)
Thus $D$ and $A$ evaluate the same ordinal displacement loss under different
dependence structures.

When $\mu=\nu$, $A(\mu,\mu)$ is the Gini mean difference of the canonical
rank distribution:
\(
A(\mu,\mu)
=
E|X-X'|
=
2\sum_{k=1}^{K-1}F_\mu(k)\{1-F_\mu(k)\},
\)
for independent $X,X'\sim\mu$. More generally,
\(
A(\mu,\nu)
=
\sum_{k=1}^{K-1}
\left[
F_\mu(k)\{1-F_\nu(k)\}
+
F_\nu(k)\{1-F_\mu(k)\}
\right],
\)
to be compared with
\(
D(\mu,\nu)
=
\sum_{k=1}^{K-1}|F_\mu(k)-F_\nu(k)|.
\)

At the opposite Fr\'echet extreme, the countermonotone coupling
$(Q_\mu(U),Q_\nu(1-U))$ maximizes expected absolute rank displacement. Hence,
with the maximal-reassignment functional of
Appendix~\ref{appendix:maximal-reassignment},
\(
M(\mu,\nu)
=
E|Q_\mu(U)-Q_\nu(1-U)|,
\;
D(\mu,\nu)\leq A(\mu,\nu)\leq M(\mu,\nu).
\)
The three quantities therefore evaluate the same ordinal displacement loss
under comonotone, independent, and countermonotone dependence, respectively.

The distinction is clearest when $\mu=\nu$. Then $D(\mu,\mu)=0$, since no
marginal change has occurred, while $A(\mu,\mu)$ measures dispersion under
independent pairing and $M(\mu,\mu)$ the greatest aggregate displacement
compatible with the common marginal. Hence the Gini-type quantity is not a
measure of distributional change in the sense characterized in
Section~\ref{subsec:ordinal-geometry}, but an intermediate coupling-specific
functional within the Fr\'echet class.

All of these quantities refer to the canonical rank distribution. They do
not impose cardinal values on the substantive categories beyond the
threshold-neutral rank geometry, and $A(\mu,\mu)$ should not be confused with
the normalized Gini coefficient for cardinal outcomes.

\subsection{Higher-Order Quantile Discrepancies}

The one-dimensional representation of the threshold-additive ordinal
geometry suggests a broader class of distributional discrepancies. Recall
that
\(
r_w:\Xi\rightarrow\mathbb R,
\;
r_w(\xi_i)=a_i,
\)
is the isometric representation of $(\Xi,d_w)$, where
\[
a_1=0,
\qquad
a_i=\sum_{\ell=1}^{i-1}w_\ell,
\quad i=2,\ldots,K.
\]
Let
\(
\mu_w:=(r_w)_\#\mu,
\;
\nu_w:=(r_w)_\#\nu
\)
denote the corresponding probability distributions on the metric
coordinates $a_1<\cdots<a_K$, and let $Q_{\mu,w}$ and $Q_{\nu,w}$
denote their quantile functions. For $p\geq1$, define
\[
D_{w,p}(\mu,\nu)
:=
\left\{
\int_0^1
\left|
Q_{\mu,w}(u)-Q_{\nu,w}(u)
\right|^p
\,du
\right\}^{1/p}.
\]

Because $r_w$ is an isometry, the one-dimensional optimal-transport
representation gives
\[
D_{w,p}(\mu,\nu)^p
=
\min_{\pi\in\Pi(\mu,\nu)}
\sum_{i=1}^K\sum_{j=1}^K
d_w(\xi_i,\xi_j)^p\pi_{ij}.
\]
Thus $D_{w,p}$ is the Wasserstein--$p$ distance associated with the
threshold-additive ground metric $d_w$. It is invariant to
order-preserving relabellings that preserve the specified ordinal geometry
and, conditional on that geometry, depends only on the two marginal
distributions. Under threshold neutrality and unit normalization,
$d_w(\xi_i,\xi_j)=|i-j|$, so this construction reduces to the
Wasserstein--$p$ discrepancy between the canonical rank distributions.

The case $p=1$ is nevertheless distinctive. In that case,
\[
D_{w,1}(\mu,\nu)
=
D_w(\mu,\nu)
=
\sum_{k=1}^{K-1}
w_k
\left|
F_\mu(k)-F_\nu(k)
\right|,
\]
so displacement decomposes additively across the successive ordinal
thresholds. Equivalently, crossing several thresholds has exactly the sum
of the costs of crossing them individually. This is the structure selected
by the threshold-separability and local ground-metric consistency
principles of Section~\ref{subsec:ordinal-geometry}.

For $p>1$, the transport criterion instead uses the convex displacement
cost
\(
d_w(\xi_i,\xi_j)^p.
\)
A movement spanning several thresholds therefore receives more than the
sum of the corresponding one-threshold $p$th-power costs. For example, if
$i<j<\ell$, then
\[
d_w(\xi_i,\xi_\ell)^p
=
\left\{
d_w(\xi_i,\xi_j)+d_w(\xi_j,\xi_\ell)
\right\}^p
>
d_w(\xi_i,\xi_j)^p
+
d_w(\xi_j,\xi_\ell)^p
\]
whenever $p>1$. Hence higher-order Wasserstein discrepancies place
disproportionately greater weight on long ordinal movements and no longer
admit the additive threshold-crossing interpretation that characterizes
$D_w$.

The distinction is therefore substantive rather than merely a choice of
metric order. The $p=1$ discrepancy is appropriate when ordinal
displacement is valued linearly in cumulative threshold-crossing cost.
Higher values of $p$ may instead be useful in applications in which
movements spanning several ordered categories are regarded as
disproportionately consequential relative to the same amount of
threshold movement distributed across shorter reallocations. Such a
criterion represents a different measurement principle and would require
an axiomatic foundation different from the threshold-separability
argument used for $D_w$. Developing such a characterization is left for
future work.

\end{document}